\documentclass[prodmode,license]{ec-acmsmall}

\usepackage[latin9]{inputenc}
\usepackage{verbatim}
\usepackage{amsthm}
\usepackage{amsmath}
\usepackage{amssymb}
\usepackage{graphicx}
\usepackage[numbers]{natbib}

\makeatletter
\theoremstyle{plain}
\newtheorem{thm}{\protect\theoremname}
  \theoremstyle{definition}
  \newtheorem{defn}[thm]{\protect\definitionname}
  \theoremstyle{plain}
  \newtheorem{cor}[thm]{\protect\corollaryname}
  \theoremstyle{plain}
  \newtheorem{lem}[thm]{\protect\lemmaname}
  \theoremstyle{definition}
  \newtheorem{example}[thm]{\protect\examplename}
  \theoremstyle{plain}
  \newtheorem{prop}[thm]{\protect\propositionname}

\doi{http://dx.doi.org/10.1145/2600057.2602862}
\usepackage{multibib}
\newcites{app}{appendix_ref}
\acmVolume{X}
\acmNumber{X}
\acmArticle{X}
\acmYear{2014}
\acmMonth{2}
\usepackage[numbers]{natbib}

\makeatother

  \providecommand{\corollaryname}{Corollary}
  \providecommand{\definitionname}{Definition}
  \providecommand{\examplename}{Example}
  \providecommand{\lemmaname}{Lemma}
  \providecommand{\propositionname}{Proposition}
\providecommand{\theoremname}{Theorem}

\begin{document}
\markboth{E.A. Meirom et al.}{Network Formation Games with Heterogeneous Players and the Internet Structure}

\title{Network Formation Games with Heterogeneous Players and the Internet
Structure}
\begin{abstract}
We study the structure and evolution of the Internet's Autonomous
System (AS) interconnection topology as a game with heterogeneous
players. In this network formation game, the utility of a player depends
on the network structure, e.g., the distances between nodes and the
cost of links. We analyze static properties of the game, such as the
prices of anarchy and stability and provide explicit results concerning
the generated topologies. Furthermore, we discuss dynamic aspects,
demonstrating linear convergence rate and showing that only a restricted
subset of equilibria is feasible under realistic dynamics. We also
consider the case where utility (or monetary) transfers are allowed
between the players. 
\end{abstract}

\author{Eli A. Meirom, Shie Mannor, Ariel Orda \affil{Tecnion - Israel Institute of Technology}}
\maketitle

\category{C.2.2}{Computer-Communication Networks}{Network Protocols}
\terms{Design, Algorithms, Performance}
\keywords{Wireless sensor networks, media access control, multi-channel, radio interference, time synchronization}
\acmformat{Eli .A. Meirom, Shie Mannor, Ariel Orda, 2014. Network Formation Games and the Internet Structure.}
\begin{bottomstuff}
This research was supported in part by the European Union through the CONGAS project (http://www.congas-project.eu/) in the 7th Framework Programme.
\end{bottomstuff}

\section{Introduction}

The structure of large-scale communication networks, in particular
that of the Internet, has carried much interest. The Internet is a
living example of a large, self organized, many-body complex system.
Understanding the processes that shape its topology would provide
tools for engineering its future.

The Internet is composed of multiple Autonomous Systems (ASs), which
are contracted by different economic agreements. These agreements
dictate the routing pathways among the ASs. With some simplifications,
we can represent the resulting network as a graph, where two nodes
(ASs) are connected by a link if traffic is allowed to traverse through
them. The statistical properties of this complex graph, such as the
degree distribution, clustering properties etc., have been extensively
investigated \citep{Gregori2013,Vazquez2002,Siganos2003}. However,
such findings per se lack the ability to either predict the future
evolution of the Internet nor to provide tools for shaping its development.

Most models, notably ``preferential attachment'' \citep{Barabasi1999},
emulate the network evolution by probabilistic rules and recover some
of the statistical aspects of the network. However, they fail to account
for many other features of the network \citep{1019306}, as they treat
the ASs as passive elements rather than economic, profit-maximizing
entities. Indeed, in this work we examine some findings that seem
to contradict the predictions of such models but are explained by
our model.

Game theory is by now a widely used doctrine in network theory and
computer science in general. It describes the behavior of interacting
rational agents and the resulting equilibria, and it is one of the
main tools of the trade in estimating the performance of distributed
algorithms \citep{Borkar2007}. In the context of communication networks,
game theory has been applied extensively to fundamental control tasks
such as bandwidth allocation \citep{Lazar1997}, routing \citep{Altman2000,Orda1993}
and flow control \citep{Altman1994}. It was also proven fruitful
in other networking domains, such as network security \citep{Roy2010}
and wireless network design \citep{Charilas2010}. 

Recently, there has been a surge of studies exploring networks created
by rational players, where the costs and utility of each player depend
on the network structure. Some studies emphasized the context of wireless
networks \citep{5062080} whereas other discussed the inter-AS topology
\citep{Anshelevich2011,Alvarez2012}. These works fall within the
realm of \emph{network formation games} \citep{Johari2006,Jackson1996}.
The focus of those studies has been to detail some specific models,
and then investigate the equilibria's properties, e.g., establishing
their existence and obtaining bounds on the ``price of anarchy''
and ''price of stability''. The latter bound (from above and below,
correspondingly) the social cost deterioration at an equilibrium compared
with a (socially) optimal solution. Taking a different approach,\citet{Lodhi2012a}
present an analytically-intractable model , hence  simulations are
used in order to obtain statistical characteristics of the resulting
topology. 

Nonetheless, most of these studies assume that the players are identical,
whereas the Internet is composed of many types of ASs, such as minor
ISPs, CDNs, tier-1 ASs etc. Only a few studies have considered the
effects of heterogeneity on the network structure. Addressing social
networks, \citep{Vandenbossche2012} describes a network formation
game in which the link costs are heterogeneous and the benefit depends
only on a player's nearest neighbors (i.e., no spillovers). In \citep{Johari2006},
the authors discuss directed networks formation, where the information
(or utility) flows in one direction along a link; the equilibria's
existence properties of the model's extension to heterogeneous players
was discussed in \citep{Alvarez2012}.

With very few exceptions, e.g., \citep{Arcaute2013}, the vast majority
of studies on the application of game theory to networks, and network
formation games in particular, focus on static properties. However,
it is not clear that the Internet, nor the economic relations between
ASs, have reached an equilibrium. In fact, dynamic inspection of the
inter-AS network presents evidence that the system may in fact be
far from equilibrium. Indeed, new ASs emergence, other quit business
or merge with other ASs, and new contracts are signed, often employing
new business terms. Hence, a \emph{dynamic} study of inter-AS network
formation games is called for.

The aim of this study is to address the above two major challenges,
namely \emph{heterogeneity} and \emph{dynamicity}. Specifically, we
establish an analytically-tractable model, which explicitly accounts
for the heterogeneity of players. Then, we investigate both its static
properties as well as its dynamic evolution.

We model the inter-AS connectivity as a network formation game with
\emph{heterogeneous} players that may share costs by monetary transfers.
We account for the inherent bilateral nature of the agreements between
players, by noting that the establishment of a link requires the agreement
of both nodes at its ends, while removing a link can be done unilaterally.
The main contributions of our study are as follows:
\begin{itemize}
\item We evaluate static properties of the considered game, such as the
prices of anarchy and stability and characterize additional properties
of the equilibrium topologies. In particular, the optimal stable topology
and examples of worst stable topologies are expressed explicitly.
\item We discuss the dynamic evolution of the inter-AS network, calculate
convergence rates and basins of attractions for the different final
states. Our findings provide useful insight towards incentive design
schemes for achieving optimal configurations. Our model predicts
the existence of a \emph{settlement-free} clique, and that most of
the other contracts between players include monetary transfers.
\end{itemize}
Game theoretic analysis is dominantly employed as a ``toy model''
for contemplating about real-world phenomena. It is rarely confronted
with real-world data, and to the best of our knowledge, it was never
done in the context of inter-AS network formation games. In this study
we go a step further from traditional formal analysis, and we do consider
real inter-AS topology data. A preliminary data analysis, which supports
our findings, appears in the appendix.

The paper is organized as follows. In the next section, we describe
our model. We discuss two variants, corresponding to whether utility
transfers (e.g., monetary transfers) are allowed or not. We present
static results in Section 3, followed by dynamic analysis in Section
4. Section 5 addresses the case of permissible monetary transfers,
both in the static and dynamic aspects. Finally, conclusions are presented
in Section 6. Full proofs and some technical details are omitted from
this version and can be found in the appendix.

\section{Model}

Our model is inspired by the inter-AS interconnection network, which
is formed by drawing a link between every two ASs that mutually agree
to allow bidirectional communication. The utility of each AS, or player,
depends on the resulting graph's connectivity. We can imagine this
as a game in which a player's strategy is defined by the links it
would like to form, and, if permissible, the price it will be willing
to pay for each. In order to introduce heterogeneity, we consider
two types of players, namely \emph{major league} (or t\emph{ype-A})
players, and \emph{minor league} (or \emph{type-B}) players. The former
may represent main network hubs (e.g., in the context of the Internet,
Level-3 providers), while the latter may represent local ISPs.

The set of type-A (type B) players is denoted by $T_{A}$ ( $T_{B}$).
A link connecting node $i$ to node $j$ is denoted as either $(i,j)$
or $ij$. The total number of players is $N=|T_{A}|+|T_{B}|$, and
we always assume $N\geq3$. The \emph{shortest distance} between nodes
$i$ and $j$ is the minimal number of hops along a path connecting
them and is denoted by $d(i,j)$. Finally, The degree of node $i$
is denoted by $deg(i)$.

\subsection{Basic model}

The utility of every player depends on the aggregate distance from
all other players. Most of the previous studies assume that each player
has a specific traffic requirement for every other player. This results
in a huge parameter space. However, it is reasonable to assume that
an AS does not have exact flow perquisites to every individual AS
in the network, but would rather group similar ASs together according
to their importance. Hence, a player would have a strong incentive
to maintain a good, fast connection to the major information and content
hubs, but would relax its requirements for minor ASs. Accordingly,
we represent the connection quality between players as their graph
distance, since many properties depend on this distance, for example
delays and bandwidth usage. As the connection to major players is
much more important, we add a weight factor $A$ to the cost function
in the corresponding distance term. The last contribution to the cost
is the link price, $c.$ This term represents factors such as the
link\textquoteright{}s maintenance costs, bandwidth allocation costs
etc. 

The structure of our cost function extends the work of \citet{Fabrikant2003}
and \citet{Corbo2005}. This model was studied extensively, including
numerous extenstions, e.g., \citet{Demaine2007,Anshelevich2003}.
Here, we focus on the hetereogenous dynamic case.

We allow different types of players to incur different link costs,
$c_{A},c_{B}$. For example, major player have greater financial resources,
reducing the effective link cost. They have incorporate advanced infrastructure
that allows them to cope successfully with the increased traffic.
Alternatively, players may evaluate the relative player's importance,
which is expressed by the factor $A$, differently. For example, a
search engine may spend significant resources in order to maintain
a fast connection to a content provider, in order to be able to index
its content efficiently. A domestic ISP or a university hub will care
less about the connection quality. As it will turn out, the relevant
quantity is $A/c$, and therefore it is sufficient to allow a variation
in one parameter only, which for simplicity will assume it is the
link cost $c.$ Formally, the (dis-)utility of players is represented
as follows.
\begin{defn}
The cost function,$\, C(i)$, of node $i$, is defined as:\label{cost-definition}
\begin{eqnarray*}
C_{A}(i) & \triangleq & deg(i)\cdot c_{A}+A\sum_{j\in T_{A}}d(i,j)+\sum_{j\in T_{B}}d(i,j)\\
C_{B}(i) & \triangleq & deg(i)\cdot c_{B}+A\sum_{j\in T_{A}}d(i,j)+\sum_{j\in T_{B}}d(i,j)
\end{eqnarray*}

where $A>1$ represents the relative importance of class A nodes over
class B nodes.

Then, the \emph{social cost} is defined as $\mathcal{S}=\sum_{i}C(i)$

Set $c\triangleq\left(c_{A}+c_{B}\right)/2$. We assume $c_{A}\leq c_{B}.$
The optimal (minimal) social cost is denoted as $\mathcal{S}_{optimal}$. 

\begin{defn}
The change in cost of player $i$ as a result of the addition of link
$(j,k)$ is denoted by $\Delta C(i,E+jk)\triangleq C\left(i,E\cup(j,k)\right)-C\left(i,E\right)$.

We will sometimes use the abbreviation $\Delta C(i,jk)$. When $(j,k)\in E$,
we will use the common notation $\Delta C(i,E-jk)\triangleq C\left(i,E\right)-C\left(i,E\setminus(j,k)\right)$.\end{defn}

\end{defn}
Players may establish links between them if they consider this will
reduce their costs. We take into consideration the agreement's bilateral
nature, by noting that the establishment of a link requires the agreement
of both nodes at its ends, while removing a link can be done unilaterally.
This is known as a \emph{pairwise-stable }equilibrium\citep{Jackson1996,Arcaute2013}. 
\begin{defn}
The players' strategies are \emph{pairwise-stable }if for all $i,j\in T_{A}\cup T_{B}$
the following hold:

a) if $ij\in E$, then $\Delta C(i,E-ij)>0$;
\end{defn}
b) if $ij\notin E$, then either $\Delta C(i,E+ij)>0$ or $\Delta C(j,E+ij)>0$. 

The corresponding graph is referred to as a \emph{stabilizable} graph.

\subsection{Utility transfer}

In the above formulation, we have implicitly assumed that players
may not transfer utilities. However, often players are able to do
so, in particular via monetary transfers. We therefore consider also
an extended model that incorporates such possibility. Specifically,
the extended model allows for a monetary transaction in which player
$i$ pays player $j$ some amount $P_{ij}$ iff the link $(i.j)$
is established. Player $j$ sets some minimal price $w_{ij}$ and
if $P_{ij}\geq w_{ij}$ the link is formed. The corresponding change
to the cost function is as follows.
\begin{defn}
\label{monetary cost definition}The cost function of player $i$
when monetary transfers are allowed is $\tilde{C}(i)\triangleq C(i)+\sum_{j,ij\in E}\left(P_{ij}-P_{ji}\right)$.
\end{defn}
Note that the social cost remains the same as in Def. \ref{monetary cost definition}
as monetary transfers are canceled by summation.

Monetary transfers allow the sharing of costs. Without transfers,
a link will be established only if \emph{both} parties, $i$ and $j$,
reduce their costs, $C(i,E+ij)<0$ and $C(j,E+ij)<0$. Consider, for
example, a configuration where $\Delta C(i,E+ij)<0$ and $\Delta C(j,E+ij)>0$.
It may be beneficial for player $i$ to offer a lump sum $P_{ij}$
to player $j$ if the latter agrees to establish $(i,j)$. This will
be feasible only if the cost function of both players is reduced.
It immediately follows that if $\Delta C(i,E+ij)+\Delta C(i,E+ij)<0$
then there is a value $P_{ij}$ such that this condition is met. Hence,
it is beneficial for both players to establish a link between them.
In a game theoretic formalism, if the \emph{core} of the two players
game is non-empty, then they may pick a value out of this set as the
transfer amount. Likewise, if the core is empty, or $\Delta C(i,E+ij)+\Delta C(j,E+ij)>0$,
then the best response of at least one of the players is to remove
the link, and the other player has no incentive to offer a payment
high enough to change the its decision. Formally:
\begin{cor}
\textup{\emph{\emph{}%
\begin{comment}
\emph{Is this a lemma or a definition? It's a delicate point.}
\end{comment}
\emph{\label{lem:edges with monetary transfers.-1}When monetary transfers
are allowed, the link $(i,j)$ is established iff 
\[
\Delta C(i,E+ij)+\Delta C(j,E+ij)<0
\]
 The link is removed iff 
\[
\Delta C(i,E-ij)+\Delta C(j,E-ij)>0
\]
}}}

\textup{\emph{\emph{}%
\begin{comment}
\emph{Is this a lemma or a definition? It's a delicate point.}
\end{comment}
\emph{\label{lem:edges with monetary transfers.}When monetary transfers
are allowed, the link $(i,j)$ is established iff $\Delta C(i,E+ij)+\Delta C(j,E+ij)<0$.
The link is removed iff $\Delta C(i,E-ij)+\Delta C(j,E-ij)>0$.}}}
\end{cor}
In the remainder of the paper, whenever monetary transfers are feasible,
we will state it explicitly, otherwise the basic model (without transfers)
is assumed.

\section{Basic model - Static analysis}

In this section we discuss the properties of stable equilibria. Specifically,
we first establish that, under certain conditions, the major players
group together in a clique (section \ref{sub:The-type-A-clique}).
We then describe a few topological characteristics of all equilibria
(section \ref{sub:Pair-wise-equilibria}). 

As a metric for the quality of the solution we apply the commonly
used measure of the social cost, which is the sum of individual costs.
We evaluate the \emph{price of anarchy}, which is the ratio between
the social cost at the worst stable solution and its value at the
optimal solution, and the \emph{price of stability}, which is the
ratio between the social cost at the best stable solution and its
value at the optimal solution (section \ref{sub:PoA and PoS}).

\subsection{\label{sub:The-type-A-clique}The type-A clique}

Our goal is understanding the resulting topology when we assume strategic
players and myopic dynamics. Obviously, if the link's cost is extremely
low, every player would establish links with all other players. The
resulting graph will be a clique. As the link's cost increase, it
becomes worthwhile to form direct links only with major players. In
this case, only the major players' subgraph is a clique. The first
observation leads to the following result.
\begin{lem}
If $c_{B}<1$ then the only stabilizable graph is a clique.
\end{lem}
 If two nodes are at a distance $L+1$ of each other, then there
is a path with $L$ nodes connecting them. By establishing a link
with cost $c$, we are shortening the distance between the end node
to $\sim L/2$ nodes that lay on the other side of the line. The average
reduction in distance is also $\approx L/2$, so by comparing $L^{2}\approx4c$
we obtain a bound on $L$, as follows:
\begin{lem}
\label{lem:The-longest-distance}The longest distance between any
node $i$ and node $j\in T_{B}$ is bounded by $2\sqrt{c_{B}}$. The
longest distance between nodes $i,j\in T_{A}$ is bounded by $\sqrt{(1-2A)^{2}+4c_{A}}-2\left(A-1\right)$.
In addition, if $c_{A}<A$ then there is a link between every two
type-A nodes.
\end{lem}

Lemma \ref{lem:The-longest-distance} indicates that if $1<c_{A}<A$
then the type $A$ nodes will form a clique (the ``nucleolus'' of
the network). The type $B$ nodes form structures that are connected
to the type $A$ clique (the network nucleolus). These structures
are not necessarily trees and will not necessarily connect to a single
point of the type-A clique only. This is indeed a very realistic scenario,
found in many configurations. In the appendix we compare this result
to actual data on the inter-AS interconnection topology.

If $c_{A}>A$ then the type-A clique is no longer stable. This setting
does not correspond to the observed nature of the inter-AS topology
and we shall focus in all the following sections on the case $1<c_{A}<A$.
Nevertheless, in the appendix we treat the case $c_{A}>A$ explicitly.

\subsection{\label{sub:Pair-wise-equilibria}Equilibria's properties}

Here we describe common properties of all pair-wise equilibria. We
start by noting that, unlike the findings of several other studies
\citet{Arcaute2013,5173479,NisanN.RoughgardenT.TardosE.2007}, in
our model, at equilibrium, the type-B nodes are not necessarily organized
in trees. This is shown in the next example.

\begin{figure}
\centering{}\includegraphics[width=0.75\columnwidth]{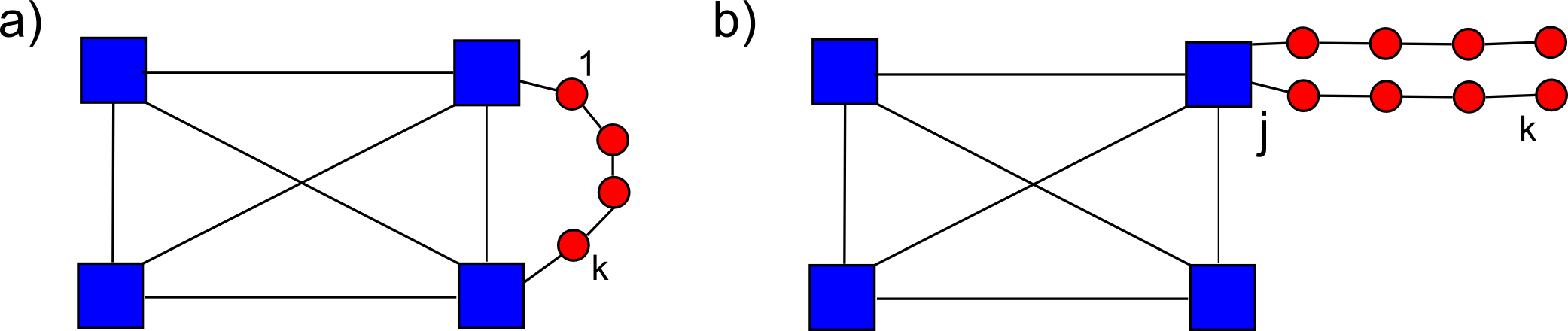}\caption{\label{fig:loop example}Non optimal networks. The type-A clique is
in blue squares, the type-B players are in red circles. a) The network
described in Example \ref{example}. b) \label{fig:A-poor-equilibrium}A
poor equilibrium, as described in the appendix.}
\end{figure}

\begin{example}
\label{example}Assume for simplicity that $c_{A}=c_{B}=c$. Consider
a line of length $k$ of type B nodes, $(1,2,3...,k)$ such that $\sqrt{8c}>k+1>\sqrt{2c}$
or equivalently $\left(k+1\right)^{2}<8c<4\left(k+1\right)^{2}$ .
In addition, the links $(j_{1},1)$ and $(j_{2},k)$ exist, where
$j\in T_{A}$, i.e., the line is connected at both ends to different
nodes of the type-A clique, as depicted in Fig \ref{fig:loop example}.
In \citep{Meirom2013} we show that this is a stabilizable graph.

\end{example}
A stable network cannot have two ``heavy'' trees, ``heavy'' here
means that there is a deep sub-tree with many nodes, as it would be
beneficial to make a shortcut between the two sub-trees(details appear
in the appendix). In other words, trees must be shallow and small.
This means that, while there are many equilibria, in all of them nodes
cannot be too far apart, i.e., a small-world property. Furthermore,
the trees formed are shallow and are not composed of many nodes.

\subsection{\label{sub:PoA and PoS}Price of Anarchy \& Price of Stability }

As there are many possible link-stable equilibria, a discussion of
the price of anarchy is in place. First, we explicitly find the optimal
configuration. Although we establish a general expression for this
configuration, it is worthy to also consider the limiting case of
a large network, $|T_{B}|\gg1,|T_{A}|\gg1$. Moreover, typically,
the number of major league players is much smaller than the other
players, hence we also consider the limit $|T_{B}|\gg|T_{A}|\gg1$.

\begin{figure}
\centering{}\includegraphics[width=0.7\columnwidth]{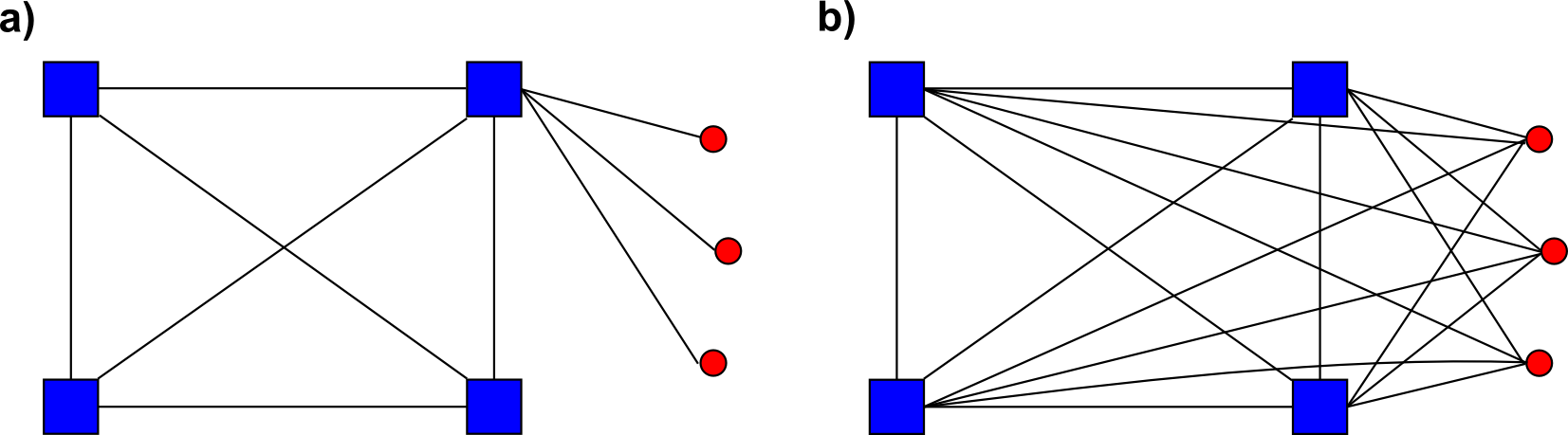}\caption{\label{fig:The-optimal-solution} \label{fig:The-optimal-state - monetary transfers}The
optimal solution, as described in Lemma \ref{lem:optimal solution}.
If $\left(A+1\right)/2<c$ the optimal solution is described by a),
otherwise by b). When monetary transfers (section \ref{sec:Monetary-transfers})
are allowed, both configurations are stabilizable. Otherwise, only
a) is stabilizable.}
\end{figure}

\begin{prop}
\label{lem:optimal solution}Consider the network where the type $B$
nodes are connected to a specific node $j\in T_{A}$ of the type-A
clique. The social cost in this stabilizable network (Fig. \ref{fig:The-optimal-solution}(a))
is 
\begin{eqnarray*}
\mathcal{S} & = & 2|T_{B}|\left(|T_{B}|-1+c+\left(A+1\right)(|T_{A}|-1/2)\right)+|T_{A}|\left(|T_{A}|-1\right)\left(c_{A}+A\right).
\end{eqnarray*}

Furthermore, if $|T_{B}|\gg1,|T_{A}|\gg1$ then, omitting linear terms
in $|T_{B}|,|T_{A}|$, 
\[
\mathcal{S}=2|T_{B}|(|T_{B}|+\left(A+1\right)|T_{A}|)+|T_{A}|^{2}\left(c+A\right).
\]
Moreover, if $\frac{A+1}{2}\le c$ then this network structure is
socially optimal and the price of stability is $1$, otherwise the
price of stability is
\[
PoS=\frac{2|T_{B}|(|T_{B}|+\left(A+1\right)|T_{A}|)+|T_{A}|^{2}\left(c_{A}+A\right)}{2|T_{B}|\left(|T_{B}|+\left(\frac{A+1}{2}+c\right)|T_{A}|\right)+|T_{A}|^{2}\left(c_{A}+A\right)}.
\]

Finally, if \textup{$|T_{B}|\gg|T_{A}|\gg1$, }\textup{\emph{then
the price of stability is asymptotically $1$.}}
\end{prop}
\begin{proof}
This structure is immune to removal of links as a disconnection of
a $(type-B,type-A)$ link will disconnect the type-B node, and the
type-A clique is stable (lemma \ref{lem:The-longest-distance}). For
every player $j$ and $i\in T_{B}$, any additional link $(i,j)$
will result in $\Delta C(j,E+ij)\geq c_{B}-1>0$ since the link only
reduces the distance $d(i,j)$ from 2 to 1. Hence, player $j$ has
no incentive to accept this link and no additional links will be formed.
This concludes the stability proof. 

We now turn to discuss the optimality of this network structure. First,
consider a set of type-A players. Every link reduce the distance of
at least two nodes by at least one, hence the social cost change by
introducing a link is negative, since $2c_{A}-2A<0$. Therefore, in
any optimal configuration the type-A nodes form a complete graph.
The other terms in the social cost are due to the inter-connectivity
of type-B nodes and the type-A to type-B connections. As $deg(i)=1$
for all $i\in T_{B}$ the cost due to link's prices is minimal. Furthermore,
$d(i,j)=1$ and the distance cost to node $j$ (of type A) is minimal
as well. For all other nodes $j'$, $d(i,j')=2$. 

Assume this configuration is not optimal. Then there is a \emph{topologically
different} configuration in which there exists an additional node
$j'\in T_{A}$ for which $d(i,j')=1$ for some $i\in T_{B}$. Hence,
there's an additional link $(i,j)$. The social cost change is $2c+1+A$
. Therefore, if $\frac{A+1}{2}\le c$ this link reduces the social
cost. On the other hand, if $\frac{A+1}{2}>c$ every link connecting
a type-B player to a type-A player improves the social cost, although
the previous discussion show these link are unstable. In this case,
the optimal configuration is where all type-B nodes are connected
to all the type-A players, but there are no links linking type-B players.
This concludes the optimality proof.

The cost due to inter-connectivity of type A nodes is

\[
c_{A}|T_{A}|\left(|T_{A}|-1\right)+A|T_{A}|\left(|T_{A}|-1\right)=|T_{A}|\left(|T_{A}|-1\right)\left(c_{A}+A\right).
\]

The first expression is due to the cost of $|T_{A}|$ clique's links
and the second is due to distance (=1) between each type-A node. The
distance of each type B nodes to all the other nodes is exactly 2,
except to node $j$, to which its distance is 1. Therefore the social
cost due to type B nodes is 
\begin{align*}
 & 2|T_{B}|(|T_{B}|-1)+2c_{B}|T_{B}|+2\left(A+1\right)|T_{B}|\left(\left(|T_{A}|-1\right)+\left(A+1\right)+2\left(A+1\right)(|T_{A}|-1)\right)\\
= & 2|T_{B}|\left(|T_{B}|-1+c_{B}+\left(A+1\right)(|T_{A}|-1/2)\right).
\end{align*}

The terms on the left hand side are due to (from left to right) the
distance between nodes of type B, the cost of each type-B's single
link, the cost of type-B nodes due to the distance (=2) to all member
of the type-A clique bar $j$ and the cost of type $B$ nodes due
to the distance (=1) to node $j$. The social cost is 
\begin{eqnarray*}
\sum C(i) & = & 2|T_{B}|\left(|T_{B}|-1+c+\left(A+1\right)(|T_{A}|-1/2)\right)+|T_{A}|\left(|T_{A}|-1\right)\left(c_{A}+A\right).
\end{eqnarray*}

To complete the proof, note that if $\frac{A+1}{2}>c$ the latter
term in the social cost of the optimal (and unstable) solution is
\[
2|T_{B}|(|T_{B}|-1)+2c|T_{B}|\left(1+|T_{A}|\right)+\left(A+1\right)|T_{B}||T_{A}|=2|T_{B}|\left(|T_{B}|-1+\left(\frac{A+1}{2}+c\right)|T_{A}|\right).
\]

As the number of links is $|T_{B}|\left(1+|T_{A}|\right)$ and the
distance of type-B to type-A nodes is 1. The optimal social cost is
then
\[
2|T_{B}|\left(|T_{B}|-1+\left(\frac{A+1}{2}+c\right)|T_{A}|\right)+|T_{A}|\left(|T_{A}|-1\right)\left(c_{A}+A\right).
\]

Considering all quantities in the limit $|T_{B}|\gg|T_{A}|\gg1$ completes
the proof.\end{proof}

Next, we evaluate the price of anarchy. The social cost in the stabilizable
topology presented in Fig \ref{fig:A-poor-equilibrium}, composed
of a type-A clique and long lines of type-B players, is calculated
in \citep{Meirom2013}. The ratio between this value and the optimal
social cost constitutes a lower bound on the price of anarchy. An
upper bound is obtained by examining the social cost in any topology
that satisfies Lemma \ref{lem:The-longest-distance}. The result in
the large network limit is presented by the following proposition.

\begin{prop}
\label{thm:summary-of-results}If $c_{B}<A$ and $|T_{B}|\gg|T_{A}|\gg1$
the price of anarchy is $\Theta(c_{B})$. 
\end{prop}

\section{Basic model - Dynamics}

The Internet is a rapidly evolving network. In fact, it may very well
be that it would never reach an equilibrium as ASs emerge, merge,
and draft new contracts among them. Therefore, a dynamic analysis
is a necessity. We first define the dynamic rules. Then, we analyze
the basin of attractions of different states, indicating which final
configurations are possible and what their likelihood is. We shall
establish that reasonable dynamics converge to \emph{just a few} equilibria.
Lastly, we investigate the speed of convergence, and show that convergence
time is \emph{linear} in the number of players.

\subsection{Setup \& Definitions}

At each point in time, the network is composed of a subset $N'\subset T_{A}\cup T_{B}$
of players that already joined the game. The cost function is calculated
with respect to the set of players that are present (including those
that are joining) at the considered time. The game takes place at
specific times, or \emph{turns}, where at each turn only a single
player is allowed to remove or initiate the formation of links. We
split each turn into \emph{acts}, at each of which a player either
forms or removes a single link. A player's turn is over when it has
no incentive to perform additional acts.
\begin{defn}
Dynamic Rule \#1: In player $i$'s turn it may choose to act $m\in\mathcal{N}$
times. In each act, it may remove a link $(i,j)\in E$ or, if player
$j$ agrees, it may establish the link $(i,j)$. Player $j$ would
agree to establish $(i,j)$ iff $C(j;E+(i,j))-C(j;E)<0$.
\end{defn}
The last part of the definition states that, during player's $i$
turn, all the other players will act in a greedy, rather than strategic,
manner. For example, although it may be that player $j$ prefers that
a link $(i,j')$ would be established for some $j'\neq j$, if we
adopt Dynamic Rule \#1 it will accept the establishment of the less
favorable link $(i,j).$ In other words, in a player's turn, it has
the advantage of initiation and the other players react to its offers.
This is a reasonable setting when players cannot fully predict other
players' moves and offers, due to incomplete information \citep{5173479}
such as the unknown cost structure of other players. Another scenario
that complies with this setting is when the system evolves rapidly
and players cannot estimate the condition and actions of other players.

The next two rules consider the ratio of the time scale between performing
the strategic plan and evaluation of costs. For example, can a player
remove some links, disconnect itself from the graph, and then pose
a credible threat? Or must it stay connected? Does renegotiating take
place on the same time scale as the cost evaluation or on a much shorter
one? The following rules address the two limits. 

\begin{defn}
Dynamic Rule \#2a: Let the set of links at the current act $m$ be
denoted as $E_{m}$. A link $(i,j)$ will be added if $i$ asks to
form this link and $C(j;E_{m}+ij)<C(j;E_{m})$. In addition, any link
$(i,j)$ can be removed in act $m.$\end{defn}

The alternative is as follows.
\begin{defn}
Dynamic Rule \#2b: In addition to Dynamic Rule \#2a, player $i$ would
only remove a link $(i,j)$ if $C(i;E_{m}-ij)>C(i;E_{m})$ and would
establish a link if both $C(j;E_{m}+ij)<C(j;E_{m})$ and $C(i;E_{m}+ij)<C(i;E_{m})$. 
\end{defn}
The difference between the last two dynamic rules is that, according
to Dynamic Rule \#2a, a player may perform a strategic plan in which
the first few steps will increase its cost, as long as when the plan
is completed its cost will be reduced.  On the other hand, according
to Dynamic Rule \#2b, its cost must be reduced \emph{at each act},
hence such ``grand plan'' is not possible. Note that we do not need
to discuss explicitly disconnections of several links, as these can
be done unilaterally and hence iteratively.Finally, the following lemma will be useful in the next section.
\begin{lem}
\label{lem:decay time}Assume $N$ players act consecutively in a
(uniformly) random order\textup{\emph{ at integer times, which we'll
denote by $t$.}} the probability $P(t)$ that a specific player did
not act $k\mathcal{\in N}$ times by $t\gg N$ decays exponentially.\end{lem}

\subsection{\label{sub:dynamical Results}Results}

After mapping the possible dynamics, we are at a position to consider
the different equilibria's basins of attraction. Specifically, we
shall establish that, in most settings, the system converges to the
optimal network, and if not, then the network's social cost is asymptotically
equal to the optimal social cost. The main reason behind this result
is the observation that a disconnected player has an immense bargaining
power, and may force its optimal choice. As the highest connected
node is usually the optimal communication partner for other nodes,
new arrivals may force links between them and this node, forming a
star-like structure. There may be few star centers in the graphs,
but as one emerges from the other, the distance between them is small,
yielding an optimal (or almost optimal) cost.

We outline the main ideas of the proof. The first few type-B players,
in the absence of a type-A player, will form a star. The star center
can be considered as a new type of player, with an intermediate importance,
as presented in Fig. \ref{fig:credible threat corr.}. We monitor
the network state at any turn and show that the minor players are
organized in two stars, one centered about a minor player and one
centered about a major player (Fig. \ref{fig:credible threat corr.}(a)).
Some cross links may be present (Fig.\,\ref{fig:cross-tiers}). By
increasing its client base, the incentive of a major player to establish
a direct link with the star center is increased. This, in turn, increases
the attractiveness of the star's center in the eyes of minor players,
creating a positive feedback loop. Additional links connecting it
to all the major league players will be established, ending up with
the star's center transformation into a member of the type-A clique.
On the other hand, if the star center is not attractive enough, then
minor players may disconnect from it and establish direct links with
the type-A clique, thus reducing its importance and establishing a
negative feedback loop. The star will become empty, and the star's
center $x$ will be become a stub of a major player, like every other
type-B player. The latter is the optimal configuration, according
to proposition \ref{lem:optimal solution}. We analyze the optimal
choice of the active player, and establish that the optimal action
of a minor player depends on the number of players in each structure
and on the number of links between the major players and the minor
players' star center $x$. The latter figure depends, in turn, on
the number of players in the star. We map this to a two dimensional
dynamical system and inspect its stable points and basins of attraction
of the aforementioned configurations.

\begin{thm}
\label{cor:credible threat part 3}If the game obeys Dynamic Rules
\#1 and \#2a, then, in any playing order:

a) The system converges to a solution in which the total cost is at
most 

\begin{eqnarray*}
\mathcal{S} & = & |T_{A}|\left(|T_{A}|-1\right)\left(c+A\right)+2c_{B}|T_{B}|+\left(A+1\right)\left(3|T_{A}||T_{B}|-|T_{A}|+|T_{B}|\right)+2\left(|T_{B}|-1\right)^{2};
\end{eqnarray*}

furthermore, by taking the large network limit  $|T_{B}|\gg|T_{A}|\gg1$,
we have $\mathcal{S}/\mathcal{S}_{optimal}\rightarrow1$ .

b) Convergence to the optimal stable solution occurs if either:

1) \textup{$A\cdot k_{A}>k+1$, }\textup{\emph{where $k\geq0$ is
the number of type-B nodes that first join the network, followed later
by $k_{A}$ consecutive type-A nodes (``initial condition'').}}

2)$A\cdot|T_{A}|>|T_{B}|$ (``final condition'').

c) In all of the above, if every player plays at least once in O(N)
turns, convergence occurs after O(N) steps. Otherwise, if players
play in a uniformly random order, the probability the system has not
converged by turn $t$ decays exponentially with $t$.\end{thm}
\begin{proof}
Assume $c_{A}\geq2$. Denote the first type-A player that establish
a link with a type-B player as $k$. First, we show that the network
structure is composed of a type-A (possibly empty) clique, a set of
type-B players $S$ linked to player $x$, and an additional (possibly
empty) set of type-B players $L$ connected to the type-A player $k$.
See Fig. \ref{fig:credible threat corr.}(a) for an illustration.
In addition, there is a set $D$ type-A nodes that are connected to
node $x$, the star center. After we establish this, we show that
the system can be mapped to a two dimensional dynamical system. Then,
we evaluate the social cost at each equilibria, and calculate the
convergence rate. We assume $(k,x)\in E$ and discuss the case $(k,x)\notin E$
in the appendix. 

\begin{figure}
\centering{}\includegraphics[width=0.75\columnwidth]{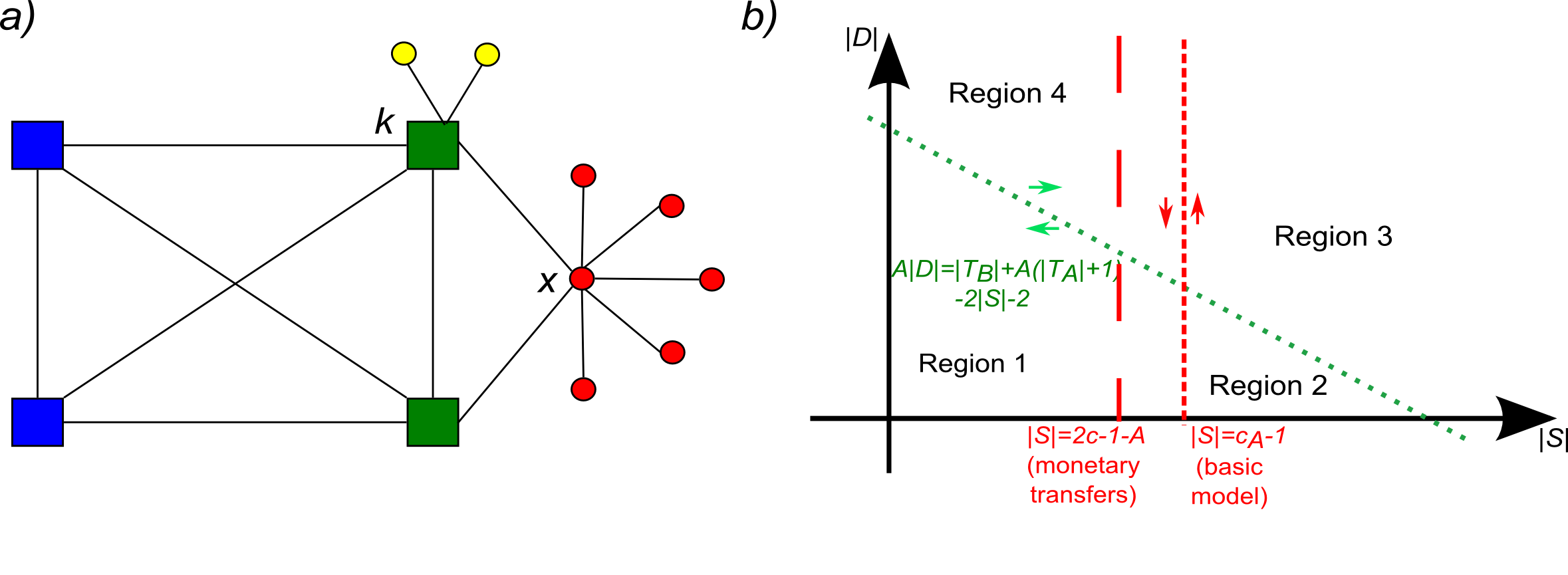}\caption{\label{fig:credible threat corr.}a) The network structures described
in Theorem \ref{cor:credible threat part 3}. The type-A clique contains
$|T_{A}|=4$ nodes (squares), and there are $|S|=5$ nodes in the
star (red circles). There are $|L|=2$ nodes that are connected directly
to node $k$ (yellow circles). The number of type-A nodes that are
connected to node 1, the star center, is $|D|=2$ (green squares).
b) The phase state of Theorem \ref{cor:credible threat part 3}. The
dotted green line is the $|S|$ increase / decrease nullcline. The
dotted (dashed) red line is the nullcline for the increase / decrease
in $|D|$ when monetary transfers are forbidden (allowed). (Proposition
\ref{prop:monetary-dyanmics-1}). \label{fig:The-phase-state}}
\end{figure}

We prove by induction. At turn $t\leq2$, after the first two players
joined the network, this is certainly true. Denote the active player
at time $t$ as $r.$ Consider the following cases:

1. $r\in T_{A}$: Since $1<c_{A}<A,$ all links to the other type-A
nodes will be established (lemma \ref{lem:optimal solution}) or maintained,
if $r$ is already connected to the network. Clearly, the optimal
link in $r$'s concern is the link with star center $x$. As $c_{B}<A$
every minor player will accept a link with a major player even if
it reduces its distance only by one. Therefore, the link $(r,x)$
is formed if the change of cost of the major player $r$,
\begin{equation}
\Delta C(r,E+rx)=c_{A}-|S|-1\label{eq: term1}
\end{equation}
is negative. In this case, the number of type A players connected
to the star's center, $|D|$, will increase by one. If this expression
is positive and player $r$ is connected to at least another major
player (as otherwise the graph is disconnected), the link will be
dissolved and $|D|$ will be reduced by one. It is not beneficial
for $r$ to form an additional link to any type-B player, as they
only reduce the distance from a single node by one (see the discussion
in lemma \ref{lem:optimal solution} in the appendix). 

2. $r\in T_{B},\, r\neq x$ : First, assume that $r$ is a newly arrived
player, and hence it is disconnected. Obviously, in its concern, a
link to the star's center, player $x$, is preferred over a link to
any other type-B player. Similarly, a link to a type-A player that
is linked with the star's center is preferred over a link with a player
that maintains no such link.

We claim that either $(r,k)$ or $(r,x)$ exists. Denote the number
of type-A player at turn $t$ as $m_{A}.$ The link $(r,x)$ is preferred
in $r$'s concern if the expression 

\begin{equation}
C(r,E+rk)-C(r,E+rx)=-A(1+m_{A}-|D|)+1+|S|-|L|\label{eq: term 2}
\end{equation}

is positive, and will be established as otherwise the network is disconnected.
If the latter expression is negative, $(r,k)$ will be formed. The
same reasoning as in case 1 shows that no additional links to a type-B
player will be formed. Otherwise, if $r$ is already connected to
the graph, than according to Dynamic Rule \#2a, $r$ may disconnect
itself, and apply its optimal policy, increasing or decreasing $|L|$
and $|S|$.

3. $r=x$, the star's center: $r$ may not remove any edge connected
to a type-B player and render the graph disconnected. On the other
hand, it has no interest in removing links to major players. On the
contrary, it will try to establish links with the major players, and
these will be formed if eq. \ref{eq: term1} is negative. An additional
link to a minor player connected to $k$ will only reduce the distance
to it by one and since $c_{B}>2$ player $x$ would not consider this
move worthy.

The dynamical parameters that govern the system dynamics are the number
of players in the different sets, $|S|$, $|L|$, and $|D|$. Consider
the state of the system after all the players have player once. Using
the relations $|S|+|L|+1=|T_{B}|,\: m_{A}=|T_{A}|$ we note the change
in $|S|$ depends on $|S|$ and $|L|$ while the change in $|D|$
depends only on $|S|.$ We can map this to a 2D dynamical, discrete
system with the aforementioned mapping. In Fig.\,\ref{fig:The-phase-state}
the state is mapped to a point in phase space $(|S|,|L|)$. The possible
states lie on a grid, and during the game the state move by an single
unit either along the $x$ or $y$ axis. There are only two stable
points, corresponding to $|S|=0,|D|=1$, which is the optimal solution
(Fig. \ref{fig:The-optimal-solution}(a)), and the state $|S|=|T_{B}|-1$
and $|D|=|T_{A}|$. 

If at a certain time expression \ref{eq: term1} is positive and expression
\ref{eq: term 2} is negative (region 3 in Fig.\,\ref{fig:The-phase-state}(b)),
the type-B players will prefer to connect to player $x$. This, in
turn, increases the benefit a major player gains by establishing a
link with player $x.$ The greater the set of type-A that have a direct
connection with $x$, having $|D|$ members, the more utility a direct
link with $x$ carries to a minor player. Hence, a positive feedback
loop is established. The end result is that all the players will form
a link with $x$. In particular, the type-A clique is extended to
include the type-B player $x$. Likewise, if the reverse condition
applies, a feedback loop will disconnect all links between node $x$
to the clique (except node $k$) and all type-B players will prefer
to establish a direct link with the clique. The end result in this
case is the optimal stable state. The region that is relevant to the
latter domain is region 1. 

However, there is an intermediate range of states, described by region
2 and region 4, in which the player order may dictate to which one
of the previous states the system will converge. For example, starting
from a point in region 4, if the type-A players move first, changing
the $|D|$ value, than the dynamics will lead to region 1, which converge
to the optimal solution. However, if the type-B players move first,
then the system will converge to the other equilibrium point.

We now turn to calculate the social cost at the different equilibria.
If $|D|=|T_{A}|$ and $|S|=|T'_{B}|-1$, The network topology is composed
of a $|T_{A}|$ members clique, all connected to the center $x$,
that, in turn, has $|T_{B}|-1$ stubs. The total cost in this configuration
is 
\begin{eqnarray}
S & = & |T_{A}|\left(|T_{A}|-1\right)\left(c_{A}+A\right)+2c_{B}|T_{B}|+\left(A+1\right)|T_{A}|+2\left(|T_{B}|-1\right)\nonumber \\
 &  & +2\left(|T_{B}|-1\right)\left(A+1\right)+2\left(|T_{B}|-2\right)\left(|T_{B}|-1\right)+\left(c_{B}+c_{A}\right)|T_{A}|/2\label{eq:cost at star}
\end{eqnarray}

where the costs are, from the left to right: the cost of the type-A
clique, the cost of the type-B star's links, the distance cost $(=1)$
between the clique and node $x$, the distance $(=1)$ cost between
the star's members and node $x$, the distance $(=2)$ cost between
the clique and the star's member, the distance $(=2)$ cost between
the star's members, and the cost due to major player link's to the
start center $x$. Adding all up, we have for the total cost
\begin{eqnarray}
\mathcal{S} & \leq & |T_{A}|\left(|T_{A}|-1\right)\left(c+A\right)+2c_{B}|T_{B}|+\left(A+1\right)\left(3|T_{A}||T_{B}|+|T_{B}|\right)+2\left(|T_{B}|-1\right)^{2}.\label{eq:social cost}
\end{eqnarray}

Convergence is fast, and as soon as all players have acted three times
the system will reach equilibrium. If every player plays at least
once in $o(N$) turns convergence occurs after $o(N)$ turns, otherwise
the probability the system did not reach equilibrium by time $t$
decays exponentially with $t$ according to lemma \ref{lem:decay time}
(in the appendix).

We now relax our previous assumption $c_{A}\geq2$. If $c_{A}\leq2$
and the active player $r\in T_{A}$ then it will form a link with
the star's center according to eq. \ref{eq: term1}. If $r\in S$
it may establish a link $(r,j)$ with a type A player, which will
later be replaced, in $j$'s turn, with the link $(j,x)$ according
to the previous discussion. In the appendix we discuss explicitly
the case where $(k,x)\notin E$ and show that in this case, additional
links may be formed, e.g., a link between one of $k'$s stubs, $i\in L$,
and the star's center $x$, as presented in Fig.\,\ref{fig:cross-tiers}.
These links only reduce the social cost, and do not change the dynamics,
and the system will converge to either one of the aforementioned states.
Taking the limit $T_{B}\rightarrow\infty$ and $T_{B}\in\omega\left(T_{A}\right)$
in eq. \ref{eq:social cost}, we get  $\mathcal{S}/\mathcal{S}_{optimal}\rightarrow1$.
This concludes the proof.

\begin{figure}
\begin{centering}
\includegraphics[width=0.6\columnwidth]{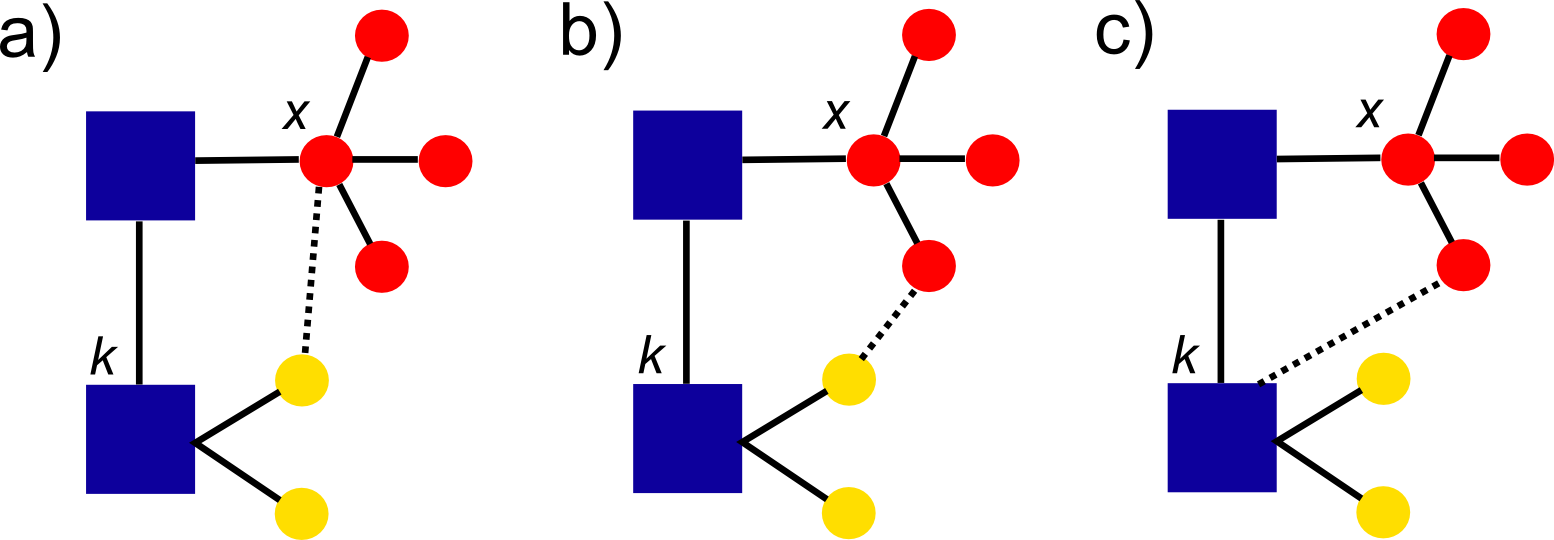}\caption{\label{fig:cross-tiers}Additional feasible cross-tiers links, as
described in \citet{Meirom2013}. The star players $S$ are in red,
the set $L$ is in yellow. a) a link between the star center and $i\in L$.
b) a cross-tier link $(i,j)$ where $i\in S,j\in L$. c) a minor player
- major player link, $(i,j)$ where $i\in T_{A}$ and $j\in S.$ }

\par\end{centering}

\end{figure}

\end{proof}

If the star's center has a principal role in the network, then links
connecting it to all the major league players will be established,
ending up with the star's center transformation into a member of the
type-A clique. This dynamic process shows how an effectively new major
player emerges out of former type-B members in a natural way. Interestingly,
Theorem \ref{cor:credible threat part 3} also shows that there exists
a transient state with a better social cost than the final state.
In fact, in a certain scenario, the transient state is better than
the optimal stable state.

So far we have discussed the possibility that a player may perform
a strategic plan, implemented by Dynamic Rule \#2a. However, if we
follow Dynamic Rule \#2b instead, then a player may not disconnect
itself from the graph. The previous results indicate that it is not
worthy to add additional links to the forest of type-B nodes. Therefore,
no links will be added except for the initial ones, or, in other words,
renegotiation will always fail. The dynamics will halt as soon as
each player has acted once. Formally:

\begin{prop}
\label{prop:credible theat part 4}If the game obeys Dynamic Rules
\#1 and \#2b, then the system will converge to a solution in which
the total cost is at most 
\begin{eqnarray*}
\mathcal{S} & = & |T_{A}|\left(|T_{A}|-1\right)\left(c_{A}+A\right)+3|T_{B}|^{2}+2c_{B}|T_{B}|+2|T_{A}||T_{B}|\left(A+1\right)\,.
\end{eqnarray*}
Furthermore, for $|T_{B}|\gg|T_{A}|\gg1$, we have $\mbox{\ensuremath{\mathcal{S}/\mathcal{S}_{optimal}\leq3/2.}}$
Moreover, if every player plays at least once in O(N) turns, convergence
occurs after O(N) steps. Otherwise, if players play in a uniformly
random order, the probability the system has not converged by turn
$t$ decays exponentially with $t$.\end{prop}
\begin{proof}
We discuss the case $c_{A}\geq2.$ The extension for $c_{A}<2$ appears
in the appendix. The first part of the proof follows the same lines
of the previous theorem (Theorem \ref{cor:credible threat part 3}).
We claim that at any given turn, the network structure is composed
of the same structures as before (See Fig. \ref{fig:credible threat corr.}(a))
. Here, we discuss the scenario where $(k,x)\in E,$ and we address
the other possibility, which may give rise to the structures shown
in Fig. \ref{fig:cross-tiers} in the appendix. 

We prove by induction. Clearly, at turn one the induction assumption
is true. Note that for newly arrived players, are not affected by
either Dynamic Rules \#2a or \#2b. Hence, we only need to discuss
the change in policies of existing players. The only difference from
the dynamics described in the Theorem \ref{cor:credible threat part 3}
is that the a type-B players may not disconnect itself. In this case,
as the discussion there indicates the star center $x$ will refuse
a link with $i\in L$ as it only reduce $d(i,x)$ by two. Equivalently,
$k$ will refuse to establish additional links with $i\in|S|.$

In other words, as soon the first batch of type A player arrives,
all type-B players will become stagnant, either they become leaves
of either node $k$, $|L|$, or members of the star $|S|$, according
to the the sign of \ref{eq: term 2} at the time they. The maximal
distance between a type-A player and a type B player is $2$. The
maximal value of the type B - type B term is the social cost function
is when $|L|=|S|=|T_{B}|/2$. In this case, this term contributes
$3|T_{B}|^{2}$ to the social cost. Therefore, the social cost is
bounded by 
\begin{equation}
\mathcal{S}=|T_{A}|\left(|T_{A}|-1\right)\left(c_{A}+A\right)+3|T_{B}|^{2}+2c_{B}|T_{B}|+2|T_{A}||T_{B}|\left(A+1\right)\label{eq:bound for 2b-1}
\end{equation}

where we included the type-A clique's contribution to the social cost
and used $c_{B}\geq c_{A}.$ Taking the limit $N\rightarrow\infty$
in eq. \ref{eq:bound for 2b-1} and using $T_{A}\in\omega(1)$, $T_{B}\in\omega(T_{A})$,
we obtain $\mbox{\ensuremath{\mathcal{S}/\mathcal{S}_{optimal}\leq3/2}}$. \end{proof}

Theorem \ref{cor:credible threat part 3} and Proposition \ref{prop:credible theat part 4}
shows that the intermediate network structures of the type-B players
are not necessarily trees, and additional links among the tier two
players may exist, as found in reality. Furthermore, our model predicts
that some cross-tier links, although less likely, may be formed as
well. If Dynamic Rule \#2a is in effect, These structures are only
transient, otherwise they might remain permanent. 

The dynamical model can be easily generalized to accommodate various
constraints. Geographical constraints may limit the service providers
of the minor player. The resulting type-B structures represent different
geographical regions. Likewise, in remote locations state legislation
may regulate the Internet infrastructure. If at some point regulation
is relaxed, it can be modeled by new players that suddenly join the
game.

\section{Monetary transfers }

\label{sec:Monetary-transfers}

So far we assumed that a player cannot compensate other players for
an increase in their costs. However, contracts between different ASs
often do involve monetary transfers. Accordingly, we turn to consider
the effects of introducing such an option on the findings presented
in the previous sections. As before, we first consider the static
perspective and then turn to the dynamic perspective.

\subsection{Statics}

In the previous sections we showed that, if $A>c_{A}>1,$ then it
is beneficial for each type-A player to be connected to all other
type-A players. We focus on this case. 

Monetary transfers allow for a redistribution of costs. It is well
known in the game theoretic literature that, in general, this process
increases the social welfare.Indeed, the next proposition indicates
an improvement on Proposition \ref{lem:optimal solution}. Specifically,
it shows that the optimal network is always stabilizable, even when
$\frac{A+1}{2}>c$. Without monetary transfers, the additional links
in the optimal state (Fig. \ref{fig:The-optimal-state - monetary transfers}),
connecting a major league player with a minor league player, are unstable
as the type-A players lack any incentive to form them. By allowing
monetary transfers, the minor players can compensate the major players
for the increase in their costs. It is worthwhile to do so only if
the social optimum of the two-player game implies it. The existence
or removal of an additional link does not inflict on any other player,
as the distance between every two players is at most two.

\begin{prop}
\label{prop:optimality under monetary}\textup{\emph{The price of
stability is $1$.}} If $\frac{A+1}{2}\leq c\,,$ then Proposition
\ref{lem:optimal solution} holds. Furthermore, if $\frac{A+1}{2}>c$,
then the optimal stable state is such that all the type $B$ nodes
are connected to all nodes of the type-A clique. In the latter case,
the social cost of this stabilizable network is $\mathcal{S}=2|T_{B}|\left(|T_{B}|+\left(\frac{A+1}{2}+c\right)|T_{A}|\right)+|T_{A}|^{2}\left(c+A\right).$
Furthermore, if $|T_{B}|\gg1,|T_{A}|\gg1$ then, omitting linear terms
in $|T_{B}|,|T_{A}|$, $\mathcal{S}=2|T_{B}|(|T_{B}|+\left(A+c\right)|T_{A}|)+|T_{A}|^{2}\left(c+A\right).$
\end{prop}

In the network described by Fig. \ref{fig:The-optimal-state - monetary transfers},
the minor players are connected to multiple type-A players. This emergent
behavior, where ASs have multiple uplink-downlink but very few (if
at all) cross-tier links, is found in many intermediate tiers.

Next, we show that, under mild conditions on the number of type-A
nodes, the price of anarchy is $3/2$, i.e., \emph{a fixed number}
that does not depend on any parameter value. As the number of major
players increases, the motivation to establish a direct connection
to a clique member increases, since such a link reduces the distance
to all clique members. As the incentive increases, players are willing
to pay more for this link, thus increasing, in turn, the utility of
the link in a major player's perspective. With enough major players,
all the minor players will establish direct links. Therefore, any
stable equilibrium will result in a very compact network with a diameter
of at most three. This is the main idea behind the following theorem.

\begin{thm}
\label{prop:maximal distance from clique with money}The maximal distance
of a type-B node from a node in the type-A clique is $\max\left\{ \left\lfloor \sqrt{\left(A|T_{A}|\right)^{2}+4cA|T_{A}|}-A|T_{A}|\right\rfloor ,2\right\} $.
Moreover, if \textup{$|T_{B}|\gg1,|T_{A}|\gg1$ }\textup{\emph{then
the price of anarchy is upper-bounded by 3/2.}}\end{thm}

This theorem shows that by allowing monetary transfers, the maximal
distance of a type-B player to the type-A clique depends inversely
on the number of nodes in the clique and the number of players in
general. The number of ASs increases in time, and we may assume the
number of type-A players follows. Therefore, we expect a decrease
of the mean ``node-core distance'' in time. Our data analysis, which
appears in the appendix, indicates that this real-world distance indeed
decreases in time.

\subsection{Dynamics\label{sub: monetary Dynamics}}

We now consider the dynamic process of network formation under the
presence of monetary transfers. For every node $i$ there may be several
nodes, indexed by \emph{j, }such that $\Delta C(j,ij)+\Delta C(i,ij)<0,$
and player \emph{i }needs to decide on the order of players with which
it will ask to form links. We point out that the order of establishing
links is potentially important. The order by which player player $i$
will establish links depends on the pricing mechanism. There are several
alternatives and, correspondingly, several possible ways to specify
player \emph{i's }preferences\emph{, }each leading to a different
dynamic rule. 

Perhaps the most naive assumption is that if for player $j,$ $\Delta C(j,ij)>0$,
then the price it will ask player $i$ to pay is $P_{ij}=\max\{\Delta C(j,ij),0\}.$
In other words, if it is beneficial for player $j$ to establish a
link, it will not ask for a payment in order to do so. Otherwise,
it will demand the minimal price that compensates for the increase
in its costs. This dynamic rule represents an efficient market. This
suggests the following preference order rule.

\begin{defn}
Preference Order \#1: Player $i$ will establish a link with a player
$j$ such that $\Delta C(i,ij)+min\{\Delta C(j,ij),0\}$ is minimal.
The price player $i$ will pay is $P_{ij}=max\{\Delta C(j,ij),0\}$.\end{defn}

As established by the next theorem, Preference Order \#1 leads to
the optimal equilibrium fast. In essence, if the clique is large enough,
then it is worthy for type-B players to establish a direct link to
the clique, compensating a type-A player, and follow this move by
disconnecting from the star. Therefore, monetary transfers increase
the fluidity of the system, enabling players to escape from an unfortunate
position. Hence, we obtain an improved version of Theorem \ref{cor:credible threat part 3}.
\begin{thm}
Assume the players follow Preference Order \#1 and Dynamic Rule \#1,
and either Dynamic Rule \#2a or \#2b. If $\frac{A+1}{2}>c$, then
the system converges to the optimal solution. If every player plays
at least once in O(N) turns, convergence occurs after o(N) steps.
Otherwise, e.g., if players play in a random order, convergence occurs
exponentially fast. 
\end{thm}

Yet, the common wisdom that monetary transfers, or utility transfers
in general, should increase the social welfare, is contradicted in
our setting by the following proposition. Specifically, there are
certain instances, where allowing monetary transfers yields a decrease
in the social utility. In other words, if monetary transfers are allowed,
then the system may converge to a sub-optimal state.

\begin{prop}
Assume $\frac{A+1}{2}\leq c$. Consider the case where monetary transfers
are allowed and the game obeys Dynamic Rules \#1,\#2a and Preference
Order \#1. Then:

a) The system will either converge to the optimal solution or to a
solution in which the social cost is 
\begin{eqnarray*}
\mathcal{S} & = & |T_{A}|\left(|T_{A}|-1\right)\left(c_{A}+A\right)+2\left(|T_{B}|-1\right)^{2}+\left(A+1\right)\left(3|T_{A}||T_{B}|-|T_{A}|+|T_{B}|\right)+2c|T_{B}|.
\end{eqnarray*}
For $|T_{B}|\rightarrow\infty$, $|T_{B}|\in\omega\left(|T_{A}|\right)$
we have $S/S_{optimal}\rightarrow1\,$. In addition, if one of the
first $\left\lfloor c-1\right\rfloor $ nodes to attach to the network
is of type-A then the system converges to the optimal solution. 

b) For some parameters and playing orders, the system converges to
the optimal state if monetary transfers are forbidden, but when transfers
are allowed it fails to do so. This is the case, for example, when
the first $k$ players are of type-B, and $2c-A-1<k<c-1$.\end{prop}

\begin{proof}
a) We claim that, at any given turn $t$, the network is composed
of the same structures as in Theorem \ref{cor:credible threat part 3}.
We use the notation described there. See Fig. \ref{fig:credible threat corr.}
for an illustration. We assume that the link $(k,x)$ exists and elaborate
in the appendix on the scenario that, at some point, the link $(k,x)$
is removed.

We prove by induction. At turn $t=1$ the induction hypothesis is
true. We'll discuss the different configurations at time $t$. 

1. $r\in T_{A}$: As before, all links to the other type-A nodes will
be established or maintained, if $r$ is already connected to the
network. The link $(r,x)$ will be formed if the change of cost of
player $r$,
\begin{equation}
\Delta C(r,E+rx)+\Delta C(x,E+rx)=2c-A-|S|-1\label{eq: term1-1}
\end{equation}
is negative. In this case $|S|$ will increase by one. If this expression
is positive and $(r,x)\in E,$ the link will be dissolved and $|D|$
will be reduced. It is not beneficial for $r$ to form an additional
link to any type-B player, as they only reduce the distance from a
single node by 1 and $\frac{A+1}{2}\leq c$. 

2. $r\in T_{B}$ :The discussion in Theorem \ref{cor:credible threat part 3}
shows that a newly arrived may choose to establish its optimal link,
which would be either $(r,k)$ or $(r,x)$ according to the sign of
expression \ref{eq: term 2}. As otherwise the graph is disconnected,
such link will cost nothing. Similarly, if $r$ is already connected,
it may disconnect itself as an intermediate state and use its improved
bargaining point to impose its optimal choice. Hence, the formation
of either $(r,k)$ or $(r,x)$ is not affected by the inclusion of
monetary transfers to the basic model. Assume the optimal move for
$r$ is to be a member of the star, $r\in S$. If $\Delta C(k,E+kr)+\Delta C(r,E+kr)=2c-A|m_{A}|-1-|L|$
is negative, than this link will be formed. In this case, $r$ is
a member of both $S$ and $L$, and we address this by the transformation
$|S|\leftarrow|S|$, $|L|\leftarrow|L|+1$ and $|T_{B}|\leftarrow|T_{B}|+1.$
Similarly, if $r\in L$ than it will establish links with the star
center $x$ if and only if $2c<|S|+1$. The analogous transformation
is, $|S|\leftarrow|S|+1$, $|L|\leftarrow|L|$ and $|T_{B}|\leftarrow|T_{B}|+1.$
The rest of the proof follows along the lines of Theorem \ref{cor:credible threat part 3}
and is detailed in the appendix.

b) If dynamic rule \#2a is in effect, the nullcline represented by
eq. \ref{eq: term1-1} is shifted to the left compared to the nullcline
of eq. \ref{eq: term1}, increasing region 3 and region 2 on the expanse
of region 1 and region 4. Therefore, there are cases where the system
would have converge to the optimal state, but allowing monetary transfers
it would converge to the other stable state. Intuitively, the star
center may pay type-A players to establish links with her, reducing
the motivation for one of her leafs to defect and in turn, increasing
the incentive of the other players to directly connect to it. Hence,
monetary transfers reduce the threshold for the positive feedback
loop discussing in Theorem \ref{cor:credible threat part 3}.\end{proof}

The latter proposition shows that the emergence of an effectively
new major league player, namely the star center, occurs more frequently
with monetary transfers, although the social cost is hindered. 

A more elaborate choice of a price mechanism is that of ``strategic''
pricing. Specifically, consider a player $j^{*}$ that knows that
the link $(i,j^{*})$ carries the least utility for player $i$. It
is reasonable to assume that player $j$ will ask the minimal price
for it, as long as it is greater than its implied costs. We will denote
this price as $P_{ij^{*}}$. Every other player $x$ will use this
value and demand an additional payment from player $i$, as the link
$(i,x)$ is more beneficial for player $i$. Formally,

\begin{defn}
Pricing mechanism \#2: Set $j^{*}$ as the node that maximizes $\Delta C(i,E+ij*)$.
Set $P_{ij^{*}}=\max\{-\Delta C(j*,E,ij*),0\}$. Finally, set $\alpha_{ij}=\Delta C(i,E+ij)-\left(\Delta C(i,E+ij^{*})+P_{ij^{*}}\right).$
The price that player $j$ requires in order to establish \emph{$(i,j)$
}is\emph{ }$P_{ij}=\max\{0,\alpha_{ij},-\Delta C(j,E+ij)\}$.\end{defn}

As far as player $i$ is concerned, all the links $(i,j)$ with $P_{ij}=\alpha_{ij}$
carry the same utility, and this utility is greater than the utility
of links for which the former condition is not valid. Some of these
links have a better connection value, but they come at a higher price.
Since all the links carry the same utility, we need to decide on some
preference mechanism for player $i$. The simplest one is the ``cheap''
choice, by which we mean that, if there are a few equivalent links,
then the player will choose the cheapest one. This can be reasoned
by the assumption that a new player cannot spend too much resources,
and therefore it will choose the ``cheapest'' option that belongs
to the set of links with maximal utility.

\begin{defn}
Preference order \#2: Player $i$ will establish links with player
$j$ if player $j$ minimizes $\Delta\tilde{C}(i,ij)=\Delta C(i,ij)+P_{ij}$
and $\Delta\tilde{C}(i,ij)<0$.

If there are several players that minimize $\Delta\tilde{C}(i,ij)$,
then player $i$ will establish a link with a player that minimizes
$P_{ij}$. If there are several players that satisfy the previous
condition, then one out of them is chosen randomly.\end{defn}
Note that low-cost links have a poor ``connection value'' and therefore
the previous statement can also be formulated as a preference for
links with low connection value. 

We proceed to consider the dynamic aspects of the system under such
conditions.

\begin{prop}
\label{prop:monetary-dyanmics-1}Assume that:

A) Players follow Preference Order \#2 and Dynamic Rule \#1, and either
Dynamic Rule \#2a or \#2b. 

B) There are enough players such that $2c<T_{A}\cdot A+T_{B}^{2}/4$.

C) At least one out of the first $m$ players is of type-A, where
$m$ satisfies $m\geq\sqrt{A^{2}+4c-1}-A$.

Then, if the players play in a non-random order, the system converges
to a state where all the type-B nodes are connected directly to the
type-A clique, except perhaps lines of nodes with summed maximal length
of $m$. In the large network limit, $\mathcal{S}/\mathcal{S}_{optimal}<3/2+c$
.

D) If $2c>(A-1)+|T_{B}|/|T_{A}|$ then the bound in (C) can be tightened
to $\mathcal{S}/\mathcal{S}_{optimal}<3/2$.\end{prop}

In order to obtain the result in Proposition \ref{prop:maximal distance from clique with money},
we had to assume a large limit for the number of type-A players. Here,
on the other hand, we were able to obtain a similar result yet without
that assumption, i.e., solely by dynamic considerations. 

It is important to note that, although our model allows for monetary
transfers, in \emph{every} resulting agreement between major players
no monetary transaction is performed. In other words, our model predicts
that the major players clique will form a \emph{settlement-free} interconnection
subgraph, while in major player - minor player contracts transactions
will occur, and they will be of a transit contract type. Indeed, this
observation is well founded in reality.

\section{Conclusions}

Does the Internet resembles a clique or a tree? Is it contracting
or expanding? Can one statement be true on one segment of the network
while the opposite is correct on a different segment? The game theoretic
model presented in this work, while abstracting many details away,
focuses on the essence of the strategical decision-making that ASs
perform. It provides answers to such questions by addressing the different
roles ASs play.

The static analysis has indicated that in all equilibria, the major
players form a clique. Our model predicts that the major players clique
will form a \emph{settlement-free} interconnection subgraph, while
in major player - minor player contracts transactions will occur,
and they will be of a transit contract type. This observation is supported
by the empirical evidence,showing the tight tier-1 subgraph, and the
fact these ASs provide transit service to the other ASs. 

We discussed multiple dynamics, which represent different scenarios
and playing orders. The dynamic analysis showed that, when the individual
players act selfishly and rationally, the system will converge to
either the optimal configuration or to a state in which the social
cost differs by a negligible amount from the optimal social cost.
This is important as a prospective mechanism design. Furthermore,
although a multitude of equilibria exist, the dynamics restrict the
convergence to a limited set. In this set, the minor players' dominating
structures are lines and stars. We also learned that, as the number
of major players increase, the distance of the minor players to the
core should decrease. This theoretical finding was also confirmed
empirically (see the appendix).

In our model, ASs are lumped into two categories. The extrapolation
of our model to a general (multi-tier) distribution of player importance
is an interesting and relevant future research question, the buds
of which are discussed in the appendix. In addition, there are numerous
contract types (e.g., p2p, customer-to-provider, s2s) ASs may form.
While we discussed a network formed by the main type (c2p), the effect
of including various contract types is yet to be explored.

\section*{}
\bibliographystyle{acmsmall}

\pagebreak{}

\appendix

\section{Appendix: Data analysis}

\begin{figure}
\centering{}\includegraphics[width=1\textwidth]{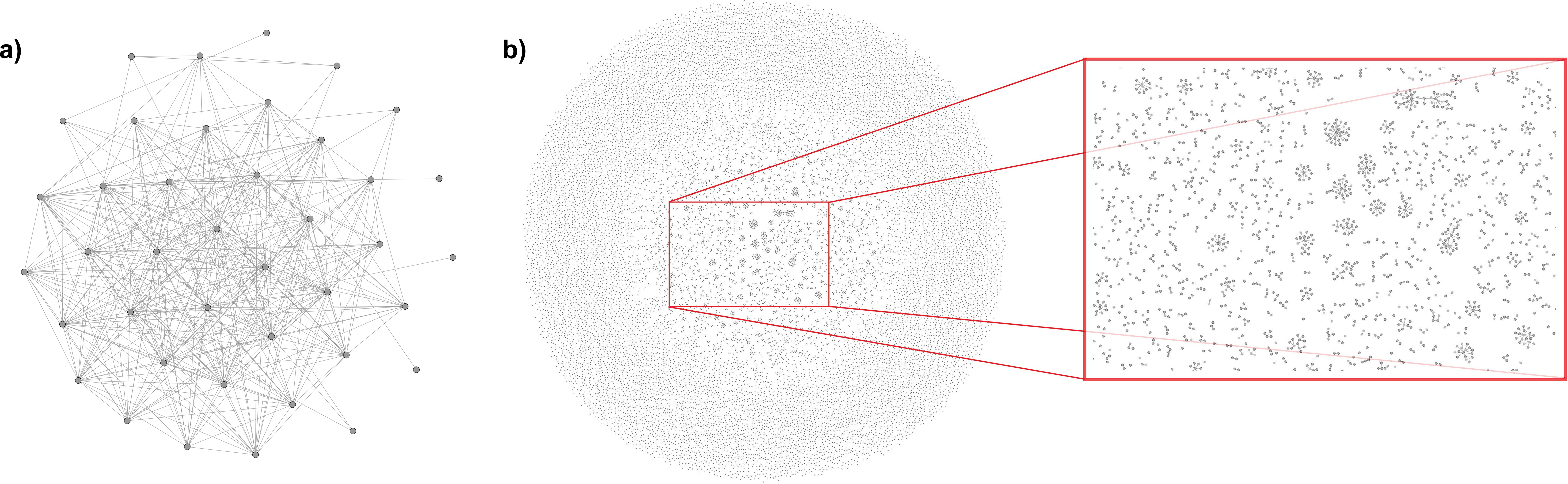}\caption{Structure of the AS topology. a) The sub-graph of the top 40 ASs,
according to CAIDA ranking, in January, 2006. b) The minor nodes sub-graph
was created by omitting nodes in higher k-core $(k\geq3$) and removing
any links from the shell to the core. The subgraph contains 16,442
nodes, which are $\sim75\%$ of the ASs in the networks. Left: The
full network map. Singleton are displayed in the exterior and complex
objects in the center. Right: A zoom-in on a sample (red box) of the
subgraph. The complex structure are mainly short lines and stars (or
star-like objects).}
\label{fig:clique subgraph}
\end{figure}

As discussed in the Introduction, the Internet is composed of autonomous
subsystems, each we consider to be a player. It is one particular
case to which our model can be applied, and in fact it has served
as the main motivation for our study. Accordingly, in this appendix
we compare our theoretical findings with actual monthly snapshots
of the inter-AS connectivity, reconstructed from BGP update messages
\citeapp{Gregori2011}.

Our model predicts that, for $A>c>1$ , the type-A (``major league'')
players will form a highly connected subset, specifically a clique
(Section \ref{sub:The-type-A-clique}). The type-B players, in turn,
form structures that are connected to the clique. Figure \ref{fig:clique subgraph}
presents the graph of a subset of the top 100 ASs per January 2006,
according to CAIDA ranking \citeapp{CAIDA}. It is visually clear
that the inter-connectivity of this subset is high. Indeed, the top
100 ASs graph density, which is the ratio between the number of links
present and the number of possible links, is 0.23, compared to a mean
$0.024\pm0.004$ for a random connected set of 100. It is important
to note that we were able to obtain similar results by ranking the
top ASs using topological measures, such as betweeness, closeness
and k-core analysis \citeapp{Meirom2013} .

Although, in principle, there are many structures the type-B players
(``minor players'') may form, the dynamics we considered indicates
the presence of stars and lines mainly (Sections \ref{sub:dynamical Results}
and \ref{sub: monetary Dynamics}). While the partition of ASs to
just two types is a simplification, we still expect our model to predict
fairly accurately the structures at the limits of high-importance
ASs and marginal ASs. A \emph{$k$-core} of a graph is the maximal
connected subgraph in which all nodes have degree of at least $k$.
The \emph{$k$-shell} is obtained after the removal of all the $k$-core
nodes. In Fig. \ref{fig:clique subgraph}, a snapshot of the sub-graph
of the marginal ASs is presented, using a $k$-core separation ($k=3$),
where all the nodes in the higher cores are removed. The abundance
of lines and stars is visually clear. In addition, the spanning tree
of this subset, which consists of 75\% of the ASs in the Internet,
is formed by removing just 0.02\% of the links in this sub-graph,
a strong indication for a forest-like structure.

In the dynamic aspect, we expect the type-A players sub-graph to converge
to a complete graph. We evaluate the mean node-to-node distance in
this subset as a function of time by using quarterly snapshots of
the AS graph from January, 2006 to October, 2008. Indeed, the mean
distance decreases approximately linearly. The result is presented
in Fig. \ref{fig:shell-core distance graph}. Also, the distance value
tends to $1$, indicating the almost-completeness of this sub-graph.

\begin{figure}
\centering{}\includegraphics[width=0.9\columnwidth]{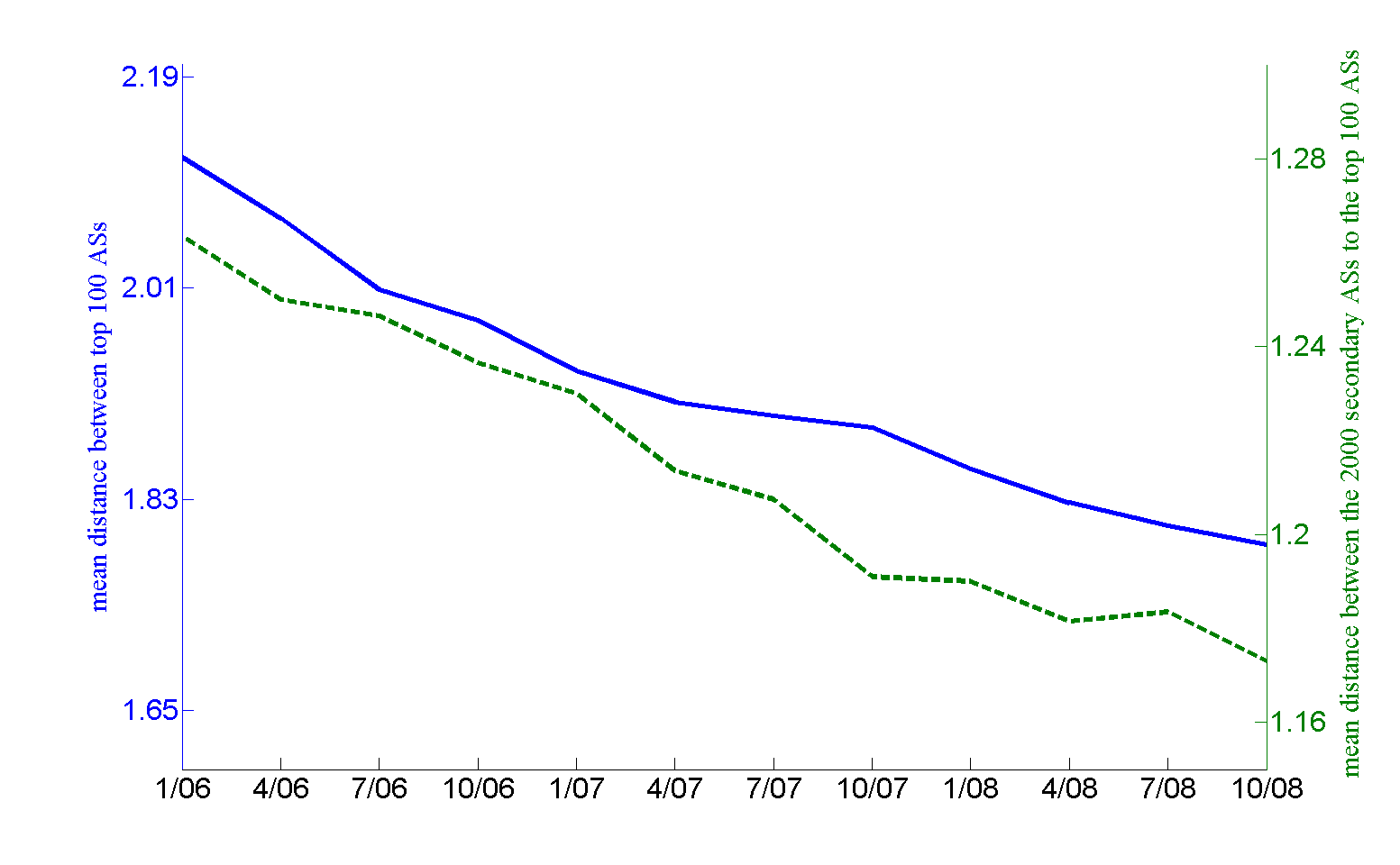}\caption{In solid blue: the mean distance of an AS in the top 100 ASs CAIDA
ranking to all the other top 100 ASs, from January, 2006 to October,
2008. In dashed green: The mean shortest distance of a secondary AS
(ranked 101-2100) from any top AS (ranked 1-100).}
\label{fig:shell-core distance graph}
\end{figure}

For a choice of core $\mathcal{C}$, the \emph{node-core distance}
of a node $i\notin\mathcal{C}$ is defined as the shortest path from
node $i$ to any node in the core. In Section \ref{sec:Monetary-transfers},
we showed that, by allowing monetary transfers, the maximal distance
of a type-B player to the type-A clique (the maximal ``node-core
distance'' in our model) depends inversely on the number of nodes
in the clique and the number of players in general. Likewise, we expect
the mean ``node-core'' distance to depend inversely on the number
of nodes in the clique. The number of ASs increases in time, and we
may assume the number of type-A players follows. Therefore, we expect
a decrease of the aforementioned mean ``node-core distance'' in
time. Fig \ref{fig:shell-core distance graph} shows the mean distance
of the secondary leading 2000 ASs, ranked 101-2100 in CAIDA ranking,
from the set of the top 100 nodes. The distance decreases in time,
in agreement with our model. Furthermore, our dynamics indicate that
the type-B nodes would be organized in stars, for which the mean ``node-core''
distance is close to two, and in singleton trees, for which the ``node-core''
distance is one. Indeed, as predicted, the mean ``node-core'' distance
proves to be between one and two.

It is widely assumed that the evolvement of the Internet follows a
``preferential attachment'' process \citeapp{Barabasi1999}. According
to this process, the probability that a new node will attach to an
existing node is proportional (in some models, up to some power) to
the existing node's degree. An immediate corollary is that the probability
that a new node will connect to any node in a set of nodes is proportional
to the set's sum of degrees. The sum of degrees of the secondary ASs
set is \textasciitilde{}1.9 greater than the sum of degrees in the
core, according to the examined data \citeapp{Gregori2011}. Therefore,
a ``preferential attachment'' class model predicts that a new node
is likely to attach to the shell rather than to the core. As all the
nodes in the shell have a distance of at least one from the core,
the new node's distance from the core will be at least two. Since
the initial mean ``shell-core'' distance is \textasciitilde{}1.26,
a model belonging to the ``preferential attachment'' class predicts
that the mean distance will be pushed to two, and in general increase
over time. However, this is contradicted by the data that shows (Fig
\ref{fig:shell-core distance graph}) a decrease of the aforementioned
distance. The slope of the latter has the 95\% confidence bound of
$(-3.1\cdot10^{-3},-2.3\cdot10^{-3}$) hops/month, a strong indication
of a negative trend, in disagreement with the ``preferential attachment''
model class. In contrast, this trend is predicted by our model, per
the discussion in Section \ref{sec:Monetary-transfers}. In fact,
if the Internet is described by a random, power law (``scale free'')
network, then the mean distance should grow as $\Theta(\log N)$ or
$\Theta(\log\log N)$ (\citeapp{Cohen2003}). However, experimental
observations shows that the mean distance grows slower than that (\citeapp{Pastor-Satorras2001}
), and it fact it may even be reduced with the network size, as predicted
by our model. 

% \bibliographystyleapp{acmsmall} 
% \bibliographyapp{fillers,Games_General,AS_structure,Games_Heterogenous}

\pagebreak{}

\section{Appendix : Detailed proofs}

\subsection{Basic model - Static Analysis}

In this section we discuss the properties of stable equilibria. Specifically,
we first establish that, under certain conditions, the major players
group together in a clique (section \ref{sub:The-type-A-clique-1}).
We then describe a few topological characteristics of all equilibria
(section \ref{sub:Pair-wise-equilibria-1}). 

As a metric for the quality of the solution we apply the commonly
used measure of the social cost, which is the sum of individual costs.
We evaluate the \emph{price of anarchy}, which is the ratio between
the social cost at the worst stable solution and its value at the
optimal solution, and the \emph{price of stability}, which is the
ratio between the social cost at the best stable solution and its
value at the optimal solution (section \ref{sub:PoA and PoS-1}).

\subsubsection{Preliminaries}

The next lemma will be useful in many instances. It measures the benefit
of connecting the two ends of a long line of players, as presented
in \ref{fig:long line-1}. If the line is too long, it is better for
both parties at its end to form a link between them.
\begin{lem}
\label{lem:shortcut benefit}Assume of lone line having $k$ nodes,
$(x_{1},x_{2},...x_{k}).$ By establishing the link $(x_{1},x_{k})$
the sum of distances $\sum_{i}d(x_{1},x_{i})$ $\left(\sum d(x_{k},x_{i})\right)$
is reduced by
\[
\frac{k\left(k-2\right)+mod(k,2)}{4}
\]
\end{lem}
\begin{proof}
Without the link $(x_{1},x_{k})$ the sum of distances is given by
the algebraic series
\[
\sum_{i}d(x_{1},x_{i})=\sum_{i=1}^{k-1}i=\frac{k\left(k-1\right)}{2}
\]

If $k$ is odd, than the the addition of the link $(x_{1},x_{k})$
we have 
\begin{eqnarray*}
\sum_{i}d(x_{1},x_{i}) & = & 2\sum_{i=1}^{\left\lfloor k/2\right\rfloor }i=\left(\left\lfloor k/2\right\rfloor +1\right)\left\lfloor k/2\right\rfloor \\
 & = & \left(\left(k-1\right)/2+1\right)\left(\left(k-1\right)/2\right)\\
 & = & \frac{k^{2}-1}{4}
\end{eqnarray*}

If $k$ is even, the corresponding sum is 
\begin{eqnarray*}
\sum_{i}d(x_{1},x_{i}) & = & \sum_{i=1}^{k/2}i+\sum_{i=1}^{k/2-1}i\\
 & = & \frac{\left(k/2+1\right)k}{4}+\frac{\left(k/2-1\right)k}{4}\\
 & = & \frac{k^{2}}{4}
\end{eqnarray*}

We conclude that the difference for $k$ even is 
\[
\frac{k^{2}}{4}-\frac{k}{2}=\frac{k\left(k-2\right)}{4}
\]

and for odd $k$ is
\[
\frac{k^{2}}{4}-\frac{k}{2}+\frac{1}{4}=\frac{k\left(k-2\right)+1}{4}
\]

\end{proof}

\subsubsection{\label{sub:The-type-A-clique-1}The type-A clique}

Our goal is understanding the resulting topology when we assume strategic
players and myopic dynamics. Obviously, if the link's cost is extremely
low, every player would establish links with all other players. The
resulting graph will be a clique. As the link's cost increase, it
becomes worthwhile to form direct links only with major players. In
this case, only the major players' subgraph is a clique. The first
observation leads to the following result.
\begin{lem}
If $c_{B}<1$ then the only stabilizable graph is a clique.
\end{lem}
In a clique $d(i,j)=1$ for all $i,j$. Assume $d(i,j)>1$. Then by
establishing a link $(i,j)$ the cost of both parties is reduced,
as each party reduces its distance to at least one player, and $c_{A}<c_{B}<1$.
Hence we can't have $d(i,j)>1$.

\begin{figure}
\centering{}\includegraphics[width=0.7\columnwidth]{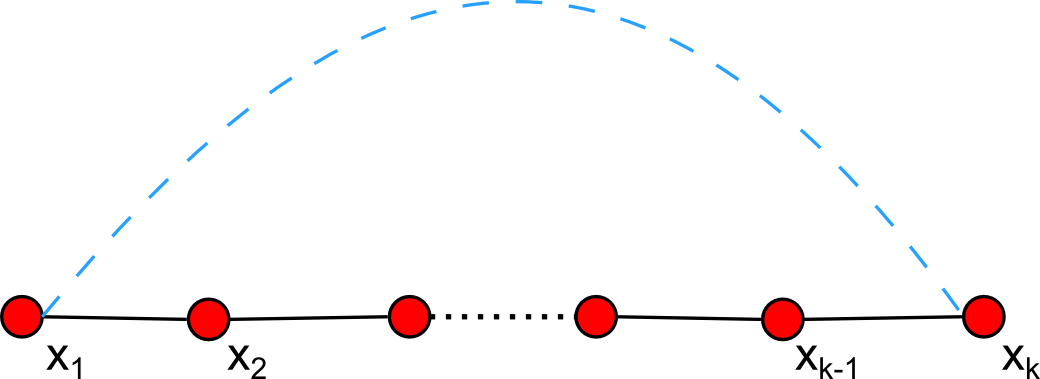}\caption{\label{fig:long line-1}The scenario described in Lemma \ref{lem:shortcut benefit}.
The additional link is dashed in blue.}
\end{figure}

In fact, we can use the same reasoning to generalize for $c>1$. If
two nodes are at a distance $L+1$ of each other, then there is a
path with $L$ nodes connecting them. By establishing a link with
cost $c$, we are shortening the distance between the end node to
$\sim L/2$ nodes that lay on the other side of the line. The average
reduction in distance is also $\approx L/2$, so by comparing $L^{2}\approx4c$
we obtain a bound on $L$, as follows:
\begin{lem}
\label{lem:The-longest-distance-2}The longest distance between any
node $i$ and node $j\in T_{B}$ is bounded by $2\sqrt{c_{B}}$. The
longest distance between nodes $i,j\in T_{A}$ is bounded by $\sqrt{(1-2A)^{2}+4c_{A}}-2\left(A-1\right)$.
In addition, if $c_{A}<A$ then there is a link between every two
type-A nodes.\end{lem}
\begin{proof}
We bound the maximal distance between two nodes by conisdering the
cost reduction of establishing a direct link between the two nodes
$i,j$ at the perimeter of a length $k$ line. We show that if the
line length is $\geq\left\lfloor 2\sqrt{c_{B}}\right\rfloor $ then
it is beneficial to establish such link. Assume $d(i,j)=k\geq\left\lfloor 2\sqrt{c_{B}}\right\rfloor >1$
and $i\in T_{B}$. Then there exist nodes $(x_{0}=i,x_{2},..x_{k-1}=j)$
such that $d(i,x_{\alpha})=\alpha$. By adding a link $(i,j)$ the
change in cost of node $i$, $\Delta C(i,E+ij)$ is, according to
lemma \ref{lem:shortcut benefit}, 
\begin{align*}
 & \Delta C(i,E+ij)\\
= & c_{B}-\sum_{\alpha=1}^{k-1}d(i,x_{\alpha})(1+\delta_{x_{\alpha},A}(A-1))\\
 & +\sum_{\alpha=1}^{k-1}d'(i,x_{\alpha})(1+\delta_{x_{\alpha},A}(A-1))\\
= & c_{B}-\sum_{\alpha=1}^{k-1}\left(d(i,x_{\alpha})-d'(i,x_{\alpha})\right)(1+\delta_{x_{\alpha},A}(A-1))\\
< & c_{B}-\sum_{\alpha=1}^{k-1}d(i,x_{\alpha})+\sum_{\alpha=1}^{k-1}d'(i,x_{\alpha})\\
< & c_{B}-\frac{k\left(k-2\right)+mod(k,2)}{4}\\
< & c_{B}-\frac{k\left(k-2\right)}{4}<0
\end{align*}
 where $d'(i,x_{\alpha})<d(i,x_{\alpha})$ is the distance after the
addition of the link $(i,j)$ and $\delta_{x_{\alpha},A}=1$ iff $x_{\alpha}\in T_{A}$.
Therefore, it is of the interest of player $i$ to add the link. 

Consider the case that $j\in T_{A}$ and 
\begin{equation}
d(i,j)=k-1\geq\sqrt{(1-2A)^{2}+4c_{B}}-2\left(A-1\right)>1\label{eq: distance-1}
\end{equation}
 The change in cost after the addition of the link $(i,j)$ is 
\begin{align*}
 & c_{B}-\sum_{\alpha=1}^{k-1}d(i,x_{\alpha})(1+\delta_{x_{\alpha},A}(A-1))\\
 & +\sum_{\alpha=1}^{k-1}d'(i,x_{\alpha})(1+\delta_{x_{\alpha},A}(A-1))\\
= & c_{B}-\sum_{\alpha=1}^{k-1}\left(d(i,x_{\alpha})-d'(i,x_{\alpha})\right)(1+\delta_{x_{\alpha},A}(A-1))\\
< & c_{B}-\sum_{\alpha=1}^{k-2}d(i,x_{\alpha})+\sum_{\alpha=1}^{k-2}d'(i,x_{\alpha})+A\left(d(i,x_{k})-d'(i,x_{k})\right)\\
< & c_{B}-\frac{k\left(k-2\right)+mod(k-1,2)}{4}-\left(k-1\right)A\\
 & +\left(k-1\right)+A-1\\
< & c_{B}-\frac{k\left(k-2\right)}{4}-\left(A-1\right)\left(k-2\right)\\
= & c_{B}-k^{2}/4-k(A-1)\\
< & 0
\end{align*}

Therefore it is beneficial for player $i$ to establish the link.
Similarly, if $i\in T_{A}$ then eq. \ref{eq: distance-1} is replaced
by $k-1\geq\sqrt{(1-2A)^{2}+4c_{A}}-2\left(A-1\right)>1$.

In particular, if we do not omit the $mod(k-1,2)$ term and set $k=3$
we get that if $2\sqrt{(-1+A)A+c_{A}}-2(A-1)<2$ the distance between
two type-A nodes is smaller than 2, in other words, they connected
by a link. The latter expression can be recast to the simple form
$c<A$.

Recall that if the cost of both parties is reduced (the change of
cost of node $j$ is obtained by the change of summation to $0..k-1$)
a link connecting them will be formed. Therefore, if $i,j\in T_{B}$
then maximal distance between then is $d(i,j)\leq max\{2\sqrt{c_{B}},1\}$
as otherwise it would be beneficial for both $i,j$ to establish a
link that will reduce their mutual distance to $1.$ Likewise, if
$i,j\in T_{A}$ then 
\[
d(i,j)\leq\sqrt{(1-2A)^{2}+4c_{A}}-2\left(A-1\right)
\]
using an analogous reasoning. If $i\in T_{A}$ and $j\in T_{B}$ then
it'll be worthy for player $i$ to establish the link only if $d(i,j)\geq\left\lfloor 2\sqrt{c_{B}}\right\rfloor $.
In this case it'll be also worthy for player $j$ to establish the
link since 
\[
d(i,j)\geq\left\lfloor 2\sqrt{c_{B}}\right\rfloor \geq\sqrt{(1-2A)^{2}+4c_{B}}-2\left(A-1\right)
\]
and the link will be established. Notice however that if 
\[
\left\lfloor 2\sqrt{c_{B}}\right\rfloor \geq d(i,j)\geq\sqrt{(1-2A)^{2}+4c_{A}}-2\left(A-1\right)
\]
 then although it is worthy for player $j$ to establish the link,
it is isn't worthy for player $i$ to do so and the link won't be
established. This concludes our proof.
\end{proof}
Lemma \ref{lem:The-longest-distance-2} indicates that if $1<c_{A}<A$
then the type $A$ nodes will form a clique (the ``nucleolus'' of
the network). The type $B$ nodes form structures that are connected
to the type $A$ clique (the network nucleolus). These structures
are not necessarily trees and will not necessarily connect to a single
point of the type-A clique only. This is indeed a very realistic scenario,
found in many configurations.

If $c_{A}>A$ then the type-A clique is no longer stable. This setting
does not correspond to the observed nature of the inter-AS topology
and we shall focus in all the following sections on the case $1<c_{A}<A$.
Nevertheless, as a flavor of the results for $c_{A}>A$ we present
the following proposition, which is stated for general heterogeneity
of players, rather than a dichotomy of types. Here, we denote that
the relative importance of player $j$ in player $i$'s concern is
as $A_{ij}$.
\begin{prop}
\label{cor:general optimality-1}Assume the cost function of player
$i$ is given by the form
\[
C(i)\triangleq deg(i)\cdot c+\sum_{j\in T_{A}}A_{ij}d(i,j)
\]

where either $c>A_{ij}$ or $c>A_{ji}$. Then a star is a stable formation.
Furthermore, if $A_{ij}=A_{j}$ define $j^{*}$ as the node for which
$A_{j^{*}}$is maximal. Then a star with node $j^{*}$ at its center
is the optimal stable structure in terms of social utility. \end{prop}
\begin{proof}
Clearly, it is not worthy for player either player $i$ or player
$j$ to reduce their distance from 2 to the 1 since either 
\[
\Delta C(i,E+ij)=c-A_{ij}>0
\]

or 
\[
\Delta C(j,E+ij)=c-A_{ji}>0
\]

and the link $(i,j)$ will not be established. It is also not possible
to remove any links without disconnecting the network. This proofs
the stability of the star.

Regarding the optimality of the network structure, a player must have
be connected to at least one node in order to be connected to the
network. With no additional links, the minimal distance to all other
nodes is 2 and the discussion before indicates it is not beneficial
to add extra links to reduce the distance to only one node. The social
cost of a star with $x_{0}$ at its center is 
\begin{eqnarray*}
\sum C(i) & = & c\left(N-1\right)+2(N-1)\sum_{x\neq x_{0}}A_{x}\\
 &  & +(N-1)A_{x_{0}}+\sum_{x\neq x_{0}}A_{x}\\
 & = & c\left(N-1\right)+2(N-1)\sum_{x}A_{x}\\
 &  & -(N-1)A_{x_{0}}-\sum_{x\neq x_{0}}A_{x}
\end{eqnarray*}

where $N$ is the number of players, $d(x,x')=2$ for all $x,x'\neq x_{0}$
and $d(x,x_{0})=1$. The first two terms are constants. In order to
minimize the latter expression, one needs to maximize
\[
(N-1)A_{x_{0}}+\sum_{x\neq x_{0}}A_{x}=(N-2)A_{x_{0}}+\sum_{x}A_{x}
\]

the latter is clearly maximized by choosing $A_{x_{0}}$as maximal.
Hence, the optimal star is a star with $x_{0}$at its center.

Assume the optimal stable structure is not a star. Then, there is
at least three nodes $i,y_{1},y_{2}$ such that $d(i,y_{1})=d(i,y_{2})=d(y_{1},y_{2})=1$
as in the star configuration the link's term is minimal and $d(x,x')=2$
for all $x,x'\neq x_{0}$ and $d(x,x_{0})=1$. However, the above
discussion shows the annihilation of at least one of the links of
the clique $(y_{1},y_{2},i$) is beneficial for at least one of the
players and this structure would not be stable.
\end{proof}

\subsubsection{\label{sub:Pair-wise-equilibria-1}Equilibria's properties}

Here we describe common properties of all pair-wise equilibria. We
start by noting that, unlike the findings of several other studies
\citet{Arcaute2013,5173479,Fabrikant2003,NisanN.RoughgardenT.TardosE.2007},
in our model, at equilibrium, the type-B nodes are not necessarily
organized in trees. This is shown in the next example.

\begin{figure}
\centering{}\includegraphics[width=0.9\columnwidth]{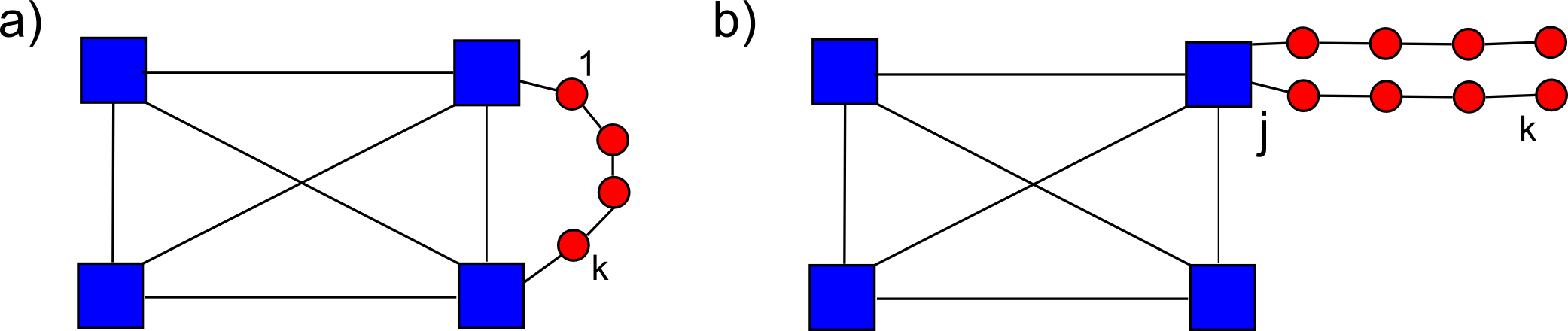}\caption{\label{fig:loop example-2}Non optimal networks. The type-A clique
is in blue squares, the type-B players are in red circles. a) The
network described in Example \ref{example-1}. b) \label{fig:A-poor-equilibrium-2}A
poor equilibrium, as described in the appendix.}
\end{figure}

\begin{example}
\label{example-1}Assume for simplicity that $c_{A}=c_{B}=c$. Consider
a line of length $k$ of type B nodes, $(1,2,3...,k)$ such that $\sqrt{8c}>k+1>\sqrt{2c}$
or equivalently $\left(k+1\right)^{2}<8c<4\left(k+1\right)^{2}$ .
In addition, the links $(j_{1},1)$ and $(j_{2},k)$ exist, where
$j\in T_{A}$, i.e., the line is connected at both ends to different
nodes of the type-A clique, as depicted in Fig \ref{fig:loop example-2}.
We show in the appendix that this is a stabilizable graph.

We now show that this structure is stabilizable. For simplicity, assume
$mod(k-1,4)=0$ ($k$ is odd and $\frac{k-1}{2}$ is odd).

\begin{figure}
\centering{}\includegraphics[width=0.8\columnwidth]{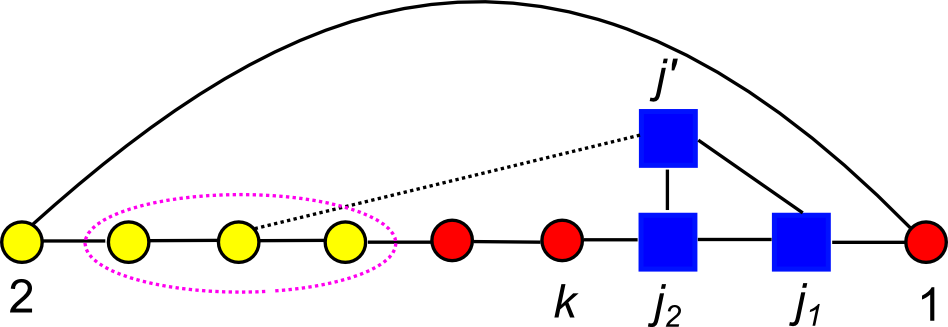}\caption{\label{fig:example-1}A line of $k=7$ nodes. By removing the link
$(1,2)$ only the distances from player 1 to the yellow players are
affected. By establishing the link $(j,4)$ only the distances from
player $j'$ to the players encircled by the purple dashed ellipse
are affected.}
\end{figure}

Any link removal $(x_{1},x_{2})$ in the circle $(j_{1},1...k,j_{2},j_{1})$
will result in a line with nodes $x_{1}$ and $x_{2}$ at its ends
(Fig. \ref{fig:example-1}). The type-B players that have the most
incentive to disconnect a link are either node $1$ or node $k$,
as the type-A nodes will be closest to either one of them after the
link removal (at distances one and two hops, Fig. \ref{fig:example-1}).
Therefore, if players $1$ or $k$ would not deviate, no type-B player
will deviate as well.

W.l.o.g, we discuss node 1. Since $c_{B}<A$, it is not beneficial
for it to disconnect the link $(j_{1},1)$. Assume the link $(1,2)$
is removed. A simple geometric observation shows that the distance
to nodes $\{2,..,\frac{k+3}{2}\}$ is affected, while the distance
to all the other nodes remains intact (Fig. \ref{fig:example-1}).
The mean increase in distance is $\frac{k+1}{2}$ and the number of
affected nodes is $\frac{k+1}{2}$ . However, 
\[
\Delta C(1,E-12)=c-\frac{(k+1)}{2}^{2}<0
\]
 and player $1$ would prefer the link to remain. The same calculation
shows that it is not beneficial for player $j_{1}$ to disconnect
$(j_{1},1)$.

Clearly, if it not beneficial for $j\in T_{A},\, j\neq j_{1},j_{2}$
to establish an additional link to a type-B player then it is not
beneficial to do so for $j_{1}$ or $j_{2}$ as well. The optimal
additional link connecting $j$ and a type-B player is $\mathcal{E}=(j,\frac{k+1}{2})$,
that is, a link to the middle of the ring (Fig. \ref{fig:example-1}).
A similar geometric observation shows that by establishing this link,
only the distances to nodes $\left\{ \frac{k-1}{4},...,\frac{3k+1}{4}\right\} $
are affected (Fig. \ref{fig:example-1}). The reduction in cost is
\[
\Delta C(j,E+\mathcal{E})=c-\frac{(k+1)}{8}^{2}>0
\]

and it is not beneficial to establish the link. 

In order to complete the proof, we need to show that no additional
type-B to type-B links will be formed. By establishing such link,
the distance of at least one of the parties to the type-A clique is
unaffected. The previous calculation shows that by adding such link
the maximal reduction of cost due to shortening the distance to type-B
players is bounded from above by $\frac{(k+1)}{8}^{2}$. Therefore,
as before, no additional type-B to type-B links will be formed.

This completes the proof that this structure is stabilizable.
\end{example}
Next, we bound from below the number of equilibria. For simplicity,
we discuss the case where $c_{A}=c_{B}=c$. We accomplish that by
considering the number of equilibria where the type-B players are
organized in a forest (multiple trees) and the allowed forest topologies.
The following lemma restricts the possible sets of trees in an equilibrium.
Intuitively, this lemma states that we can not have two ``heavy''
trees, ``heavy'' meaning that there is a deep sub-tree with many
nodes, as it would be beneficial to make a shortcut between the two
sub-trees.

\begin{figure}
\centering{}\includegraphics[width=0.6\columnwidth]{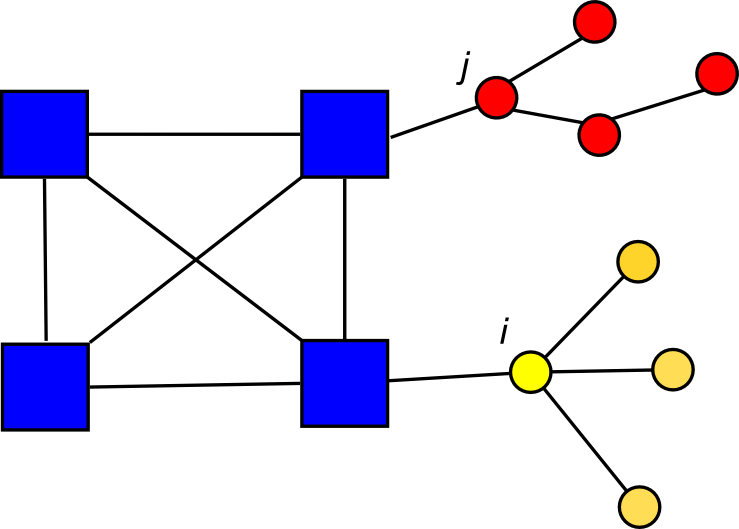}\caption{\label{fig:forest-1}Node $j$ is on the third level of the tree of
formed by starting a BFS from node $i$, as discussed in lemma \ref{lem:subtree-1}.
The forest of type-B nodes is composed of two trees, in yellow and
red (lemma \ref{lem:forest sub-graph-1}). Their roots are $i$ and
$j,$ correspondingly. The maximal depth in this forest is three.}
\end{figure}

\begin{lem}
\label{lem:subtree-1}Assume $c_{A}=c_{B}=c$. Consider the BFS tree
formed starting from node $i.$ Assume that node $j$ is $k$ levels
deep in this tree. Denote the sub-tree of node $j$ in this tree by
$\mathcal{T}_{i}(j)$ (Fig. \ref{fig:forest-1}) In a link stable
equilibrium, the number of nodes in sub-trees satisfy either \textup{$|\mathcal{T}_{i}(j)|<c/k$
}\textup{\emph{or }}\textup{$|\mathcal{T}_{j}(i)|<c/k$.}\end{lem}
\begin{proof}
Assume $|\mathcal{T}_{i}(j)|>c/k$. Consider the change in cost of
player $i$ after the addition of the link $(i,j)$
\begin{eqnarray*}
 &  & \Delta C(i,E+ij)\\
 & = & c+\sum_{x_{\alpha}\in\mathcal{T}i(j)}d'(i,x_{\alpha})(1+\delta_{x_{\alpha},A}(A-1))\\
 &  & +\sum_{x_{\alpha}\notin\mathcal{T}_{i}(j)}d'(i,x_{\alpha})(1+\delta_{x_{\alpha},A}(A-1))\\
 &  & -\sum_{x_{\alpha}\in\mathcal{T}_{i}(j)}d(i,x_{\alpha})(1+\delta_{x_{\alpha},A}(A-1))\\
 &  & -\sum_{x_{\alpha}\notin\mathcal{T}_{i}(j)}d(i,x_{\alpha})(1+\delta_{x_{\alpha},A}(A-1))\\
 & = & c+\sum_{x_{\alpha}\in\mathcal{T}_{i}(j)}\left(d'(i,x_{\alpha})-d(i,x_{\alpha})\right)(1+\delta_{x_{\alpha},A}(A-1))\\
 &  & +\sum_{x_{\alpha}\notin\mathcal{T}_{i}(j)}\left(d'(i,x_{\alpha})-d(i,x_{\alpha})\right)(1+\delta_{x_{\alpha},A}(A-1))\\
 & < & c+\sum_{x_{\alpha}\in\mathcal{T}_{i}(j)}\left(d'(i,x_{\alpha})-d(i,x_{\alpha})\right)\\
 & = & c-k|T_{i}(j,k)|\\
 & < & 0
\end{eqnarray*}

since the distance was shorten by $k$ for every node in the sub-tree
of $j$. 

Therefore, it is beneficial for player $i$ to establish the link.
Likewise, if $|\mathcal{T}_{j}(i)|<c/k$ then it would be beneficial
for player $j$ to establish the link and the link will be established.
Hence, one of the conditions must be violated.
\end{proof}
The following lemma considers the structure of the type-B players'
sub-graph. It builds on the results of lemma \ref{lem:subtree-1}
to reinforce the restrictions on trees, showing that trees must be
shallow and small. 
\begin{lem}
\label{lem:forest sub-graph-1}Assume $c_{A}=c_{B}=c$. If the sub-graph
of type-B nodes is a forest (Fig. \ref{fig:forest-1}), then there
is at most one tree with depth greater than $\sqrt{c/2}$ and there
is at most one tree with more than $c/2$ nodes. The maximal depth
of a tree in the forest is $\sqrt{2c}-1$. Every type-B forest in
which every tree has a maximal depth of $\min\left\{ \sqrt{c/2},\sqrt{c}-3\right\} $
and at most $\min\left\{ c/2,\sqrt{c}\right\} $ nodes is stabilizable.\end{lem}
\begin{proof}
Assume there are two trees $S_{1},S_{2}$ that have depth greater
than $\sqrt{c/2}$ . The distance between the nodes at the lowest
level is greater than $2\sqrt{c/2}+1$ as the trees are connected
by at least one node in the type-A clique, $d(i,j)\geq2$ (Fig. \ref{fig:forest-1}).
This contradicts with Lemma \ref{lem:The-longest-distance-2}. 

Assume there are two trees $S_{1},S_{2}$ with roots $i,j$ that have
more than $c/2$ nodes. In the BFS tree that is started from node
$i$ node $j$ is at least in the second level (as they are connected
by at least one node in the type-A clique). This contradicts with
Lemma \ref{lem:subtree-1}.

Finally, following the footsteps of Lemma \ref{lem:subtree-1} proof,
consider two trees $\mathcal{T}_{i}$ and $\mathcal{T}_{j}$, with
corresponding roots $i$ and $j$ (i.e., nodes $i$ and $j$ have
a direct link with the type-A clique). Consider a link $(x,y)$, where
$x\in\mathcal{T}_{i}$ and $y\in\mathcal{T}_{j}$. At least one of
them does not reduce its distance to the type-A clique by establishing
this link. W.l.o.g, we'll assume this is true for player $x$. Therefore,
\begin{eqnarray*}
 &  & \Delta C(x,E+xy)\\
 & = & c+\sum_{z\in\mathcal{T}_{j}}\left(d'(z,x)-d(z,x)\right)\\
 & > & 0
\end{eqnarray*}

as the maximal reduction in distance is $\sqrt{c}$ and the maximal
number of nodes in $\mathcal{T}_{j}$ is also $\sqrt{c}$. Therefore,
it is not beneficial for player $x$ to establish this link, and the
proof is completed.
\end{proof}
Finally, the next proposition provides a lower bound on the number
of link-stable equilibria by a product of $|T_{A}|$ and a polynomial
with a high degree ($\approx2^{\sqrt{c}}$) in $|T_{B}|$.
\begin{prop}
Assume $c_{A}=c_{B}=c$. The number of link-stable equilibria in which
the sub-graph of $T_{B}$ is a forest is at least $o(|T_{A}||T_{B}|^{N_{c}})$,
where $N_{c}=o(2^{\frac{\sqrt{c}}{2}}/\sqrt{c})$ is a function of
$c$ only. Therefore, the number of link-stable equilibria is at least
\textup{$o(|T_{A}||T_{B}|^{N_{c}}).$}\end{prop}
\begin{proof}
For simplicity, we consider the case where $c\gg1$ and count the
number of different forests that are composed of trees up to depth
$\sqrt{c/2}$ and exactly $\sqrt{c}$ nodes. Let's define the number
of different trees by $N_{c}$. Note that $N_{c}$ is independent
of $|T_{B}|.$ The number of different forests of this type is bound
from below by the expression $\left(\begin{array}{c}
|T_{B}|+N_{c}\\
N_{c}
\end{array}\right)$. Using Striling's approximation $\left(\begin{array}{c}
|T_{B}|+N_{c}\\
N_{c}
\end{array}\right)=o(|T_{B}|^{N_{c}})$. The number $N_{C}$ can be bounded in a similar fashion by $\left(\begin{array}{c}
\left\lfloor \sqrt{c}\right\rfloor \\
\left\lfloor \sqrt{c/2}\right\rfloor 
\end{array}\right)=o(2^{\frac{\sqrt{c}}{2}}/\sqrt{c})$, which is the number of trees with $\sqrt{c}$ elements, depth $\sqrt{c/2}$
and only one non-leaf node at each level of tree. 

Each tree can be connected to either one of the type-A nodes, and
therefore the number of possible configurations is at least $o(|T_{A}||T_{B}|^{N_{c}})$
. 
\end{proof}
To sum up, while there are many equilibria, in all of them nodes cannot
be too far apart, i.e., a small-world property. Furthermore, the trees
formed are shallow and are not composed of many nodes.

\subsubsection{\label{sub:PoA and PoS-1}Price of Anarchy \& Price of Stability }

As there are many possible link-stable equilibria, a discussion of
the price of anarchy is in place. First, we explicitly find the optimal
configuration. Although we establish a general expression for this
configuration, it is worthy to also consider the limiting case of
a large network, $|T_{B}|\gg1,|T_{A}|\gg1$. Moreover, typically,
the number of major league players is much smaller than the other
players, hence we also consider the limit $|T_{B}|\gg|T_{A}|\gg1$.

\begin{figure}
\centering{}\includegraphics[width=0.9\columnwidth]{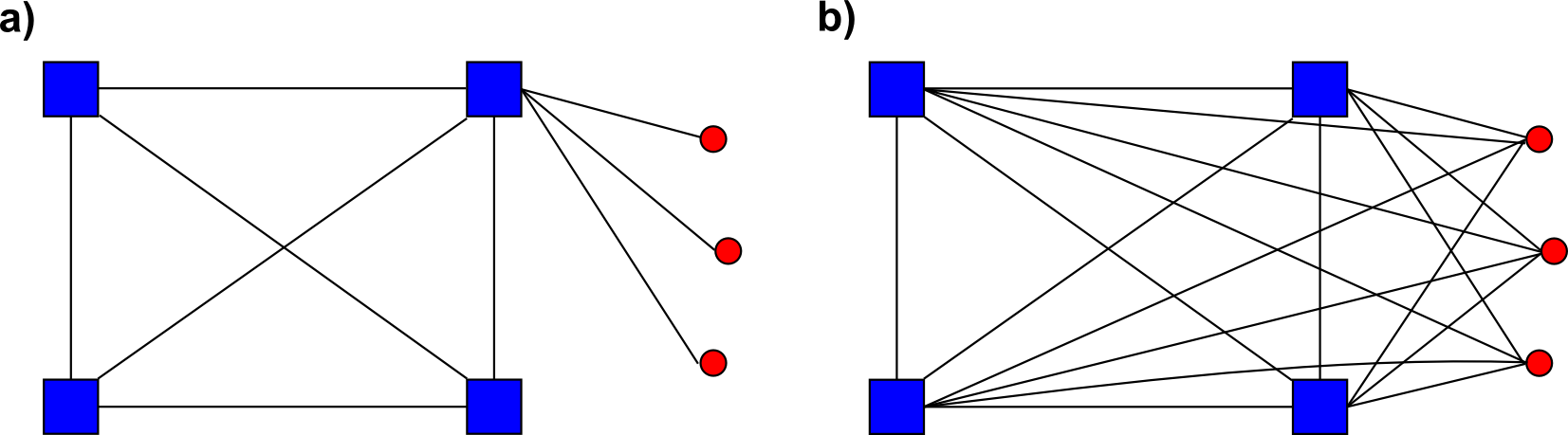}\caption{\label{fig:The-optimal-solution-1} \label{fig:The-optimal-state - monetary transfers-1}The
optimal solution, as described in Lemma \ref{lem:optimal solution-1}.
If $\left(A+1\right)/2<c$ the optimal solution is described by a),
otherwise by b). When monetary transfers (section \ref{sec:Monetary-transfers})
are allowed, both configurations are stabilizable. Otherwise, only
a) is stabilizable.}
\end{figure}

\begin{prop}
\label{lem:optimal solution-1}Consider the network where the type
$B$ nodes are connected to a specific node $j\in T_{A}$ of the type-A
clique. The social cost in this stabilizable network (Fig. \ref{fig:The-optimal-solution-1}(a))
is 
\begin{eqnarray*}
\mathcal{S} & = & 2|T_{B}|\left(|T_{B}|-1+c+\left(A+1\right)(|T_{A}|-1/2)\right)+|T_{A}|\left(|T_{A}|-1\right)\left(c_{A}+A\right).
\end{eqnarray*}

Furthermore, if $|T_{B}|\gg1,|T_{A}|\gg1$ then, omitting linear terms
in $|T_{B}|,|T_{A}|$, 
\[
\mathcal{S}=2|T_{B}|(|T_{B}|+\left(A+1\right)|T_{A}|)+|T_{A}|^{2}\left(c+A\right).
\]
Moreover, if $\frac{A+1}{2}\le c$ then this network structure is
socially optimal and the price of stability is $1$, otherwise the
price of stability is
\[
PoS=\frac{2|T_{B}|(|T_{B}|+\left(A+1\right)|T_{A}|)+|T_{A}|^{2}\left(c_{A}+A\right)}{2|T_{B}|\left(|T_{B}|+\left(\frac{A+1}{2}+c\right)|T_{A}|\right)+|T_{A}|^{2}\left(c_{A}+A\right)}.
\]

Finally, if \textup{$|T_{B}|\gg|T_{A}|\gg1$, }\textup{\emph{then
the price of stability is asymptotically $1$.}}
\end{prop}
\begin{proof}
This structure is immune to removal of links as a disconnection of
a $(type-B,type-A)$ link will disconnect the type-B node, and the
type-A clique is stable (lemma \ref{lem:The-longest-distance-2}).
For every player $j$ and $i\in T_{B}$, any additional link $(i,j)$
will result in $\Delta C(j,E+ij)\geq c_{B}-1>0$ since the link only
reduces the distance $d(i,j)$ from 2 to 1. Hence, player $j$ has
no incentive to accept this link and no additional links will be formed.
This concludes the stability proof. 

We now turn to discuss the optimality of this network structure. First,
consider a set of type-A players. Every link reduce the distance of
at least two nodes by at least one, hence the social cost change by
introducing a link is negative, since $2c_{A}-2A<0$. Therefore, in
any optimal configuration the type-A nodes form a complete graph.
The other terms in the social cost are due to the inter-connectivity
of type-B nodes and the type-A to type-B connections. As $deg(i)=1$
for all $i\in T_{B}$ the cost due to link's prices is minimal. Furthermore,
$d(i,j)=1$ and the distance cost to node $j$ (of type A) is minimal
as well. For all other nodes $j'$, $d(i,j')=2$. 

Assume this configuration is not optimal. Then there is a \emph{topologically
different} configuration in which there exists an additional node
$j'\in T_{A}$ for which $d(i,j')=1$ for some $i\in T_{B}$. Hence,
there's an additional link $(i,j)$. The social cost change is $2c+2+\delta_{x_{\alpha},A}(A-1)$
. Therefore, if $\frac{A+1}{2}\le c$ this link reduces the social
cost. On the other hand, if $\frac{A+1}{2}>c$ every link connecting
a type-B player to a type-A player improves the social cost, although
the previous discussion show these link are unstable. In this case,
the optimal configuration is where all type-B nodes are connected
to all the type-A players, but there are no links linking type-B players.
This concludes the optimality proof.

The cost due to inter-connectivity of type A nodes is

\[
c_{A}|T_{A}|\left(|T_{A}|-1\right)+A|T_{A}|\left(|T_{A}|-1\right)=|T_{A}|\left(|T_{A}|-1\right)\left(c_{A}+A\right).
\]

The first expression is due to the cost of $|T_{A}|$ clique's links
and the second is due to distance (=1) between each type-A node. The
distance of each type B nodes to all the other nodes is exactly 2,
except to node $j$, to which its distance is 1. Therefore the social
cost due to type B nodes is 
\begin{align*}
 & 2|T_{B}|(|T_{B}|-1)+2c_{B}|T_{B}|+2\left(A+1\right)|T_{B}|\left(\left(|T_{A}|-1\right)+\left(A+1\right)+2\left(A+1\right)(|T_{A}|-1)\right)\\
= & 2|T_{B}|\left(|T_{B}|-1+c_{B}+\left(A+1\right)(|T_{A}|-1/2)\right).
\end{align*}

The terms on the left hand side are due to (from left to right) the
distance between nodes of type B, the cost of each type-B's single
link, the cost of type-B nodes due to the distance (=2) to all member
of the type-A clique bar $j$ and the cost of type $B$ nodes due
to the distance (=1) to node $j$. The social cost is 
\begin{eqnarray*}
\sum C(i) & = & 2|T_{B}|\left(|T_{B}|-1+c+\left(A+1\right)(|T_{A}|-1/2)\right)+|T_{A}|\left(|T_{A}|-1\right)\left(c_{A}+A\right).
\end{eqnarray*}

To complete the proof, note that if $\frac{A+1}{2}>c$ the latter
term in the social cost of the optimal (and unstable) solution is
\[
2|T_{B}|(|T_{B}|-1)+2c|T_{B}|\left(1+|T_{A}|\right)+\left(A+1\right)|T_{B}||T_{A}|=2|T_{B}|\left(|T_{B}|-1+\left(\frac{A+1}{2}+c\right)|T_{A}|\right).
\]

As the number of links is $|T_{B}|\left(1+|T_{A}|\right)$ and the
distance of type-B to type-A nodes is 1. The optimal social case is
then
\[
2|T_{B}|\left(|T_{B}|-1+\left(\frac{A+1}{2}+c\right)|T_{A}|\right)+|T_{A}|\left(|T_{A}|-1\right)\left(c_{A}+A\right).
\]

Considering all quantities in the limit $|T_{B}|\gg1,|T_{A}|\gg1$
completes the proof.\end{proof}

Next, we evaluate the price of anarchy. The social cost in the stabilizable
topology presented in Fig \ref{fig:A-poor-equilibrium-2}, composed
of a type-A clique and long lines of type-B players, is calculated
in the appendix. The ratio between this value and the optimal social
cost constitutes a lower bound on the price of anarchy. An upper bound
is obtained by examining the social cost in any topology that satisfies
Lemma \ref{lem:The-longest-distance-2}. The result in the large network
limit is presented by the following proposition.

Next, we evaluate the price of anarchy. In order to do that, we use
the following two lemmas. The first lemma evaluates a lower bound
by considering the social cost in the stabilizable topology presented
in Fig \ref{fig:A-poor-equilibrium-2}(b), composed of a type-A clique
and long lines of type-B players. Later on, an upper bound is obtained
by examining the social cost in any topology that satisfies Lemma
\ref{lem:The-longest-distance-2}. 

For simplicity, in the following lemma we assume that $|T_{B}|=\min\left\{ \left\lfloor \sqrt{4c_{A}}\right\rfloor ,\left\lfloor \sqrt{4c_{B}/5}\right\rfloor \right\} m$
where $m\in\mathcal{N}$ 
\begin{lem}
\label{lem:poor eq}Consider the network where the type B nodes are
composed of $m$ long lines of length $k=\min\left\{ \left\lfloor \sqrt{3c_{A}}\right\rfloor ,\left\lfloor \sqrt{4c_{B}/5}\right\rfloor \right\} $,
and all the lines are connected at $j\in T_{A}$. The total cost in
this stabilizable network is 
\begin{eqnarray*}
S & = & |T_{A}|\left(|T_{A}|-1\right)\left(c_{A}/2+A\right)+2c_{B}|T_{B}|+\left(A+1\right)|T_{B}|\left(|T_{A}|-1\right)\left(k+3\right)/2\\
 &  & +|T_{B}|\left(\left(A+1\right)\left(k+1\right)/2+2k-4\right)+2|T_{B}|^{2}(k+2)^{2}-2m
\end{eqnarray*}

if $|T_{B}|\gg1,|T_{A}|\gg1$ then 
\begin{eqnarray*}
S & = & \left(A+1\right)|T_{A}|\left|T_{B}\right|o(c)\\
 &  & +|T_{A}|^{2}\left(c+A\right)+\left|T_{B}\right|^{2}o(c)
\end{eqnarray*}
\end{lem}
\begin{proof}
For simplicity, we assume $c_{B}\leq20c_{A}$. First, for the same
reason as in Prop. \ref{lem:optimal solution-1}, this network structure
is immune to removal of links. Consider $j'\in T_{A},j'\neq j$. Let
us observe the chain $(j,1,2,3,...k)$ where $1,2,3...k\in T_{B}$. 

By establishing a link to some node $k\geq x\geq1$, the change in
cost of player $j'$ is
\[
\Delta C(j',E+j'k)=c_{A}-\sum_{i=1}^{k}\left(d(i,j')-d'(i,j')\right).
\]

Note that this distance is non-zero only for a node $i$ that satisify
\[
i>(x-i)+1
\]

or $i<\left(x+1\right)/2$. The maximal reduction in distance is bounded
by noting that the optimal link is to node $x$ such that $2k/3+1\geq x\geq2k/3-1$.
Hence,
\begin{eqnarray*}
\sum_{i=\left\lceil \left(x+1\right)/2\right\rceil }^{k}\left(d(i,j')-d'(i,j')\right) & \leq & \frac{\left(\left\lceil \left(x+1\right)/2\right\rceil +2+k\right)\left(k-\left\lceil \left(x+1\right)/2\right\rceil \right)}{2}-2\sum_{i=1}^{k-x-1}i\\
 & \leq & \frac{\left(x/2+1+k\right)\left(k-x/2\right)}{2}-(k-x)\left(k-x-2\right)\\
 & \leq & \frac{\left(4k/3+1\right)\left(2k/3+1\right)}{2}-(k-2k/3-1)\left(k-2k/3-3\right)\\
 & \leq & 4k^{2}/9+4k/3+1-(k/3-1)\left(k/3-3\right)\\
 & \leq & 4k^{2}/9+4k/3+1-k^{2}/9-4k/3-4\\
 & \leq & k^{2}/3-3
\end{eqnarray*}

and we have, 
\begin{eqnarray*}
\Delta C(j',E+j'k) & = & c_{A}-k^{2}/3+3\\
 & > & c_{A}-\frac{k^{2}}{3}>0
\end{eqnarray*}
by lemma \ref{lem:shortcut benefit} and noting the distance to player
$j$ is unaffected. Therefore there is no incentive for player $j'$
to add the link $(j',k)$. The same calculation indicates that no
additional link $(i,i')$ will be formed between two nodes on the
same line.

Consider two lines of length $k$, $(j,x_{1},x_{2},...x_{k})$ and
$(j,y_{1},y_{2},...y_{k})$. Consider the addition of the link $(x_{k},y_{\left\lfloor k+1/2\right\rfloor })$.
That is, the addition of a link from an end of one line to the middle
player on another line. This link is optimal in $x_{k}'$s concern,
as it minimizes the sum of distances from it to players on the other
line. The change in cost is (see appendix) 
\[
\Delta C(x_{k},E+x_{k}y_{\left\lfloor k+1/2\right\rfloor })=c_{B}-5k^{2}/4>0.
\]

Note that 
\begin{eqnarray*}
\Delta C(x_{k},E+x_{k}y_{\left\lfloor k+1/2\right\rfloor }) & \leq & \Delta C(x_{k},E+x_{k}y_{i})
\end{eqnarray*}
for any other $i=1..k$, since a player gains the most from establishing
a link to the another line is the player that is furthest the most
from that line, i.e., the player at the end of the line. This concludes
the stability proof.

The total cost due to type-B nodes is 
\begin{align*}
 & 2c_{B}|T_{B}|+\left(A+1\right)|T_{B}|\left(|T_{A}|-1\right)\left(k+3\right)/2\\
 & +|T_{B}|\left(A+1\right)\left(k+1\right)/2+2m(k-1)^{2}+2(m-1)mk^{2}(k+2)^{2}.
\end{align*}
The terms represent (from left to right) the links' cost, the cost
due to the type-B players' distances to the type-A clique's nodes
(except $j$), the cost due to the type-B players' distances from
node $j$, the cost due to intra-line distances and the cost due inter-lines
distances. By using the relation $|T_{B}|=km$ we have
\begin{align*}
=2 & c_{B}|T_{B}|+\left(A+1\right)|T_{B}|\left(|T_{A}|-1\right)\left(k+3\right)/2\\
 & +|T_{B}|\left(A+1\right)\left(k+1\right)/2+2\left(|T_{B}|-m\right)(k-1)+2|T_{B}|^{2}(k+2)^{2}-2|T_{B}|m(k+2)\\
= & 2c_{B}|T_{B}|+\left(A+1\right)|T_{B}|\left(|T_{A}|-1\right)\left(k+3\right)/2\\
 & +|T_{B}|\left(\left(A+1\right)\left(k+1\right)/2+2k-4\right)+2|T_{B}|^{2}(k+2)^{2}-2m.
\end{align*}

The total cost due to type-A is as before
\[
|T_{A}|\left(|T_{A}|-1\right)\left(c_{A}+A\right).
\]

Therefore, 
\begin{eqnarray*}
\sum C(i) & = & |T_{A}|\left(|T_{A}|-1\right)\left(c_{A}/2+A\right)+2c_{B}|T_{B}|+\left(A+1\right)|T_{B}|\left(|T_{A}|-1\right)\left(k+3\right)/2\\
 &  & +|T_{B}|\left(\left(A+1\right)\left(k+1\right)/2+2k-4\right)+2|T_{B}|^{2}(k+2)^{2}-2m\\
 & \rightarrow & |T_{A}|^{2}\left(c_{A}+A\right)+\left(A+1\right)|T_{A}|\left|T_{B}\right|o(k)+\left|T_{B}\right|^{2}o(k).
\end{eqnarray*}
 
\end{proof}
The most prevalent situation is when $|T_{B}|\gg|T_{A}|\gg1$. In
this case we can bound the price of anarchy to be at least $o(k^{2})=o(c_{B})$. 

The next lemma bounds the price of anarchy from above by bounding
the maximal total cost in the a link-stable equilibrium.
\begin{lem}
\label{lem:The-worst-social utility-1}The worst social utility in
a link stable equilibrium is at most 
\begin{eqnarray*}
 &  & |T_{A}|^{2}\left(c_{A}+A\right)+|T_{B}|^{2}\left(c_{B}+\left\lfloor 2\sqrt{c_{B}}\right\rfloor \right)\\
 & + & \left(A+1\right)\left\lfloor 2\sqrt{c}\right\rfloor |T_{A}||T_{B}|
\end{eqnarray*}
\end{lem}
\begin{proof}
The total cost due to the inter-connectivity of the type-A clique
is identical for all link stable equilibria and is $|T_{A}|\left(|T_{A}|-1\right)\left(c_{A}+A\right)$.
The maximal distance between nodes $i,j$ according to Lemma \ref{lem:The-longest-distance-2}
is $\left\lfloor 2\sqrt{c_{B}}\right\rfloor $ and therefore the maximal
cost due to the distances between type-B nodes is $\left\lfloor 2\sqrt{c_{B}}\right\rfloor |T_{B}|\left(|T_{B}|-1\right)$.
Likewise, the maximal cost due to the distance between type-B nodes
and the type-A clique is $\left(A+1\right)\left\lfloor 2\sqrt{c}\right\rfloor |T_{A}||T_{B}|$.
Finally, the maximal number of links between type-B nodes is $|T_{B}|\left(|T_{B}|-1\right)/2$
and the total cost due to this part is $|T_{B}|\left(|T_{B}|-1\right)c_{B}$.

Adding all the terms we obtain the required result.
\end{proof}
To summarize, we state the previous results in a proposition.
\begin{prop}
\label{thm:summary-of-results-1}If $c_{B}<A$ and $|T_{B}|\gg|T_{A}|\gg1$
the price of anarchy is $\Theta(c_{B})$.
\end{prop}

\subsection{Basic model - Dynamics}

The Internet is a rapidly evolving network. In fact, it may very well
be that it would never reach an equilibrium as ASs emerge, merge,
and draft new contracts among them. Therefore, a dynamic analysis
is a necessity. We first define the dynamic rules. Then, we analyze
the basin of attractions of different states, indicating which final
configurations are possible and what their likelihood is. We shall
establish that reasonable dynamics converge to \emph{just a few} equilibria.
Lastly, we investigate the speed of convergence, and show that convergence
time is \emph{linear} in the number of players.

\subsubsection{Setup \& Definitions}

At each point in time, the network is composed of a subset $N'\subset T_{A}\cup T_{B}$
of players that already joined the game. The cost function is calculated
with respect to the set of players that are present (including those
that are joining) at the considered time. The game takes place at
specific times, or \emph{turns}, where at each turn only a single
player is allowed to remove or initiate the formation of links. We
split each turn into \emph{acts}, at each of which a player either
forms or removes a single link. A player's turn is over when it has
no incentive to perform additional acts.
\begin{defn}
Dynamic Rule \#1: In player $i$'s turn it may choose to act $m\in\mathcal{N}$
times. In each act, it may remove a link $(i,j)\in E$ or, if player
$j$ agrees, it may establish the link $(i,j)$. Player $j$ would
agree to establish $(i,j)$ iff $C(j;E+(i,j))-C(j;E)<0$.
\end{defn}
The last part of the definition states that, during player's $i$
turn, all the other players will act in a greedy, rather than strategic,
manner. For example, although it may be that player $j$ prefers that
a link $(i,j')$ would be established for some $j'\neq j$, if we
adopt Dynamic Rule \#1 it will accept the establishment of the less
favorable link $(i,j).$ In other words, in a player's turn, it has
the advantage of initiation and the other players react to its offers.
This is a reasonable setting when players cannot fully predict other
players' moves and offers, due to incomplete information \citep{5173479}
such as the unknown cost structure of other players. Another scenario
that complies with this setting is when the system evolves rapidly
and players cannot estimate the condition and actions of other players.

The next two rules consider the ratio of the time scale between performing
the strategic plan and evaluation of costs. For example, can a player
remove some links, disconnect itself from the graph, and then pose
a credible threat? Or must it stay connected? Does renegotiating take
place on the same time scale as the cost evaluation or on a much shorter
one? The following rules address the two limits. 

\begin{defn}
Dynamic Rule \#2a: Let the set of links at the current act $m$ be
denoted as $E_{m}$. A link $(i,j)$ will be added if $i$ asks to
form this link and $C(j;E_{m}+ij)<C(j;E_{m})$. In addition, any link
$(i,j)$ can be removed in act $m.$\end{defn}

The alternative is as follows.
\begin{defn}
Dynamic Rule \#2b: In addition to Dynamic Rule \#2a, player $i$ would
only remove a link $(i,j)$ if $C(i;E_{m}-ij)>C(i;E_{m})$ and would
establish a link if both $C(j;E_{m}+ij)<C(j;E_{m})$ and $C(i;E_{m}+ij)<C(i;E_{m})$. 
\end{defn}
The difference between the last two dynamic rules is that, according
to Dynamic Rule \#2a, a player may perform a strategic plan in which
the first few steps will increase its cost, as long as when the plan
is completed its cost will be reduced.  On the other hand, according
to Dynamic Rule \#2b, its cost must be reduced \emph{at each act},
hence such ``grand plan'' is not possible. Note that we do not need
to discuss explicitly disconnections of several links, as these can
be done unilaterally and hence iteratively.Finally, the following lemma will be useful in the next section.
\begin{lem}
\label{lem:decay time-1}Assume $N$ players act consecutively in
a (uniformly) random order\textup{\emph{ at integer times, which we'll
denote by $t$.}} the probability $P(t)$ that a specific player did
not act $k\mathcal{\in N}$ times by $t\gg N$ decays exponentially.\end{lem}

W.l.o.g, we'll discuss player 1. Set $p=1/N.$ The probability that
a player did not act $k$ times is given by the CDF Poisson distribution
$f(t,p)$ as 
\[
P(t)=e^{-t/N}\sum_{i=0}^{k}\frac{1}{i!}\left(\frac{t}{N}\right)^{i}
\]
 and taking the limit $t\gg N$ concludes the proof.

\subsubsection{\label{sub:dynamical Results-1}Results}

After mapping the possible dynamics, we are at a position to consider
the different equilibria's basins of attraction. Specifically, we
shall establish that, in most settings, the system converges to the
optimal network, and if not, then the network's social cost is asymptotically
equal to the optimal social cost. The main reason behind this result
is the observation that a disconnected player has an immense bargaining
power, and may force its optimal choice. As the highest connected
node is usually the optimal communication partner for other nodes,
new arrivals may force links between them and this node, forming a
star-like structure. There may be few star centers in the graphs,
but as one emerges from the other, the distance between them is small,
yielding an optimal (or almost optimal) cost.

We outline the main ideas of the proof. The first few type-B players,
in the absence of a type-A player, will form a star. The star center
can be considered as a new type of player, with an intermediate importance,
as presented in Fig. \ref{fig:credible threat corr.-2}. We monitor
the network state at any turn and show that the minor players are
organized in two stars, one centered about a minor player and one
centered about a major player (Fig. \ref{fig:credible threat corr.-2}(a)).
Some cross links may be present (Fig.\,\ref{fig:credible threat corr.-2}).
By increasing its client base, the incentive of a major player to
establish a direct link with the star center is increased. This, in
turn, increases the attractiveness of the star's center in the eyes
of minor players, creating a positive feedback loop. Additional links
connecting it to all the major league players will be established,
ending up with the star's center transformation into a member of the
type-A clique. On the other hand, if the star center is not attractive
enough, then minor players may disconnect from it and establish direct
links with the type-A clique, thus reducing its importance and establishing
a negative feedback loop. The star will become empty, and the star's
center $x$ will be become a stub of a major player, like every other
type-B player. The latter is the optimal configuration, according
to proposition \ref{lem:optimal solution}. We analyze the optimal
choice of the active player, and establish that the optimal action
of a minor player depends on the number of players in each structure
and on the number of links between the major players and the minor
players' star center $x$. The latter figure depends, in turn, on
the number of players in the star. We map this to a two dimensional
dynamical system and inspect its stable points and basins of attraction
of the aforementioned configurations.

\begin{thm}
\label{cor:credible threat part 3-1}If the game obeys Dynamic Rules
\#1 and \#2a, then, in any playing order:

a) The system converges to a solution in which the total cost is at
most 

\begin{eqnarray*}
\mathcal{S} & = & |T_{A}|^{2}\left(c_{A}+A\right)-|T_{A}|\left(2A+c_{A}/2\right)+2|T_{B}|^{2}+|T_{B}|\left(A+2c_{B}\right)+3|T_{A}||T_{B}|\left(A+1\right)+2;
\end{eqnarray*}

furthermore, by taking the large network limit $|T_{B}|\gg|T_{A}|\gg1$,
we have $\mathcal{S}/\mathcal{S}_{optimal}\rightarrow1$ .

b) Convergence to the optimal stable solution occurs if either:

1) \textup{$A\cdot k_{A}>k+1$, }\textup{\emph{where $k\geq0$ is
the number of type-B nodes that first join the network, followed later
by $k_{A}$ consecutive type-A nodes (``initial condition'').}}

2)$A\cdot|T_{A}|>|T_{B}|$ (``final condition'').

c) In all of the above, if every player plays at least once in O(N)
turns, convergence occurs after o(N) steps. Otherwise, if players
play in a uniformly random order, the probability the system has not
converged by turn $t$ decays exponentially with $t$.\end{thm}
\begin{proof}
Assume $c_{A}\geq2$. Denote the first type-A player that establish
a link with a type-B player as $k$. First, we show that the network
structure is composed of a type-A (possibly empty) clique, a set of
type-B players $S$ linked to player $x$, and an additional (possibly
empty) set of type-B players $L$ connected to the type-A player $k$.
See Fig. \ref{fig:credible threat corr.-2}(a) for an illustration.
In addition, there is a set $D$ type-A nodes that are connected to
node $x$, the star center. After we establish this, we show that
the system can be mapped to a two dimensional dynamical system. Then,
we evaluate the social cost at each equilibria, and calculate the
convergence rate. We first assume $(k,x)\in E$ and discuss the case
$(k,x)\notin E$ later. 

\begin{figure}
\centering{}.\includegraphics[width=1\columnwidth]{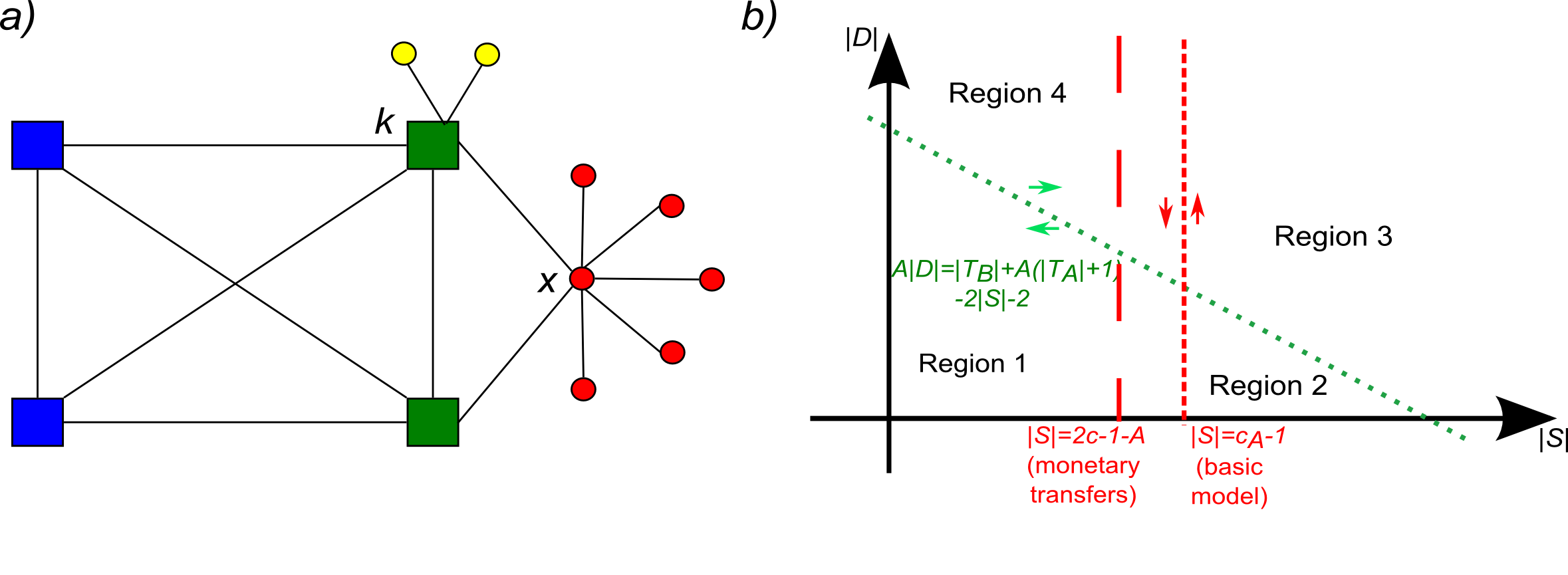}\caption{\label{fig:credible threat corr.-2}a) The network structures described
in Theorem \ref{cor:credible threat part 3}. The type-A clique contains
$|T_{A}|=4$ nodes (squares), and there are $|S|=5$ nodes in the
star (red circles). There are $|L|=2$ nodes that are connected directly
to node $k$ (yellow circles). The number of type-A nodes that are
connected to node 1, the star center, is $|D|=2$ (green squares).
b) The phase state of Theorem \ref{cor:credible threat part 3}. The
dotted green nullcline seperates the region in which $|S|$ increase
or decrease. Similiary, the dotted red line is the nullcline for the
regions in which $|D|$ increase or decrease. When monetary transfers
are forbidden allowed this nullcline is shifted, and is presented
by the dashed red line. (Proposition \ref{prop:monetary-dyanmics-1}).
\label{fig:The-phase-state-2}}
\end{figure}

We prove by induction. At turn $t=2$, after the first two players
joined the network, this is certainly true. Denote the active player
at time $t$ as $r.$ Consider the following cases:

1. $r\in T_{A}$: Since $1<c_{A}<A,$ all links to the other type-A
nodes will be established (lemma \ref{lem:optimal solution}) or maintained,
if $r$ is already connected to the network. Clearly, the optimal
link in $r$'s concern is the link with star center $x$. As $c_{B}<A$
every minor player will accept a link with a major player even if
it reduces its distance only by one. Therefore, the link $(r,x)$
is formed if the change of cost of the major player $r$,
\begin{equation}
\Delta C(r,E+rx)=c_{A}-|S|-1\label{eq: term1-3}
\end{equation}
is negative. In this case, the number of type A players connected
to the star's center, $|D|$, will increase by one. If this expression
is positive and player $r$ is connected to at least another major
player (as otherwise the graph is disconnected), the link will be
dissolved and $|D|$ will be reduced by one. It is not beneficial
for $r$ to form an additional link to any type-B player, as they
only reduce the distance from a single node by one (see the discussion
in lemma \ref{lem:optimal solution} in the appendix). 

2. $r\in T_{B},\, r\neq x$ : First, assume that $r$ is a newly arrived
player, and hence it is disconnected. Obviously, in its concern, a
link to the star's center, player $x$, is preferred over a link to
any other type-B player. Similarly, a link to a type-A player that
is linked with the star's center is preferred over a link with a player
that maintains no such link.

We claim that either $(r,k)$ or $(r,x)$ exists. Denote the number
of type-A player at turn $t$ as $m_{A}.$ The link $(r,x)$ is preferred
in $r$'s concern if the expression 

\begin{equation}
C(r,E+rk)-C(r,E+rx)=-A(1+m_{A}-|D|)+1+|S|-|L|\label{eq: term 2-2}
\end{equation}

is positive, and will be established as otherwise the network is disconnected.
If the latter expression is negative, $(r,k)$ will be formed. The
same reasoning as in case 1 shows that no additional links to a type-B
player will be formed. Otherwise, if $r$ is already connected to
the graph, than according to Dynamic Rule \#2a, $r$ may disconnect
itself, and apply its optimal policy, increasing or decreasing $|L|$
and $|S|$.

3. $r=x$, the star's center: $r$ may not remove any edge connected
to a type-B player and render the graph disconnected. On the other
hand, it has no interest in removing links to major players. On the
contrary, it will try to establish links with the major players, and
these will be formed if eq. \ref{eq: term1-3} is negative. An additional
link to a minor player connected to $k$ will only reduce the distance
to it by one and since $c_{B}>2$ player $x$ would not consider this
move worthy.

The dynamical parameters that govern the system dynamics are the number
of players in the different sets, $|S|$, $|L|$, and $|D|$. Consider
the state of the system after all the players have player once. Using
the relations $|S|+|L|+1=|T_{B}|,\: m_{A}=|T_{A}|$ we note the change
in $|S|$ depends on $|S|$ and $|L|$ while the change in $|D|$
depends only on $|S|.$ We can map this to a 2D dynamical, discrete
system with the aforementioned mapping. In Fig.\,\ref{fig:The-phase-state-2}
the state is mapped to a point in phase space $(|S|,|L|)$. The possible
states lie on a grid, and during the game the state move by an single
unit either along the $x$ or $y$ axis. There are only two stable
points, corresponding to $|S|=0,|D|=0$, which is the optimal solution
(Fig. \ref{fig:The-optimal-solution}(a)), and the state $|S|=|T_{B}|-1$
and $|D|=|T_{B}|$. 

If at a certain time expression \ref{eq: term1-3} is positive and
expression \ref{eq: term 2-2} is negative (region 3 in Fig.\,\ref{fig:The-phase-state-2}(b)),
the type-B players will prefer to connect to player $x$. This, in
turn, increases the benefit a major player gains by establishing a
link with player $x.$ The greater the set of type-A that have a direct
connection with $x$, having $|D|$ members, the more utility a direct
link with $x$ carries to a minor player. Hence, a positive feedback
loop is established. The end result is that all the players will form
a link with $x$. In particular, the type-A clique is extended to
include the type-B player $x$. Likewise, if the reverse condition
applies, a feedback loop will disconnect all links between node $x$
to the clique (except node $k$) and all type-B players will prefer
to establish a direct link with the clique. The end result in this
case is the optimal stable state. The region that is relevant to the
latter domain is region 1. 

However, there is an intermediate range of states, described by region
2 and region 4, in which the player order may dictate to which one
of the previous states the system will converge. For example, starting
from a point in region 4, if the type-A players move first, changing
the $|D|$ value, than the dynamics will lead to region 1, which converge
to the optimal solution. However, if the type-B players move first,
then the system will converge to the other equilibrium point.

We now turn to calculate the social cost at the different equilibria.
If $|D|=|T_{A}|$ and $|S|=|T'_{B}|-1$, The network topology is composed
of a $|T_{A}|$ members clique, all connected to the center $x$,
that, in turn, has $|T_{B}|-1$ stubs. The total cost in this configuration
is 
\begin{eqnarray}
S & = & |T_{A}|\left(|T_{A}|-1\right)\left(c_{A}+A\right)+2c_{B}|T_{B}|+\left(A+1\right)|T_{A}|+2\left(|T_{B}|-1\right)\nonumber \\
 &  & +2\left(|T_{B}|-1\right)\left(A+1\right)+2\left(|T_{B}|-2\right)\left(|T_{B}|-1\right)+\left(c_{B}+c_{A}\right)|T_{A}|/2\label{eq:cost at star-2}
\end{eqnarray}

where the costs are, from the left to right: the cost of the type-A
clique, the cost of the type-B star's links, the distance cost $(=1)$
between the clique and node $x$, the distance $(=1)$ cost between
the star's members and node $x$, the distance $(=2)$ cost between
the clique and the star's member, the distance $(=2)$ cost between
the star's members, and the cost due to major player link's to the
start center $x$. Adding all up, we have for the total cost
\begin{eqnarray}
\mathcal{S} & \leq & |T_{A}|\left(|T_{A}|-1\right)\left(c+A\right)+2c_{B}|T_{B}|+\left(A+1\right)\left(3|T_{A}||T_{B}|+|T_{B}|\right)+2\left(|T_{B}|-1\right)^{2}.\label{eq:social cost-1}
\end{eqnarray}

Convergence is fast, and as soon as all players have acted three times
the system will reach equilibrium. If every player plays at least
once in $o(N$) turns convergence occurs after $o(N)$ turns, otherwise
the probability the system did not reach equilibrium by time $t$
decays exponentially with $t$ according to lemma \ref{lem:decay time}
(in the appendix).

We now relax our previous assumption $c_{A}\geq2$. If $c_{A}\leq2$
and the active player $r\in T_{A}$ then it will form a link with
the star's center according to eq. \ref{eq: term1-3}. If $r\in S$
it may establish a link $(r,j)$ with a type A player, which will
later be replaced, in $j$'s turn, with the link $(j,x)$ according
to the previous discussion. In the appendix we discuss explicitly
the case where $(k,x)\notin E$ and show that in this case, additional
links may be formed, e.g., a link between one of $k'$s stubs, $i\in L$,
and the star's center $x$, as presented in Fig.\,\ref{fig:cross-tiers-2}.
These links only reduce the social cost, and do not change the dynamics,
and the system will converge to either one of the aforementioned states.
Taking the limit $T_{B}\rightarrow\infty$ and $T_{B}\in\omega\left(T_{A}\right)$
in eq. \ref{eq:social cost-1}, we get that $\mathcal{S}/\mathcal{S}_{optimal}\rightarrow1$,
and this concludes the proof.

\begin{figure}
\centering{}\includegraphics[width=0.8\columnwidth]{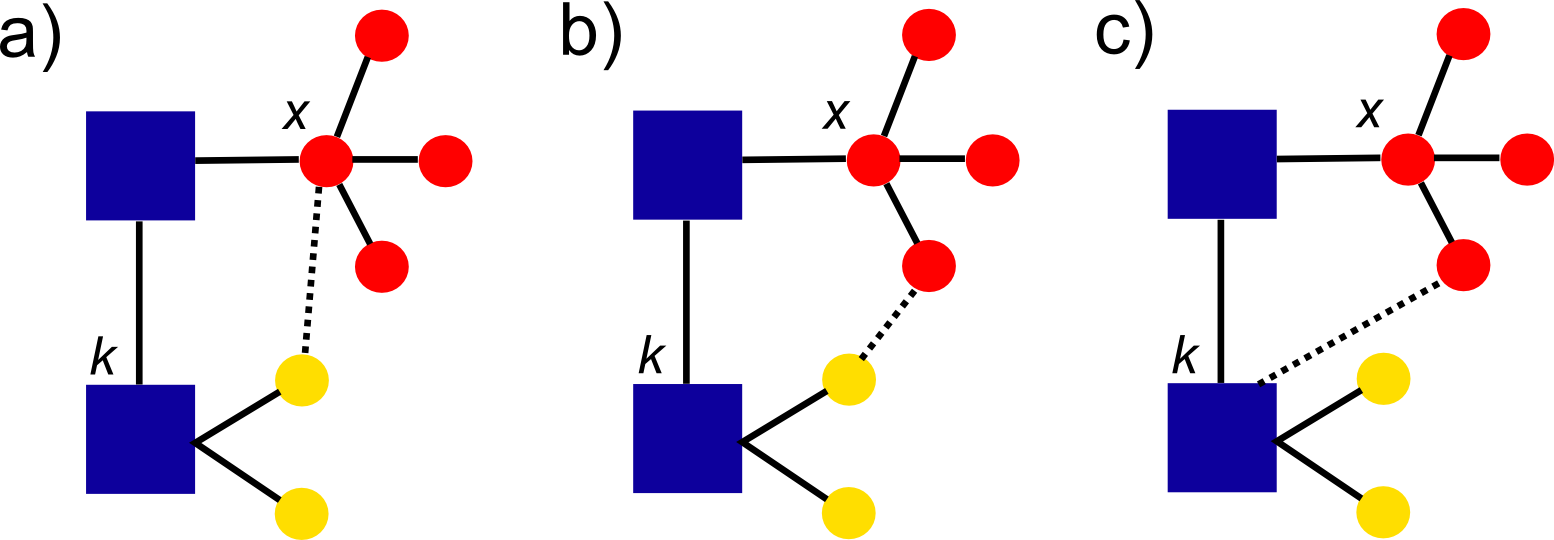}\caption{\label{fig:cross-tiers-2}Additional feasible cross-tiers links, as
described in the appendix. The star players $S$ are in red, the set
$L$ is in yellow. a) a link between the star center and $i\in L$.
b) a cross-tier link $(i,j)$ where $i\in S,j\in L$. c) a minor player
- major player link, $(i,j)$ where $i\in T_{A}$ and $j\in S.$ }
\end{figure}

We now discuss explicitly the case where, at some point, the link
$(k,x)$ is removed and assume that $c_{A}>2$, $c_{B}\geq3$. In
this case, the nullcline described by eq. \ref{eq: term 2-2} is replaced
by 
\begin{equation}
C(r,E+rk)-C(r,E+rx)=-A(1+m_{A}-|D|)+2\left(1+|S|-|L|\right).\label{eq: term 2 -kx removed-1}
\end{equation}

This changes the regions according to Fig. \ref{fig:The-dynamical-regions-1}.
Region 1, which is the basin of attraction for the optimal configuration,
increases its area, on the expense of region 4. The dynamical discussion
as described for the case $(k,x)\in E$ is still applicable, and if
the player play in a specific order, than the state vector $(|S|,|D|)$
will be in either region 1 or region 3 after $\Theta(1)$ turns. If
the players play in random order, then the system might not converge
only if player $k$ will play in every $\Theta(1)$ turns. This probability
decays exponentially, according to lemma \ref{lem:decay time-1}.

\begin{figure}
\centering{}\includegraphics[width=0.4\paperwidth]{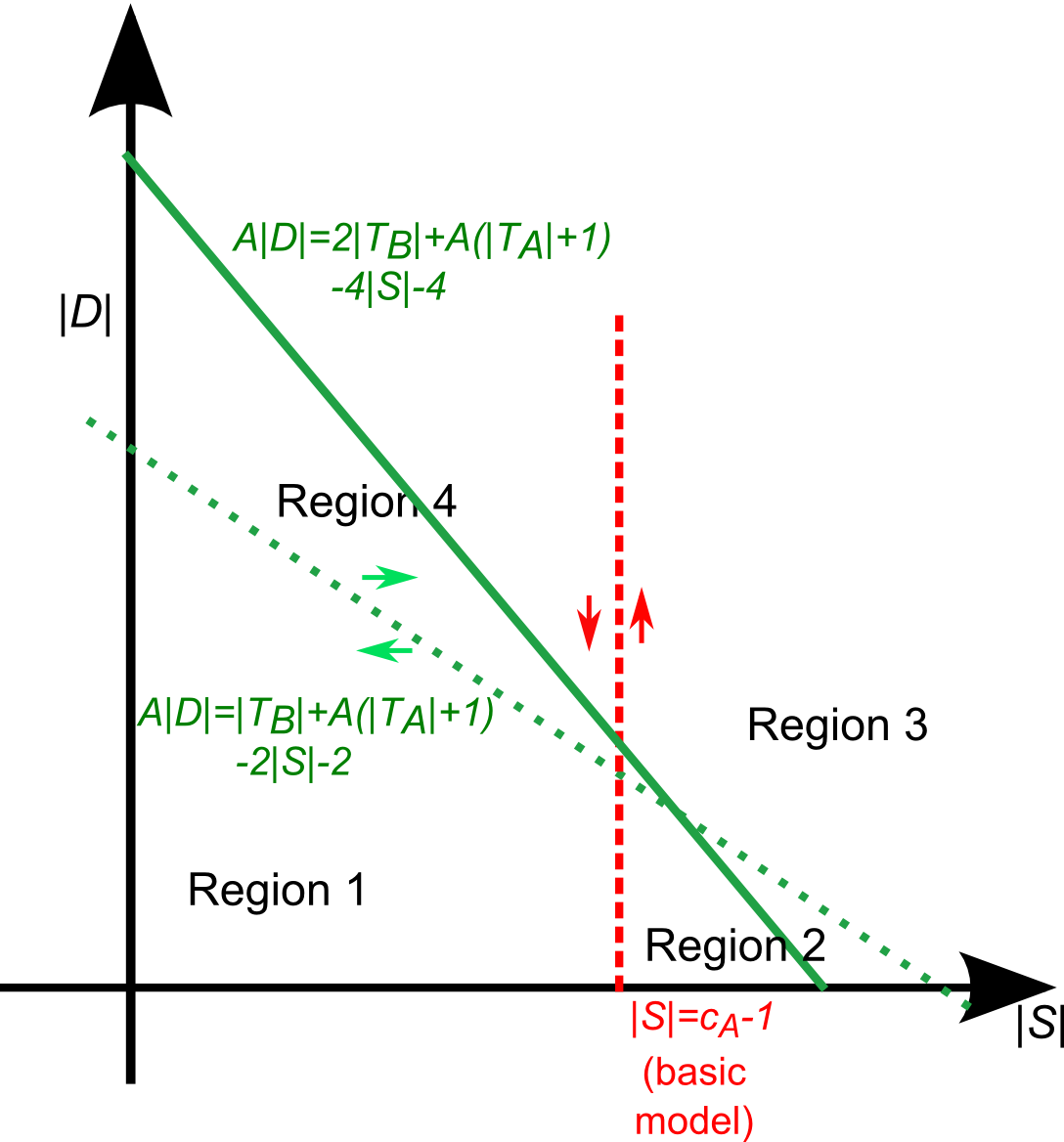}\caption{\label{fig:The-dynamical-regions-1}The dynamical regions with or
without the link $(k,x)$. The dashed green represents the new type-B
player preference nullcline when $(k,x)\in E$, according to eq. \ref{eq: term 2-2}.
According to eq. \ref{eq: term 2 -kx removed-1}, the nullcline when
$(k,x)\notin E$ is the solid green line.}
\end{figure}

In order to complete the proof, we now address the case $c_{A}\leq2$.
If $r\in S$, then $r$ will establish a link with $k\in T_{A}$,
as the distance $d(r,k)$ is reduced by two. In this case, $r$ is
a member of both $S$ and $L$, and we address this by the transformation
$|S|\leftarrow|S|$, $|L|\leftarrow|L|+1$ and $|T_{B}|\leftarrow|T_{B}|+1.$
Similarly, if $r\in L$ then it will establish links with the star
center $x$ if and only if $c_{B}\leq2$. The analogous transformation
is, $|S|\leftarrow|S|+1$, $|L|\leftarrow|L|$ and $|T_{B}|\leftarrow|T_{B}|+1.$
If $r\in T_{A}$ then it will act according to eq. \ref{eq: term1-3}.
Finally, if $r=x$ than it may both establish link with players in
$L$ and with major players according to the aforementioned discussion.

If $2\leq c_{A}\leq c_{B}$ than a link between player $k$ leaves,
$i\in L$ and a star's leaf, $j\in S$ is feasible and will be formed
when either parties are selected as the active player. Every player
$i$ may participate in only a single link of this type, as after
its establishment the maximal distance between player $i$ to every
other player is three, and an additioanl link will result in a reduction
of the sum of distances by at most two. As before, additional links
between the star center's and the major players may be formed according
to eq. \ref{eq: term1-3}. 

Consider a link $(i,j)$ between $i\in L$ and a star's leaf, $j\in S$.
Neither $x$ nor $k$ has an incentive to disconnet either $(x,j)$
or $(k,s)$ as the distance is increased by at least three. Similarly,
all the aforementioned links $(i,j)$ will not create an incentive
for a link removal $(i,j')$ or $(j,j')$ by any former partner $j'$
of the involved partied $i,j$.\end{proof}

If the star's center has a principal role in the network, then links
connecting it to all the major league players will be established,
ending up with the star's center transformation into a member of the
type-A clique. This dynamic process shows how an effectively new major
player emerges out of former type-B members in a natural way. Interestingly,
Theorem \ref{cor:credible threat part 3-1} also shows that there
exists a transient state with a better social cost than the final
state. In fact, in a certain scenario, the transient state is better
than the optimal stable state.

So far we have discussed the possibility that a player may perform
a strategic plan, implemented by Dynamic Rule \#2a. However, if we
follow Dynamic Rule \#2b instead, then a player may not disconnect
itself from the graph. The previous results indicate that it is not
worthy to add additional links to the forest of type-B nodes. Therefore,
no links will be added except for the initial ones, or, in other words,
renegotiation will always fail. The dynamics will halt as soon as
each player has acted once. Formally:

\begin{prop}
\label{prop:credible theat part 4-2}If the game obeys Dynamic Rules
\#1 and \#2b, then the system will converge to a solution in which
the total cost is at most 
\begin{eqnarray*}
\mathcal{S} & = & |T_{A}|\left(|T_{A}|-1\right)\left(c+A\right)+2\left(|T_{B}|-1\right)^{2}+\left(A+1\right)\left(3|T_{A}||T_{B}|-|T_{A}|+|T_{B}|\right)+2c|T_{B}|\,.
\end{eqnarray*}
Furthermore, for $|T_{B}|\gg|T_{A}|\gg1$, we have $\mbox{\ensuremath{\mathcal{S}/\mathcal{S}_{optimal}\leq3/2.}}$
Moreover, if every player plays at least once in O(N) turns, convergence
occurs after o(N) steps. Otherwise, e.g., if players play in a random
order, convergence occurs exponentially fast. \end{prop}
\begin{proof}
We discuss the case $c_{A}\geq2$ and $c_{B}\geq3$. The extension
for $1>c_{B}>3$ appears in the appendix. The first part of the proof
follows the same lines of the previous theorem (Theorem \ref{cor:credible threat part 3-1}).
We claim that at any given turn, the network structure is composed
of the same structures as before. Here, we discuss the scenario where
$(k,x)\in E,$ and we address the other possibility in the appendix. 

We prove by induction. Clearly, at turn one the induction assumption
is true. Note that for newly arrived players, are not affected by
either Dynamic Rules \#2a or \#2b. Hence, we only need to discuss
the change in policies of existing players. The only difference from
the dynamics described in the Theorem \ref{cor:credible threat part 3-1}
is that the a type-B players may not disconnect itself. In this case,
as the discussion there indicates the star center $x$ will refuse
a link with $i\in L$ as it only reduce $d(i,x)$ by two. Equivalently,
$k$ will refuse to establish additional links with $i\in|S|.$

In other words, as soon the first batch of type A player arrives,
all type-B players will become stagnant, either they become leaves
of either node $k$, $|L|$, or members of the star $|S|$, according
to the the sign of \ref{eq: term 2 -kx removed-1} at the time they.
The maximal distance between a type-A player and a type B player is
$2$. The maximal value of the type B - type B term is the social
cost function is when $|L|=|S|=|T_{B}|/2$. In this case, this term
contributes $3|T_{B}|^{2}$ to the social cost. Therefore, the social
cost is bounded by 
\begin{equation}
\mathcal{S}=|T_{A}|\left(|T_{A}|-1\right)\left(c_{A}+A\right)+3|T_{B}|^{2}+2c_{B}|T_{B}|+2|T_{A}||T_{B}|\left(A+1\right)\label{eq:bound for 2b-2}
\end{equation}

where we included the type-A clique's contribution to the social cost
and used $c_{B}\geq c_{A}.$ 

Assume that at some point the link $(k,x)$ was removed. In this case,
the new type-B arrival preference is changed according to eq. \ref{eq: term 2 -kx removed-1}.
Nevertheless, this change does not create an incentive for new type-B
to type-B links, and the previous conclusion holds: as soon as all
the type-B players have joined the game, they become stagnant, and
the game holds.

Consider now the case $c_{A}\leq2.$ and assume that the link $(k,x)$
exists. As in Theorem \ref{cor:credible threat part 3-1}, if the
active player is $r\in S$ it may establish a \emph{single} link $(r,j)$
with a type A player $j\in T_{A}$. As long as there is one player
$i\in S$ that is not connected to player $j\in T_{A}$, i.e., $(i,j)\notin E$,
then at player $j$'s turn, a link to the star's center $x$ will
reduced the costs of both parties. Following that, player $j$ will
remove the links to players in $S$. If, for every $i\in S$ the link
$(i,j)\in E$, then $j$ is a new star center, and every $i\in S$
will disconnect its link with $x$ at its turn. The end result is
that either for every $j\in T_{A}$ the link $(j,x)$ exists, or that
$S=\emptyset$ and the type-B nodes are leaves of various type-A nodes.
A direct calculation shows that the previous bound for the social
cost is still effective. The discussion in Theorem \ref{cor:credible threat part 3-1}
shows that if $c_{B}\leq2$, then player $x$ will not disconnect
any link to $i\in S,$ as it increases its cost by at least two. If
there exists a link $(i,j)$ with $j\in T_{B}$ and $c_{B}\geq2$
than player $x$ may disconnect the link $(i,x)$, which will only
accelerate the convergece to the aforementioned state, where $S=\emptyset$.
For every player $k\in T_{A}$ denote the set of type-B players that
have a direct link with it as $|L_{k}|$. If if $c_{B}\leq2$, there
will be additional links between $i\in L_{k}$ and $j\in L_{k'}$
for $k\neq k'.$ As before, any of the aforementioned link does not
affect the connection preference of a new type-B, which is set by
\ref{eq: term 2-2} where $|L|\leftarrow\max\{|L_{K}\}$ is the largest
set of a type-B player that is connected to $k\in T_{A}$.

If $3\geq c_{B}\geq c_{A}\geq2$ and $(k,x)\in E$, then only links
between $i\in S$ and $j\in L$ are feasible. We have shown that such
links only reduce the social cost, do not incite link removals, and
do not effect the considerations of new type-B player. 

Taking the limit $N\rightarrow\infty$ in eq. \ref{eq:bound for 2b-2}
and using $T_{A}\in\omega(1)$, $T_{B}\in\omega(T_{A})$, we obtain
$\mbox{\ensuremath{\mathcal{S}/\mathcal{S}_{optimal}\leq3/2}}$. \end{proof}

Theorem \ref{cor:credible threat part 3-1} and Proposition \ref{prop:credible theat part 4-2}
shows that the intermediate network structures of the type-B players
are not necessarily trees, and additional links among the tier two
players may exist, as found in reality. Furthermore, our model predicts
that some cross-tier links, although less likely, may be formed as
well. If Dynamic Rule \#2a is in effect, These structures are only
transient, otherwise they might remain permanent. 

The dynamical model can be easily generalized to accommodate various
constraints. Geographical constraints may limit the service providers
of the minor player. The resulting type-B structures represent different
geographical regions. Likewise, in remote locations state legislation
may regulate the Internet infrastructure. If at some point regulation
is relaxed, it can be modeled by new players that suddenly join the
game.

\subsection{Monetary transfers }

\label{sec:Monetary-transfers-1}

So far we assumed that a player cannot compensate other players for
an increase in their costs. However, contracts between different ASs
often do involve monetary transfers. Accordingly, we turn to consider
the effects of introducing such an option on the findings presented
in the previous sections. As before, we first consider the static
perspective and then turn to the dynamic perspective.

\subsubsection{Statics}

In the previous sections we showed that, if $A>c_{A}>1,$ then it
is beneficial for each type-A player to be connected to all other
type-A players. We focus on this case. 

Monetary transfers allow for a redistribution of costs. It is well
known in the game theoretic literature that, in general, this process
increases the social welfare.Indeed, the next lemma, shows that in
this setting, the maximal distance between players is smaller, compared
to Lemma \ref{lem:The-longest-distance}.
\begin{lem}
\label{lem:The-longest-distance-1-1}Allowing monetary transfers,
the longest distance between nodes i,$j\in T_{B}$ is $max\{\left\lfloor 2\sqrt{c}\right\rfloor ,1\}$.
The longest distance between nodes $i,j\in T_{A}$ is bounded by 
\[
\max\left\{ \left\lfloor 2\left(\sqrt{(A-1)^{2}+c}-(A-1)\right)\right\rfloor ,1\right\} 
\]
The longest distance between node $i\in T_{A}$ and node $j\in T_{B}$
is bounded by
\[
\max\left\{ \left\lfloor \left(\sqrt{(A-1)^{2}+4c}-(A-1)\right)\right\rfloor ,1\right\} 
\]

\end{lem}
Assume $d(i,j)=k\geq\left\lfloor 2\sqrt{c}\right\rfloor >1$ and $j\in T_{B}$.
Then there exist nodes $(x_{0}=i,x_{1,}x_{2},..x_{k}=j)$ such that
$d(i,x_{\alpha})=\alpha$. By adding a link $(i,j)$ the change in
cost of node $i$ is, according to lemma \ref{lem:The-longest-distance}
is
\begin{proof}
\begin{align*}
 & \Delta C(i,E+ij)\\
= & c-\frac{k\left(k-2\right)+mod(k,2)}{4}\\
< & c-\frac{k\left(k-2\right)}{4}\\
< & 0
\end{align*}
 Therefore, it is of the interest of player $i$ to add the link. 

Consider the case that $j\in T_{A}$ and 
\[
d(i,j)=k\geq\left\lfloor 2\left(\sqrt{(A-1)^{2}+c}-(A-1)\right)\right\rfloor >1
\]
 The change in cost after the addition of the link $(i,j)$ is 
\begin{align*}
 & c-\sum_{\alpha=1}^{k}d(i,x_{\alpha})(1+\delta_{x_{\alpha},A}(A-1))\\
 & +\sum_{\alpha=1}^{k}d'(i,x_{\alpha})(1+\delta_{x_{\alpha},A}(A-1))\\
= & c-\sum_{\alpha=1}^{k}\left(d(i,x_{\alpha})-d'(i,x_{\alpha})\right)(1+\delta_{x_{\alpha},A}(A-1))\\
< & c-\sum_{\alpha=1}^{k}d(i,x_{\alpha})+\sum_{\alpha=1}^{k}d'(i,x_{\alpha})\\
= & c-\left(1+k\right)k/2-kA+1+\left(1+k/2\right)k/2+A-1\\
= & c-k^{2}/4-k(A-1)
\end{align*}

If $i,j\in T_{B}$ then
\[
\Delta C(i,E+ij)+\Delta C(j,E+ij)<2\left(c-k^{2}/4\right)<0
\]

for $k\geq\left\lfloor 2\sqrt{c}\right\rfloor $ and the link will
be established.

Likewise, if $i,j\in T_{A}$ and 
\[
k\geq\left\lfloor 2\left(\sqrt{(A-1)^{2}+c}-(A-1)\right)\right\rfloor 
\]
 the link $(i,j)$ will be formed.

If $i\in T_{A}$, $j\in T_{B}$ and 
\[
k\geq\left\lfloor \left(\sqrt{(A-1)^{2}+4c}-(A-1)\right)\right\rfloor 
\]
 then 
\begin{eqnarray*}
\Delta C(i,E+ij)+\Delta C(j,E+ij) & < & 2c-k^{2}/2-k(A-1)\\
 & < & 0
\end{eqnarray*}

This concludes our proof.
\end{proof}
Indeed, the next proposition indicates an improvement on Proposition
\ref{lem:optimal solution}. Specifically, it shows that the optimal
network is always stabilizable, even when $\frac{A+1}{2}>c$. Without
monetary transfers, the additional links in the optimal state (Fig.
\ref{fig:The-optimal-state - monetary transfers}), connecting a major
league player with a minor league player, are unstable as the type-A
players lack any incentive to form them. By allowing monetary transfers,
the minor players can compensate the major players for the increase
in their costs. It is worthwhile to do so only if the social optimum
of the two-player game implies it. The existence or removal of an
additional link does not inflict on any other player, as the distance
between every two players is at most two.
\begin{prop}
\label{prop:optimality under monetary-1}\textup{\emph{The price of
stability is $1$.}} If $\frac{A+1}{2}\leq c\,,$ then Proposition
\ref{lem:optimal solution} holds. Furthermore, if $\frac{A+1}{2}>c$,
then the optimal stable state is such that all the type $B$ nodes
are connected to all nodes of the type-A clique. In the latter case,
the social cost of this stabilizable network is $\mathcal{S}=2|T_{B}|\left(|T_{B}|+\left(\frac{A+1}{2}+c\right)|T_{A}|\right)+|T_{A}|^{2}\left(c+A\right).$Furthermore,
if $|T_{B}|\gg1,|T_{A}|\gg1$ then, omitting linear terms in $|T_{B}|,|T_{A}|$,
$\mathcal{S}=2|T_{B}|(|T_{B}|+\left(A+c\right)|T_{A}|)+|T_{A}|^{2}\left(c+A\right).$\end{prop}
\begin{proof}
For the case $\frac{A+1}{2}\leq c$ , it was shown in Prop. \ref{lem:optimal solution}
that the optimal network is a network where all the type $B$ nodes
are connected to a specific node $j\in T_{A}$ of the type-A clique
(Fig. \ref{fig:The-optimal-state - monetary transfers}(a)) and that
this network is stabilizable. Therefore, we only need to address its
stability under monetary transfers. We apply the criteria described
in Corollary \ref{lem:edges with monetary transfers.} and show that
for every two players $i,j'$ such that $(i,j')\notin E$ we have
\[
\Delta C(i,E+ij)+\Delta C(j',E+ij')>0.
\]

If $i\in T_{B}$ then 
\[
\Delta C(i,E+ij')=c_{B}-1>0.
\]

If $i\in T_{A}$ we have that 
\[
\Delta C(i,E+ij')=c_{A}-A>0.
\]

Then, for $i\in T_{B}$ and either $j'\in T_{A}$ or $j'\in T_{B}$
we have $\Delta C(i,E+ij)+\Delta C(j',E+ij')>0$ and the link would
not be established. For every edge $(i,j)\in E$ we have that both
$\Delta C(i,E+ij)<0$ and $\Delta C(j,E+ij)<0$ (Prop. \ref{lem:optimal solution})
and therefore
\[
\Delta C(i,E+ij)+\Delta C(j,E+ij)<0.
\]

Assume $\frac{A+1}{2}>c$. It was shown in Proposition \ref{lem:optimal solution}
that the optimal network is a network where every type $B$ player
is connected to all the members of the type-A clique (Fig. \ref{fig:The-optimal-state - monetary transfers}(b)).
Under monetary transfers, this network is stabilizable, since for
for $i\in T_{B}$ , $j\in T_{A}$
\begin{eqnarray*}
 &  & \Delta C(i,E+ij)+\Delta C(j',E+ij')\\
 & = & 2c-A-1<0
\end{eqnarray*}

and the link $(i,j)$ will be formed. The previous discussion shows
that it is not beneficial to establish links between two type-B players.
Therefore, this network is stabilizable.

In conclusion, in both cases, the price of stability is 1.
\end{proof}
In the network described by Fig. \ref{fig:The-optimal-state - monetary transfers},
the minor players are connected to multiple type-A players. This emergent
behavior, where ASs have multiple uplink-downlink but very few (if
at all) cross-tier links, is found in many intermediate tiers.

Next, we show that, under mild conditions on the number of type-A
nodes, the price of anarchy is $3/2$, i.e., \emph{a fixed number}
that does not depend on any parameter value. As the number of major
players increases, the motivation to establish a direct connection
to a clique member increases, since such a link reduces the distance
to all clique members. As the incentive increases, players are willing
to pay more for this link, thus increasing, in turn, the utility of
the link in a major player's perspective. With enough major players,
all the minor players will establish direct links. Therefore, any
stable equilibrium will result in a very compact network with a diameter
of at most three. This is the main idea behind the following theorem.

\begin{thm}
\label{prop:maximal distance from clique with money-1}The maximal
distance of a type-B node from a node in the type-A clique is $\max\left\{ \left\lfloor \sqrt{\left(A|T_{A}|\right)^{2}+4cA|T_{A}|}-A|T_{A}|\right\rfloor ,2\right\} $.
Moreover, if \textup{$|T_{B}|\gg1,|T_{A}|\gg1$ }\textup{\emph{then
the price of anarchy is upper-bounded by 3/2.}}\end{thm}

\begin{proof}
Assume $d(i,j)=k>2$, where $i\in T_{B}$ and $j\in T_{A}$ but node
$j$ is not the nearest type-A node to \emph{i}. Therefore, there
exists a series of nodes $(x_{0}=i,x_{1},...,x_{k-1},x_{k}=j)$ such
that $x_{k-1}$ is a member of the type-A clique.

The change in player $j$'s cost by establishing $(i,j)$ is 
\begin{eqnarray*}
\Delta C(j,E+ij) & = & c_{A}-\sum_{\alpha=1}^{k}d(j,x_{\alpha})(1+\delta_{x_{\alpha},A}(A-1))\\
 &  & +\sum_{\alpha=1}^{k}d'(j,x_{\alpha})(1+\delta_{x_{\alpha},A}(A-1))\\
 & < & c_{A}-\sum_{\alpha=1}^{k}\left(d(j,x_{\alpha})-d'(j,x_{\alpha})\right)\\
 & < & c_{A}-k^{2}/4.
\end{eqnarray*}

The corresponding change in player $i$'s cost is
\[
\Delta C(i,E+ij)<c_{B}-k^{2}/4-(2k-4)A-\left(k-2\right)A\cdot\left(|T_{A}|-2\right).
\]

The first term is the link cost, the second and third terms are due
to change of distance from players $x_{k-1},x_{k}$ and the last term
express the change of distance form the rest of the type-A clique.
As
\[
\Delta C(i,E+ij)<c_{B}-k^{2}/4-\left(k-2\right)A\cdot|T_{A}|
\]

the total change in cost is 
\begin{eqnarray*}
 &  & \Delta C(j,E+ij)+\Delta C(i,E+ij)\\
 & < & 2c-k^{2}/2-\left(k-2\right)A\cdot|T_{A}|\\
 & < & 0
\end{eqnarray*}

for 
\[
k<\sqrt{\left(A|T_{A}|\right)^{2}+4cA|T_{A}|}-A|T_{A}|
\]

Note that as the number of member in the type-A clique, $\left|T_{A}\right|$,
increases, the right expression goes to 0, in contradiction to our
initial assumption. Therefore, in the large network limit the maximal
distance of a type-B node from a node in the type-A clique is 2. In
this case, the maximal distance between two type-A nodes is 1 (as
before), between type-A and type-B nodes is 2 and between two type-B
nodes is 3. The maximal social cost in an equilibrium is
\begin{eqnarray*}
\mathcal{S} & < & 3|T_{B}|\left(|T_{B}|-1\right)+|T_{B}|c_{B}\\
 &  & +2|T_{B}||T_{A}|\left(A+1\right)+|T_{A}|\left(|T_{A}|-1\right)\left(c_{A}+A\right)
\end{eqnarray*}

For $|T_{B}|\gg1,|T_{A}|\gg1$ we have 
\[
S=3|T_{B}|^{2}+2|T_{B}||T_{A}|\left(A+1\right)+|T_{A}|^{2}\left(c_{A}+A\right).
\]

comparing this with the optimal cost in this limit 
\[
S_{opt}=2|T_{B}|^{2}+2|T_{B}||T_{A}|\left(A+1\right)+|T_{A}|^{2}\left(c_{A}+A\right)
\]

we obtain the required result.
\end{proof}
This theorem shows that by allowing monetary transfers, the maximal
distance of a type-B player to the type-A clique depends inversely
on the number of nodes in the clique and the number of players in
general. The number of ASs increases in time, and we may assume the
number of type-A players follows. Therefore, we expect a decrease
of the mean ``node-core distance'' in time. Our data analysis, which
appears in the appendix, indicates that this real-world distance indeed
decreases in time.

\subsubsection{Dynamics\label{sub: monetary Dynamics-1}}

We now consider the dynamic process of network formation under the
presence of monetary transfers. For every node $i$ there may be several
nodes, indexed by \emph{j, }such that $\Delta C(j,ij)+\Delta C(i,ij)<0,$
and player \emph{i }needs to decide on the order of players with which
it will ask to form links. We point out that the order of establishing
links is potentially important. The order by which player player $i$
will establish links depends on the pricing mechanism. There are several
alternatives and, correspondingly, several possible ways to specify
player \emph{i's }preferences\emph{, }each leading to a different
dynamic rule. 

Perhaps the most naive assumption is that if for player $j,$ $\Delta C(j,ij)>0$,
then the price it will ask player $i$ to pay is $P_{ij}=\max\{\Delta C(j,ij),0\}.$
In other words, if it is beneficial for player $j$ to establish a
link, it will not ask for a payment in order to do so. Otherwise,
it will demand the minimal price that compensates for the increase
in its costs. This dynamic rule represents an efficient market. This
suggests the following preference order rule.

\begin{defn}
Preference Order \#1: Player $i$ will establish a link with a player
$j$ such that $\Delta C(i,ij)+min\{\Delta C(j,ij),0\}$ is minimal.
The price player $i$ will pay is $P_{ij}=max\{\Delta C(j,ij),0\}$.\end{defn}

As established by the next theorem, Preference Order \#1 leads to
the optimal equilibrium fast. In essence, if the clique is large enough,
then it is worthy for type-B players to establish a direct link to
the clique, compensating a type-A player, and follow this move by
disconnecting from the star. Therefore, monetary transfers increase
the fluidity of the system, enabling players to escape from an unfortunate
position. Hence, we obtain an improved version of Theorem \ref{cor:credible threat part 3-1}.
\begin{thm}
Assume the players follow Preference Order \#1 and Dynamic Rule \#1,
and either Dynamic Rule \#2a or \#2b. If $\frac{A+1}{2}>c$, then
the system converges to the optimal solution. If every player plays
at least once in O(N) turns, convergence occurs after o(N) steps.
Otherwise, e.g., if players play in a random order, convergence occurs
exponentially fast. \end{thm}
\begin{proof}
Assume it is player $i$'s turn. For every player $j$ such that $(i,j)\notin E$,
we have that $d(i,j)\geq2$. By establishing the link $(i,j)$ the
distance is reduced to one and 
\begin{eqnarray*}
 &  & \Delta C(j,E+ij)+\Delta C(i,E+ij)\\
 & \leq & 2c+\left(1-d(i,j)\right)\left(2+\delta_{i,A}(A-1)+\delta_{j,A}(A-1)\right).
\end{eqnarray*}

This expression is negative if either $i\in T_{A}$ or $j\in T_{A}$,
as 
\[
2c-A-1<0
\]
Therefore, if player $i\in T_{A}$ it will form links all other players,
whereas if $i\in T_{B}$ it will form links with all the type-A players.
Finally, after every player has played twice, every type-B player
has established links to all members of the type-B clique. Therefore,
the distance between every two type-B players is at most two. Consider
two type-B players, $i,j\in T_{B}$ for which the link $(i,j)$ exists.
If the link is removed, the distance will grow from one to two, per
the previous discussion. But,

\begin{eqnarray*}
 &  & \Delta C(j,E+ij)+\Delta C(i,E+ij)\\
 & = & 2c+2\left(1-2\right)\\
 & > & 0
\end{eqnarray*}
 Hence, this link will be dissolved. This process will be completed
as soon as every type-B player has played at least three times. If
every player plays at least once in $o(N$) turns convergence occurs
after $o(N)$ turns, otherwise the probability the system did not
reach equilibrium by time $t$ decays exponentially with $t$ according
to lemma \ref{lem:decay time-1}. 

The resulting network structure is composed of a type-A clique, and
every type-B player is connected to all members of the type-A clique
(Fig.\ref{fig:The-optimal-solution}(b)). As discussed in Prop. \ref{prop:optimality under monetary-1},
this structure is optimal and stable.
\end{proof}
Yet, the common wisdom that monetary transfers, or utility transfers
in general, should increase the social welfare, is contradicted in
our setting by the following proposition. Specifically, there are
certain instances, where allowing monetary transfers yields a decrease
in the social utility. In other words, if monetary transfers are allowed,
then the system may converge to a sub-optimal state.

\begin{prop}
Assume $\frac{A+1}{2}\leq c$. Consider the case where monetary transfers
are allowed and the game obeys Dynamic Rules \#1,\#2a and Preference
Order \#1. Then:

a) The system will either converge to the optimal solution or to a
solution in which the social cost is 
\begin{eqnarray*}
\mathcal{S} & = & |T_{A}|\left(|T_{A}|-1\right)\left(c_{A}+A\right)+2\left(|T_{B}|-1\right)^{2}+\left(A+1\right)\left(3|T_{A}||T_{B}|-|T_{A}|+|T_{B}|\right)+2c|T_{B}|.
\end{eqnarray*}
For $|T_{B}|\rightarrow\infty$, $|T_{B}|\in\omega\left(|T_{A}|\right)$
we have $S/S_{optimal}\rightarrow1\,$. In addition, if one of the
first $\left\lfloor c-1\right\rfloor $ nodes to attach to the network
is of type-A then the system converges to the optimal solution. 

b) For some parameters and playing orders, the system converges to
the optimal state if monetary transfers are forbidden, but when transfers
are allowed it fails to do so. This is the case, for example, when
the first $k$ players are of type-B, and $2c-A-1<k<c-1$.\end{prop}

\begin{proof}
a) We claim that, at any given turn $t$, the network is composed
of the same structures as in Theorem \ref{cor:credible threat part 3-1}.
We use the notation described there. See Fig. \ref{fig:credible threat corr.-2}
for an illustration. First, assume that the link $(k,x)$ exists.

We prove by induction. At turn $t=1$ the induction hypothesis is
true. We'll discuss the different configurations at time $t$. 

1. $r\in T_{A}$: As before, all links to the other type-A nodes will
be established or maintained, if $r$ is already connected to the
network. The link $(r,x)$ will be formed if the change of cost of
player $r$,
\begin{equation}
\Delta C(j,E+ij)+\Delta C(i,E+ij)=2c-A-|S|-1\label{eq: term1-1-1-1}
\end{equation}
is negative. In this case $|S|$ will increase by one. If this expression
is positive and $(r,x)\in E,$ the link will be dissolved and $|D|$
will be reduced. It is not beneficial for $r$ to form an additional
link to any type-B player, as they only reduce the distance from a
single node by 1 and $\frac{A+1}{2}\leq c$. 

2. $r\in T_{B}$ :The discussion in Theorem \ref{cor:credible threat part 3-1}
shows that a newly arrived may choose to establish its optimal link,
which would be either $(r,k)$ or $(r,x)$ according to the sign of
expression \ref{eq: term 2-2}. As otherwise the graph is disconnected,
such link will cost nothing. Similarly, if $r$ is already connected,
it may disconnect itself as an intermediate state and use its improved
bargaining point to impose its optimal choice. Hence, the formation
of either $(r,k)$ or $(r,x)$ is not affected by the inclusion of
monetary transfers to the basic model. Assume the optimal move for
$r$ is to be a member of the star, $r\in S$. If 
\begin{eqnarray}
\Delta C(k,E+kr)+\Delta C(r,E+kr) & = & 2c-A|m_{A}|-1-|L|\label{eq:monetary, S-1}
\end{eqnarray}
 is negative, than this link will be formed. In this case, $r$ is
a member of both $S$ and $L$, and we address this by the transformation
$|S|\leftarrow|S|$, $|L|\leftarrow|L|+1$ and $|T_{B}|\leftarrow|T_{B}|+1.$
Similarly, if $r\in L$ than it will establish links with the star
center $x$ if and only if 
\begin{eqnarray}
2c & < & |S|+1.\label{eq:monetary, L-1}
\end{eqnarray}
The analogous transformation is, $|S|\leftarrow|S|+1$, $|L|\leftarrow|L|$
and $|T_{B}|\leftarrow|T_{B}|+1.$ The is also true if $r=x$ and
the latter condition is satisfied. We have shown in Theorem \ref{cor:credible threat part 3-1}
, that such links only reduce the social cost, do not incite link
removals, and do not effect the considerations of new type-B player. 

Consider the case that at some point the link $(k,x)$ was removed.
The new player preference nullcline is described by eq. \ref{eq: term 2 -kx removed-1}.
Now, if $r\in S$, it has an increased incentive to establish a link
with $k$, as the nodes in $L$ are farther away from it. In this
case, the condition to establish the link $(k,r)$ is 
\begin{eqnarray*}
\Delta C(k,E+kr)+\Delta C(r,E+kr) & = & 2c-A|m_{A}|-1-2|L|<0
\end{eqnarray*}
(compared to eq. \ref{eq:monetary, S-1}) . Similarly, if $r\in L$
the criteria for establishing the link $(r,x)$ is $2c<2|S|+1$ (compared
to eq. \ref{eq:monetary, L-1}). The transformation described above
should be applied in either case.

As before, as soon as all the players have played two time, the system
will be in either region 1 or region 3, and from there convergence
occurs after every player has played once more.

b) If dynamic rule \#2a is in effect, the nullcline represented by
eq. \ref{eq: term1-1-1-1} is shifted to the left compared to the
nullcline of eq. \ref{eq: term1-3}, increasing region 3 and region
2 on the expanse of region 1 and region 4. Therefore, there are cases
where the system would have converge to the optimal state, but allowing
monetary transfers it would converge to the other stable state. Intuitively,
the star center may pay type-A players to establish links with her,
reducing the motivation for one of her leafs to defect and in turn,
increasing the incentive of the other players to directly connect
to it. Hence, monetary transfers reduce the threshold for the positive
feedback loop discussing in Theorem \ref{cor:credible threat part 3-1}.\end{proof}

The latter proposition shows that the emergence of an effectively
new major league player, namely the star center, occurs more frequently
with monetary transfers, although the social cost is hindered. 

A more elaborate choice of a price mechanism is that of ``strategic''
pricing. Specifically, consider a player $j^{*}$ that knows that
the link $(i,j^{*})$ carries the least utility for player $i$. It
is reasonable to assume that player $j$ will ask the minimal price
for it, as long as it is greater than its implied costs. We will denote
this price as $P_{ij^{*}}$. Every other player $x$ will use this
value and demand an additional payment from player $i$, as the link
$(i,x)$ is more beneficial for player $i$. Formally,

\begin{defn}
Pricing mechanism \#1: Set $j^{*}$ as the node which maximize $\Delta C(i,E+ij*)$.
Set $P_{ij^{*}}=max\{-\Delta C(j*,E,ij*),0\}$. Finally, set
\[
\alpha_{ij}=\Delta C(i,E+ij)-\left(\Delta C(i,E+ij^{*})+P_{ij^{*}}\right)
\]

The price that player $j$ will require in order to establish \emph{$(i,j)$
}is\emph{ }
\[
P_{ij}=max\{0,\alpha_{ij},-\Delta C(j,E+ij)\}
\]
\end{defn}

As far as player $i$ is concerned, all the links $(i,j)$ with $P_{ij}=\alpha_{ij}$
carry the same utility, and this utility is greater than the utility
of links for which the former condition is not valid. Some of these
links have a better connection value, but they come at a higher price.
Since all the links carry the same utility, we need to decide on some
preference mechanism for player $i$. The simplest one is the ``cheap''
choice, by which we mean that, if there are a few equivalent links,
then the player will choose the cheapest one. This can be reasoned
by the assumption that a new player cannot spend too much resources,
and therefore it will choose the ``cheapest'' option that belongs
to the set of links with maximal utility.

\begin{defn}
Preference order \#2: Player $i$ will establish links with player
$j$ if player $j$ minimizes $\Delta\tilde{C}(i,ij)=\Delta C(i,ij)+P_{ij}$
and $\Delta\tilde{C}(i,ij)<0$.

If there are several players that minimize $\Delta\tilde{C}(i,ij)$,
then player $i$ will establish a link with a player that minimizes
$P_{ij}$. If there are several players that satisfy the previous
condition, then one out of them is chosen randomly.\end{defn}
Note that low-cost links have a poor ``connection value'' and therefore
the previous statement can also be formulated as a preference for
links with low connection value. 

We proceed to consider the dynamic aspects of the system under such
conditions.

An immediate result of this definition is the following.
\begin{lem}
If node $j^{*}$ satisfies
\[
\Delta C(j*,ij*)<0
\]
then the link $(i,j^{*})$ will be formed. If there are few nodes
that satisfy this criterion, a link connecting $i$ and one of this
node will be picked at random.\end{lem}
\begin{proof}
As 
\[
P_{ij^{*}}=max\{-\Delta C(j*,E,ij*),0\}=0
\]
 $\Delta\tilde{C}(i,ij)$ is maximal when $\Delta C(i,ij)$ is, which
is for node $j^{*}$.
\end{proof}
The resulting equilibria following this preference order are very
diverse and depend heavily on the order of acting players. The only
general statement that can are of the form of Lemma \ref{lem:The-longest-distance-1-1}.
Before we elaborate, let us provide another useful lemma.
\begin{lem}
\label{lem:shortcut-1}Assume that according to the preference order
player $i$ will establish the link $(i,j^{*})$. If there is a node
$x$ such that
\begin{eqnarray*}
\Delta C(i,(E+ij*)+ix) & < & 0\\
\Delta C(x,(E+ij*)+ix) & < & 0\\
\Delta C(i,(E+ix)+ij^{*})+\Delta C(j^{*},(E+ix)+ij^{*}) & > & 0
\end{eqnarray*}

Then effectively the link $(i,x)$ will be formed instead of $(i,j^{*})$.
\end{lem}
The first two inequalities state that after establishing the link
$(i,j^{*})$ the link $(i,x)$ will be formed as well. However, the
last inequality indicates that after the former step, it is worthy
for player $i$ to disconnect the link $(i,j^{*})$. 

We proceed to consider the dynamic aspects of the system under such
conditions.

\begin{prop}
\label{prop:monetary-dyanmics-1-2}Assume that:

A) Players follow Preference Order \#2 and Dynamic Rule \#1, and either
Dynamic Rule \#2a or \#2b. 

B) There are enough players such that $2c<T_{A}\cdot A+T_{B}^{2}/4$.

C) At least one out of the first $m$ players is of type-A, where
$m$ satisfies $m\geq\sqrt{A^{2}+4c-1}-A$.

Then, if the players play in a non-random order, the system converges
to a state where all the type-B nodes are connected directly to the
type-A clique, except perhaps lines of nodes with summed maximal length
of $m$. In the large network limit, $\mathcal{S}/\mathcal{S}_{optimal}<3/2+c$
.

D) If $2c>(A-1)+|T_{B}|/|T_{A}|$ then the bound in (C) can be tightened
to $\mathcal{S}/\mathcal{S}_{optimal}<3/2$.\end{prop}

\begin{proof}
We prove by induction. Assume that the first type-B player is player
$k_{A}$ and the first type-A player is $k_{B}$. We first prove that
the structure up-to the first move of $\max\{k_{A},k_{B}\}$ is a
type-A clique, and an additional line of maximal length $\left|k_{B}-k_{A}\right|$
of type-B players connected to a single type-A player.

If $k_{A}>k_{B}=1$, there is a set of type-B players that play before
the first type-A joins the game, we claim they form a line. For the
first player it is true. Consider player $x$. The least utility it
may obtain is by establishing a link to a node at the end of the line,
and therefore it will connect first to one of the ends, as this would
be the cheapest link. W.l.o.g, assume that it connects to $x-1$.
The most beneficial additional link it may establish is $(1,x)$ but
according to Corollary \ref{lem:shortcut benefit}
\begin{eqnarray*}
 &  & \Delta C(x,E+1x)+\Delta C(1,E+1x)\\
 & = & 2\Delta C(1,E+1x)\\
 & = & 2c--\frac{x\left(x-2\right)+mod(x,2)}{2}\\
 & > & 2c-\frac{m\left(m-2\right)+1}{2}\\
 & > & 0
\end{eqnarray*}

and therefore no additional link will be formed. Player $k_{A}$ will
establish a link with a node at one of the line's ends, say to player
$k_{A}-1$, and as 
\begin{eqnarray*}
 &  & \Delta C(k_{A},E+1k_{A})+\Delta C(1,E+1k_{A})\\
 & > & 2c-\frac{m\left(m-2\right)+1}{2}-(m-2)(A-1)
\end{eqnarray*}

there would be no additional links between player $k_{A}$ and a member
of the line.

If $k_{A}<k_{B}=1$ then the first $k_{B}-1$ type-A player will form
a clique (same reasoning as in lemma \ref{prop:optimality under monetary-1}),
and player $k_{B}$ (of type-B) will form a link to one of the type-A
players randomly. This completes the induction proof.

For every new player, the link with the least utility is the link
connecting the new arrival and the end of the type-B strand. For a
type-A player, using lemma \ref{lem:shortcut-1}, immediately after
establishing this link it will be dissolved and the new player will
join the clique. Type B players will attach to the end of the line
and the line's length will grow. Assume that when player $x$ turn
to play there are $m_{x}$ type-A players in the clique. By establishing
the link $(i,x)$ we have (Fig. \ref{fig:suggested shortcut-1}),
\begin{eqnarray*}
\Delta C(x,E+ix) & < & c_{B}-m_{X}\cdot A\\
\Delta C(i,E+ix) & = & c_{A}-\frac{\left(x+1\right)\left(x-1\right)+mod(x+1,2)}{2}.
\end{eqnarray*}

For $x\geq m\geq\sqrt{A^{2}+4c-1}-A$
\[
\Delta C(x,ix)+\Delta C(i,ix)<2c-m_{X}\cdot A-x^{2}/4<0
\]

and therefore according to lemma \ref{lem:shortcut-1} player $x$
will connect directly to one of the nodes in the clique instead of
line\textquoteright{}s end. After all the players have played once,
the structure formed is a type-A clique, a line of maximum length
$m$, and type-B nodes that are connected directly to at least one
of the members of the clique. In a large network, the line's maximal
length is $o(\sqrt{T_{B}})$ and is negligible in comparison to terms
that are $o(T_{B})$. 

\begin{figure}
\centering{}\includegraphics[width=0.7\columnwidth]{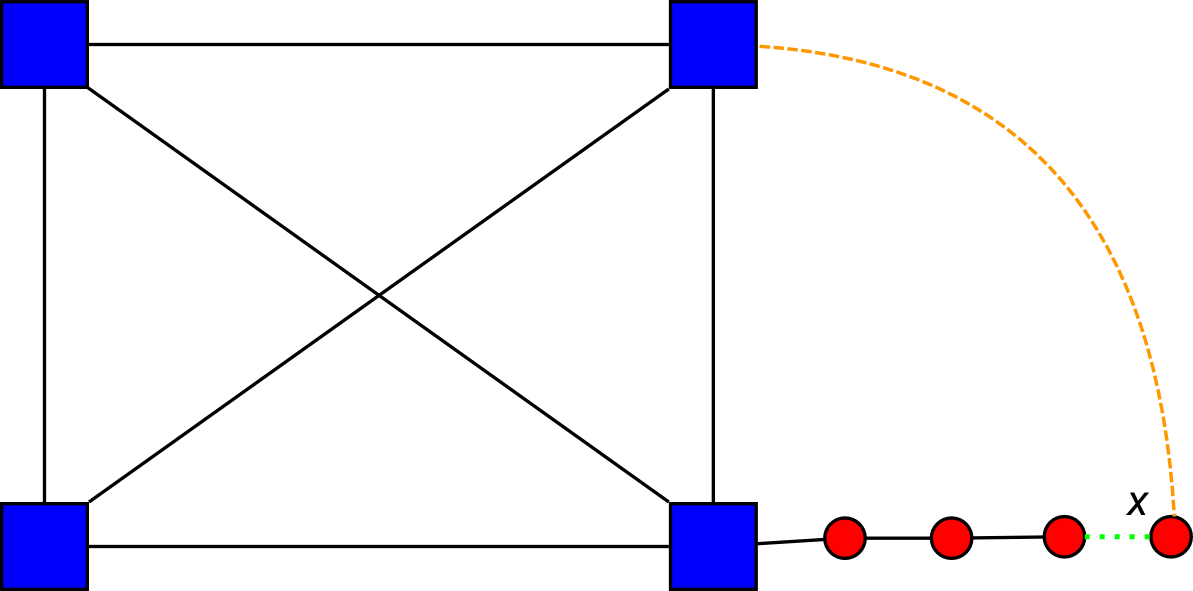}\caption{\label{fig:suggested shortcut-1}The suggested dynamics in Prop. \ref{prop:monetary-dyanmics-1-2}.
Here, $x=4$ and $m_{X}=4$. The link to the clique is dashed in orange,
the canceled link is in dotted green.}
\end{figure}

The only possible deviation in this sub-graph is establishing additional
links between a type-B player and clique members (other than the one
it is currently linked to). This will only reduce the overall distance.
Hence we can asymptotically bound the social cost by 
\begin{eqnarray*}
\mathcal{S} & < & 3|T_{B}|^{2}+2|T_{B}||T_{A}|\left(A+1\right)\\
 &  & +2|T_{B}||T_{A}|c+|T_{A}|^{2}\left(c_{A}+A\right)
\end{eqnarray*}

where the first term expresses the (maximal) distance between type-B
nodes, the second the distance cost between the type-A clique and
the type-B nodes, the third is the cost of links between type-A players
and type-B players and the third is the cost of the type-A clique.
Comparing this to the optimal solution (Proposition \ref{prop:maximal distance from clique with money-1})
under the assumption that $T_{B}>T_{A}$, we have
\[
\mathcal{S}/\mathcal{S}_{optimal}<\frac{3}{2}+c.
\]

This concludes the proof.

D) We can improve on the former bound by noting that according to
Preference Order \#2, when disconnecting from the long line a player
will reconnect to the type-A player that has the least utility in
his concern (and hence requires the lowest price). In other words,
it will connect to one of the nodes that carry the least amount of
type-B nodes at that moment. Therefore, to each type-A node roughly
$|T_{B}|/|T_{A}|$ type-B nodes will be connected. This allows us
to provide the following corollary. 

According to the aforementioned discussion by establishing a link
between node $j\in T_{B}$ and $i\in T_{A}$ (where $j$ in not connected
directly to $i$) the change of cost is 
\begin{eqnarray*}
\Delta C(i,ij) & = & c_{A}-A-\frac{|T_{B}|}{|T_{A}|}\\
\Delta C(j,ij) & = & c_{B}-1
\end{eqnarray*}

but
\[
\Delta C(i,ij)+\Delta C(j,ij)=2c-(A-1)-\frac{|T_{B}|}{|T_{A}|}>0
\]

and the link would not be establish. Hence, we can neglect the term
$2|T_{B}||T_{A}|c$ and 
\[
\sum C(i)<3|T_{B}|^{2}+2|T_{B}||T_{A}|\left(A+1\right)+|T_{A}|^{2}\left(c_{A}+A\right)
\]

By comparing with the optimal solution we obtained the required result.
\end{proof}
\begin{figure}
\centering{}\includegraphics[width=0.7\columnwidth]{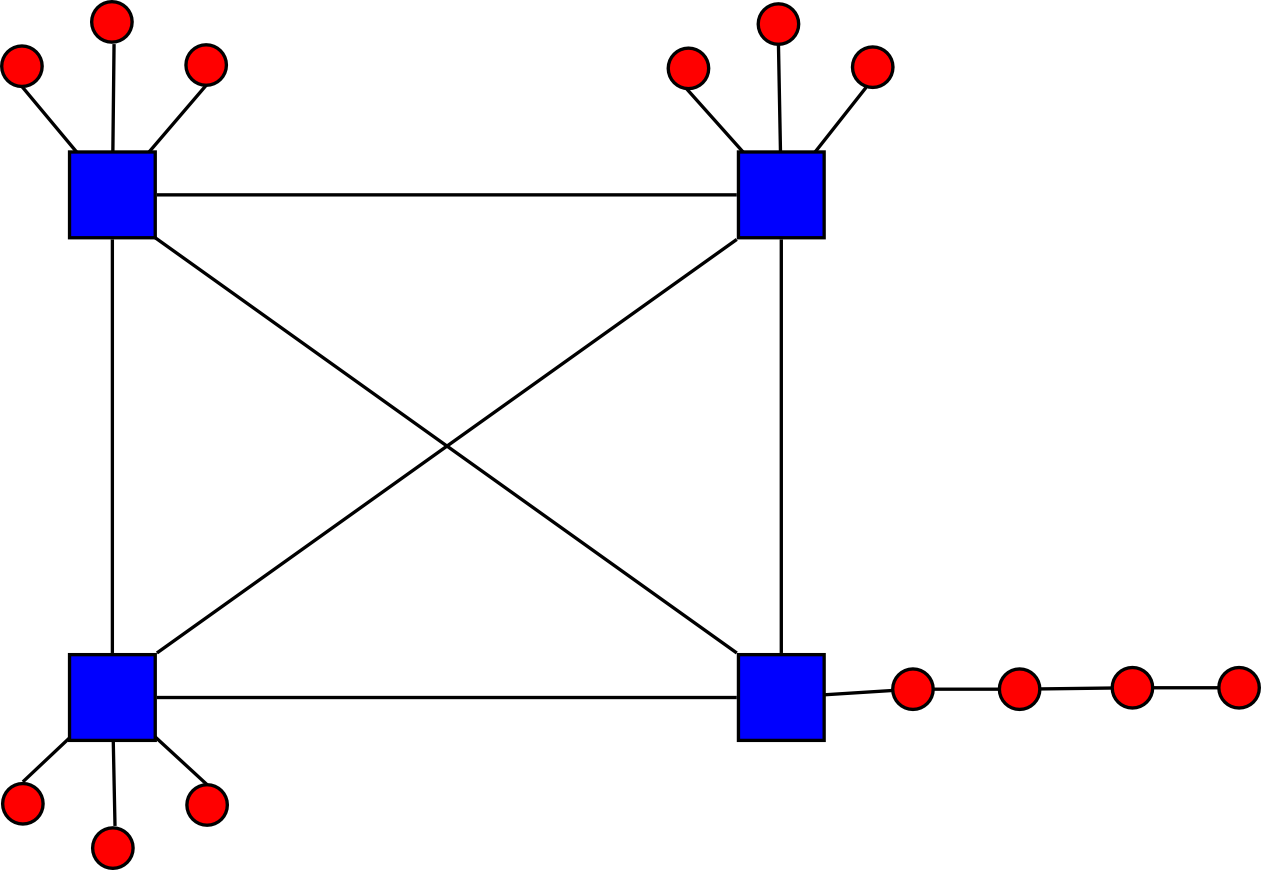}\caption{The network structure described in Prop. \ref{prop:monetary-dyanmics-1-2}.
The type-B players connect almost uniformly to the different members
of the type-A clique. In addition, there is one line of type-B players. }
\end{figure}

In order to obtain the result in Proposition \ref{prop:maximal distance from clique with money-1},
we had to assume a large limit for the number of type-A players. Here,
on the other hand, we were able to obtain a similar result yet without
that assumption, i.e., solely by dynamic considerations. 

It is important to note that, although our model allows for monetary
transfers, in \emph{every} resulting agreement between major players
no monetary transaction is performed. In other words, our model predicts
that the major players clique will form a \emph{settlement-free} interconnection
subgraph, while in major player - minor player contracts transactions
will occur, and they will be of a transit contract type. Indeed, this
observation is well founded in reality.

\subsection{Detailed Calculations}

\subsubsection*{In Proposition \ref{lem:optimal solution}:}

For the network presented in Fig. \ref{fig:The-optimal-solution}(a),
the number of links in the type-A clique is the same as the number
of pairs, which is ${|T_{A}| \choose 2}$. Each edge is counted twice
at the social cost summation (as part of its members cost) and therefore
the social cost due to the number of edges in the type-A clique is
\[
2c{|T_{A}| \choose 2}=c|T_{A}|\left(|T_{A}|-1\right).
\]

Likewise, the distance of each player in the type-A clique from every
other type-A player is one, and the distance cost is counted twice,
and we obtain the type-A clique distance cost term $A|T_{A}|\left(|T_{A}|-1\right)$.

As the distance between every two type-B players is 2, and there are
${|T_{B}| \choose 2}$ pairs, we obtain in a similar fashion the type-B
to type-B distance cost term 
\[
2\cdot2{|T_{B}| \choose 2}=2|T_{B}|\left(|T_{B}|-1\right).
\]

There are $|T_{B}|$ type-B players and their distance to $|T_{A}|-1$
clique members is 2, and their distance to a single clique member
is one. The distance is multiplied by $A$ when the sum is over a
type-B player, and it is multiplied by $1$ when the sum is over a
type-A player. Therefore, the social cost contribution of the type-A
to type-B distance term is 
\[
2\left(A+1\right)|T_{B}|\left(|T_{A}|-1\right)+\left(A+1\right)|T_{B}|.
\]

Finally, there are $|T_{B}|$ links connecting every type-B player
to the clique, and each edge is summed twice, and we obtain the type-B
players link's cost

\[
2c|T_{B}|.
\]

If, however, one observes the network presented in Fig. \ref{fig:The-optimal-solution}(b),
then every type-B players is connected by $|T_{A}|$ edges to the
type-A clique. Therefore, the last two expressions are replaced,

\[
\left(A+1\right)|T_{B}||T_{A}|+2c|T_{B}||T_{A}|
\]

where the right term is the distance term and the second is the link's
cost term, and we have taken into account the double summation.

\subsubsection*{In Lemma \ref{lem:poor eq}:}

For simplicity, assume $k$ is odd. Without the edge $(x_{k,}y_{\left(k+1\right)/2})$,
the sum of distances from player $x_{k}$ to all players along the
$(y_{1},y_{2},...y_{k})$ is 
\[
\sum_{i=k+1}^{2k}i=\frac{k\left(3k+1\right)}{2}
\]

after the addition of the link $(x_{k,}y_{\left\lfloor \left(k+1\right)/2\right\rfloor })$,
the the sum of distances is 
\[
2\sum_{i=1}^{\left(k+1\right)/2}i-1=\frac{k^{2}}{4}+k-\frac{1}{4}
\]

and therefore the reduction is the sum of distances bounded from above
by $5k^{2}/4$.
\end{document}